\documentclass{article}

\usepackage[preprint]{neurips_2021}

\usepackage[utf8]{inputenc} 
\usepackage[T1]{fontenc}    
\usepackage{hyperref}       
\usepackage{url}            
\usepackage{booktabs}       
\usepackage{amsfonts}       
\usepackage{nicefrac}       
\usepackage{microtype}      
\usepackage{xcolor}         
\usepackage{wrapfig}
\usepackage{bm}
\usepackage{amsmath}
\usepackage{amssymb}
\usepackage{amsthm}
\usepackage{amsmath}
\usepackage[inline]{enumitem}
\usepackage{algorithm}
\usepackage[noend]{algpseudocode}
\usepackage{graphicx}
\usepackage{diagbox}
\usepackage{multirow}
\usepackage{apptools}
\usepackage{diagbox}
\usepackage{makecell}
\usepackage{mathtools}
\usepackage{caption}
\usepackage{subcaption}

\newtheorem{theorem}{Theorem}[section]
\newtheorem{proposition}[theorem]{Proposition}
\newtheorem{lemma}[theorem]{Lemma}
\newtheorem{corollary}[theorem]{Corollary}
\newtheorem{definition}[theorem]{Definition}
\newtheorem{assumption}[theorem]{Assumption}
\newtheorem{remark}[theorem]{Remark}

\definecolor{azure}{rgb}{0.0, 0.5, 1.0}
\definecolor{brandeisblue}{rgb}{0.0, 0.44, 1.0}
\definecolor{darkpastelgreen}{rgb}{0.01, 0.75, 0.24}
\definecolor{darkpastelpurple}{rgb}{0.59, 0.44, 0.84}
\definecolor{darktangerine}{rgb}{1.0, 0.66, 0.07}
\definecolor{debianred}{rgb}{0.84, 0.04, 0.33}
\definecolor{hanpurple}{rgb}{0.32, 0.09, 0.98}
\definecolor{deepcerise}{rgb}{0.85, 0.2, 0.53}
\definecolor{emerald}{rgb}{0.31, 0.78, 0.47}
\definecolor{fuchsia}{rgb}{1.0, 0.0, 1.0}
\definecolor{flamingopink}{rgb}{0.99, 0.56, 0.67}
\definecolor{lightseagreen}{rgb}{0.13, 0.7, 0.67}
\definecolor{mayablue}{rgb}{0.45, 0.76, 0.98}

\newcommand{\PP}{\mathbb{P}}
\newcommand{\EE}{\mathbb{E}}

\newcommand{\RR}{\mathcal{R}}
\newcommand{\R}{\mathbb{R}}

\newcommand{\As}{\mathcal{A}}
\newcommand{\Ss}{\mathcal{S}}
\newcommand{\T}{\mathcal{T}}
\newcommand{\fsizee}{0.9}
\newcommand{\fsizeee}{0.54}

\DeclareMathOperator*{\argmin}{arg\,min}

\DeclareMathOperator*{\corr}{corr}
\DeclareMathOperator*{\cov}{cov}

\title{Towards Multi-Agent Reinforcement Learning driven Over-The-Counter Market Simulations}

\author{
Nelson Vadori \thanks{corresponding author: nelson.n.vadori@jpmorgan.com}\\
J.P. Morgan AI Research\\
\And
Leo Ardon \\
J.P. Morgan AI Research\\
\And
Sumitra Ganesh \\
J.P. Morgan AI Research\\
\And
Thomas Spooner \\
J.P. Morgan AI Research\\
\And
Selim Amrouni \\
J.P. Morgan AI Research\\
\And
Jared Vann \\
J.P. Morgan AI Research\\
\And
Mengda Xu \\
J.P. Morgan AI Research\\
\And
Zeyu Zheng \thanks{work completed in the context of an internship at J.P. Morgan AI Research.} \\
University of Michigan \\ and \\ J.P. Morgan AI Research\\
\And
Tucker Balch \\
J.P. Morgan AI Research\\
\And
Manuela Veloso \\
J.P. Morgan AI Research\\
}

\begin{document}

\maketitle

\begin{abstract}
We study a game between liquidity provider and liquidity taker agents interacting in an over-the-counter market, for which the typical example is foreign exchange. We show how a suitable design of parameterized families of reward functions coupled with shared policy learning constitutes an efficient solution to this problem. By playing against each other, our deep-reinforcement-learning-driven agents learn emergent behaviors relative to a wide spectrum of objectives encompassing profit-and-loss, optimal execution and market share. In particular, we find that liquidity providers naturally learn to balance hedging and skewing, where skewing refers to setting their buy and sell prices asymmetrically as a function of their inventory. We further introduce a novel RL-based calibration algorithm which we found performed well at imposing constraints on the game equilibrium. On the theoretical side, we are able to show convergence rates for our multi-agent policy gradient algorithm under a transitivity assumption, closely related to generalized ordinal potential games.
\end{abstract}

\tableofcontents

\vspace{10mm}

\section{Introduction}
\label{sec1}

\textbf{Market context}. We focus on a dealer market, also commonly known as over-the-counter (OTC), where a single security is being traded. Examples of such markets include foreign exchange, the largest financial market in the world. We can think of the security in question to be eurodollar. Unlike for typical stocks, trading activity between different market participants - or \textit{agents} -  does not occur on a single centralized exchange observable by everyone, rather the market is \textit{decentralized}: an agent willing to trade typically does so via a private connection with a dealer or via a liquidity aggregator, which provides prices from a set of dealers \cite{oomen}. The nature of these interactions naturally makes the market \textit{partially observable}, since a given agent does not know who other agents are connected to, neither prices which they transact at. Agents also have access to more conventional exchanges, also called ECNs (Electronic Communication Networks). ECNs are useful in they provide a reference price to market participants, and a platform to trade in case of immediate need, for example when hedging. For eurodollar, an example of such ECN is Electronic Brokerage Service (EBS), which tick size is 0.5 basis points at the time of the writing \footnote{tick sizes of 0.1 basis points are also seen in the eurodollar market. One basis point equals $10^{-4}$.}. We refer to \cite{bis1,bis2} for a review of the FX market. 

\textbf{Market agents}. Agents trading in such markets can mostly be split into two classes: Liquidity Providers (LPs), and Liquidity Takers (LTs)\footnote{an agent can also be a blend of the two, but we do not consider this case in the present paper, however the methods we employ allow for straightforward extension to this case.}. The former are essentially market makers, whose goal is to provide liquidity to the market by continuously streaming prices to LTs at which they are willing to buy and sell. The latter trade with LPs for various motives, be it for speculative reasons, or because they simply need a specific quantity of the security in question for exogenous purposes. LPs' problem is to optimally manage the net inventory they are left with after trading with LTs, which is subject to uncertain price fluctuations: this is either done by \textit{skewing}, i.e. optimally adjusting their buy and sell prices asymmetrically so as to reduce their inventory (\textit{internalization}), or by \textit{externalizing} their inventory by trading on the ECN market. While most of the literature - often to preserve analytical tractability - focuses on LPs maximizing risk-penalized profit-and-loss (PnL), in reality LPs also have other objectives, for example increasing their market share. The latter is the fraction of total LT flow a specific LP is able to get, and is not trivial to estimate due to the partial observability of the OTC market. Interestingly, the market share objective of LPs can be seen as the analogous of the LT's traded quantity target: schematically, one can think of LP and LT agents as aiming to maximize a trade-off between risk-penalized PnL on one side, and a purely quantity related target on the other, reflecting their willingness to trade independently of costs. In the present paper, we want to capture, for each agent class, a spectrum of these objectives reflecting real-life markets. A pictorial view of the market is presented in figure \ref{simpic}.

\begin{figure}[ht]
  \centering
  \centerline{\fcolorbox{black}{white}{\includegraphics[scale=0.33]{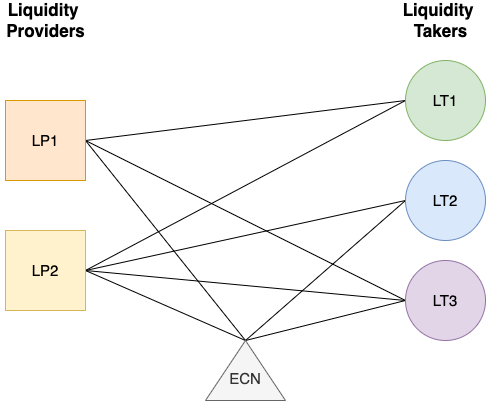}}}
  \caption{Overview of the simulator. Different colors represent different agent objectives, a mix between risk-penalized PnL, market share and traded quantity targets. Note the partial connectivity structure among agents.}
  \label{simpic}
\end{figure}

\textbf{Our goal and desiderata}. We aim at designing an OTC market simulation of the above setting - formally \textit{$n$-player partially observable general sum Markov game} - between LP agents, LT agents, and ECNs, where all agents learn optimal behaviors (also called \textit{policies}) related to their individual objectives by playing against each other. This is now possible using modern AI techniques in multi-agent reinforcement learning (MARL). Such simulation allows to study, from an empirical perspective, specific what-if scenarios and equilibrium properties, such as impact on a given LP's market share of increasing risk aversion of specific competitors, or of modifying agent connectivity. This is different from recent, interesting work studying a single LP's optimal pricing and hedging problem from a stochastic optimal control point of view \cite{Gueant2017-cw, gueant2021_1, gueant2021_2}. There, other agents are modeled as a single "rest of the world" entity, estimated statistically from historical data. By construction, this entity cannot adapt to changes in the LP's policy.

A central question related to multi-agent simulations is how to instantiate all of the agents' characteristics, including their individual utility functions. This is crucial for practical applications. Typically, one knows some information about the market. For example, we may know the flow response curve of a specific LP as in \cite{Gueant2017-cw}, i.e. how much quantity it gets as a function of its pricing. Or we may know the average market share of specific LPs. Our task then becomes to select all simulator hyperparameters, in our case agents' characteristics, so as to match these observations. We will devote a section on this topic, and introduce a novel RL-driven way to perform such calibration, which we found performed well compared to other approaches.

The $n$-player game we study is, in spirit, close to that of \cite{Bank2018-ip} who study equilibrium interactions between dealers (LPs), clients (LTs), and open market (ECN). We present below some desiderata that we have for our MARL-driven simulation:

\begin{itemize}
    \item \textbf{(Desiderata 1)} Due to the structure of the OTC market, the simulation should account for connectivity between pairs of agents, namely the simulation should take as an input a connectivity graph.
    
    \item \textbf{(Desiderata 2)} The simulation should allow agents to optimize for a spectrum of objectives reflecting real-life dealer markets.
    
    \item \textbf{(Desiderata 3)} Learning of optimal agent actions in (Desiderata 2) should be done from the definition of the game (rules) and by experimenting the game (playing).
    
    \item \textbf{(Desiderata 4)} The simulation should be flexible enough to reproduce observations that specific agents may have of the market, namely it should allow for calibration to these observations.
\end{itemize}

(Desiderata 1) is a structural requirement in line with the nature of the market. 

For (Desiderata 2), we will focus on the case where LPs' objectives cover the spectrum from maximizing PnL to market share, and LTs' objectives cover the spectrum from maximizing PnL to trading a specific quantity for exogenous motives, independent of cost. Both of these are trade-offs between a PnL component and a purely quantity related component, which makes our agents' formulation compact. Note that optimizing for a mix between a trading target and PnL is exactly optimal execution, where the goal is to trade a certain quantity $q$ while minimizing costs. As we will see later, our supertype-based design will allow us to learn efficiently optimal policies associated to these trade-offs. 

A reformulation of (Desiderata 3) is that agents' policies should emerge from agents playing the market game while optimizing for a mix of these objectives, as opposed to being handcrafted. In particular, this opens the door to discovering emergent behaviors from the (complex) game itself.

For (Desiderata 4), we want the simulation to capture certain known observations about the market, for example that the market share of a given agent is equal to some level, and the model should be flexible enough to capture these observations. Constraints such as those can be achieved by having agents of different nature, or types, and optimally balancing those types so as to match the desired targets on the emergent behavior of agents. We highlight that this is a non trivial requirement: (Desiderata 3) requires learning an equilibrium where agents' policies are optimal, hence, together with (Desiderata 4), require learning an equilibrium subject to specific constraints. As a simple example, increasing the market share of a specific LP could be achieved by increasing the risk aversion of its competitors.

\textbf{Our Contributions}. We formalize the game between LPs and LTs based on the concepts of agent type and supertype (section \ref{secpomg}) and show how to have populations of such agents efficiently learn optimal policies associated to a spectrum of objectives encompassing profit-and-loss, optimal execution and market share, using RL (section \ref{seclearn}). The ECN engine that we consider in this paper is a limit order book which evolves only as a consequence of agents' orders (market, limit, cancel). These agents are the LP and LT agents previously discussed, as well as an ECN agent which is in charge of sending orders to the ECN only. The construction of the order list of the ECN agent is based on a discrete-time Markovian model for the volume at the top levels of the book for which we establish the continuous-time limit. The quadratic variation of the latter limit process is a polynomial of order two, which we show generalizes the first order polynomial in \cite{Cont:2021}. We show how this model can be extended to the non-Markovian case using neural networks, where the evolution of the order book depends on its history (section \ref{sececn}). We study the game theoretic implications of LP agents using a shared policy suitably conditioned on their type and introduce the concept of shared equilibrium, a pure-strategy symmetric Nash equilibrium on the extended space of type-dependent stochastic policies (section \ref{secsharedeq}). We present convergence rates to such equilibria under a transitivity assumption, closely related to generalized ordinal potential games. We show how modern game theoretical tools can help analyze our complex market game (section \ref{secgamec}), namely differentiable games and their potential-Hamiltonian decomposition of \cite{pmlr-v97-balduzzi19a,dgame}. There, we introduce the concepts of Hamiltonian and potential weights of a differentiable game, and show on simple examples how these quantities give interesting insights into our market game. We introduce a novel RL-based calibration algorithm which performs well compared to a Bayesian optimization baseline (section \ref{seccalib}). There, LP and LT agents learn to reach an equilibrium jointly with a RL-based calibrator agent aiming at imposing constraints on that equilibrium. Finally, we conduct experiments to show i) the efficiency of our novel calibration algorithm, and ii) the emergent behaviors learnt by the LP agents as a function of their objectives (section \ref{secexp}). In particular, we find that LPs naturally learn to balance hedging and skewing, where the latter refers to setting their pricing asymmetrically on the bid and ask sides as a function of their inventory.

\textbf{Related work.} The market making problem has been extensively studied in finance and economics literature, largely as an optimal control problem. Classical models such as \cite{Garman1976-oy}, \cite{Amihud1980-qz}, \cite{Ho1981-yy}, and more recently \cite{Avellaneda2008-de, Gueant2013-gc, gueantberg}, focus on the role of inventory risk in determining the optimal pricing strategy for a market maker. Other models, such as \cite{Glosten1985-cx}, study the role of adverse selection risk arising from informed traders in the market. All these works model a single market maker and make assumptions about the distributions of order arrivals and price to derive analytical solutions for the market maker's pricing policy using stochastic optimal control techniques. Market making has also been studied in the agent-based modeling literature \cite{Darley2000-aw}, \cite{Das2005-ji},\cite{Das2008-mu}, \cite{Jumadinova2010-pu}, \cite{Wah2017-lb}. Most previous work focuses on market making in limit order book markets where all agents submit their orders to a central matching facility; there has been relatively little work focused on dealer markets where the agents interact directly with each other. Recently, \cite{Gueant2017-cw}, \cite{Bank2018-ip} and \cite{Ghoshal2016-kj} have extended the optimal control approach to dealer markets. In \cite{gueant2021_1,gueant2021_2}, liquidity providers' optimal hedging is considered, in addition to the optimal pricing problem. \cite{Chan2001-hp}, and more recently, \cite{Lim2018-yh} and \cite{Spooner2018-vj} have developed RL-based market making approaches for limit order book markets; however, they do not explicitly model the competing market makers or study different competitive scenarios. \cite{contwei} recently study competition and collusion among a set of market makers, and use reinforcement learning as a mean to solve for the equilibrium of the game. \cite{ijcai2020-633} study a discrete-time zero-sum game between a market maker and an adversary and show that adversarial reinforcement learning can help produce more robust policies. \cite{ganesh2019reinforcement,ardon} are some of the work closest to ours, where competition among a set of market makers and investors is studied by means of reinforcement learning. The work of \cite{jaim3} presents various algorithmic trading models using stochastic optimal control.

Regarding the literature on MARL, there are no algorithms with convergence guarantees to Nash equilibria in general sum $n$-player Markov games with continuous action and state spaces, as it is the case for our game among LP and LT agents. Our game is further partially observable, which makes matters even worse. Existing algorithms have guarantees in specific settings. For example, fictitious play \cite{fictitious} (potential games, zero-sum games); double oracle \cite{doubleo} (finite action spaces, but in the worst case, all actions of the game need to be enumerated); policy-sparse response oracles (PSRO) \cite{psro}, a modern variant of double oracle, has been proved to be efficient in practice but has no theoretical guarantees: the space of policies is infinite, hence you cannot simply enumerate it as you would do in double oracle. Most of the MARL literature has focused on cooperative games \cite{Gupta2017-it,nipsFoerster,ps4,maven}. Recently, an extension of normal form potential games to Markov games has been considered under the terminology "Markov potential game" \cite{markovpg}. Any Markov game whereby agents take actions and transition from states to states over multiple timesteps can be recast as as a one-shot game with utilities the agents' value functions $V_i(\pi_i, \pi_{-i})$. The latter game is a one-shot game over a larger pure strategy space, hence it is sometimes called a "meta-game" \cite{psro}. Markov potential games are related to potential games by the intuitive fact that the meta-game is assumed to be potential. \cite{markovpg} shows that independent policy gradient converges to Nash equilibria for this class of games. In contrast, our convergence proof will be done under a transitivity assumption similar to generalized ordinal potential games \cite{potential}, which are more general than (exact) potential games. Further, most MARL algorithms consider $n$ fully heterogeneous agents. In our case, we want to tie these agents together by the concept of type: two agents with similar risk aversion should behave similarly. The latter is not guaranteed using independent policy gradient, in addition of being sample inefficient. Rather, conditioning the policy on the agent type represents an efficient alternative. This is the concept of parameter sharing introduced in \cite{Gupta2017-it,nipsFoerster}, which falls under the centralized training with decentralized execution paradigm (CTDE). Although their work considers a cooperative setting where agents maximize a joint reward, parameter sharing is actually the only method out of their proposed three that doesn't require reward sharing, and we exploit this fact in our work. The approach we employ in this paper (based on a combination of policy sharing and definition of agent types) can be used with any on-policy RL algorithm, such as A2C \cite{a2c}, TRPO \cite{pmlr-v37-schulman15}, or PPO \cite{ppo}. It has been shown that PPO enjoys competitive performance when suitably tuned \cite{ppo2} and when using sufficient parallelism to generate on-policy experience. For this reason, we use PPO in this paper, with our experiments employing 60 to 90 CPUs in parallel to generate sufficient amount of data for training the shared policy. In contrast, there has been recent academic progress in MARL which has focused on off-policy learning, such as MADDPG \cite{maddpg} and value-decomposed Q-learning \cite{ps3,ps4}.

On the topic of calibration, \cite{zheng2020ai} uses a shared policy for worker agents earning individual rewards and paying tax. There, the RL-based tax planner shares some similarities with our RL calibrator, although our calibrator is responsible for optimally picking agent type distribution rather than public information observable by all agents, and further updates its policy on a slower timescale so as to allow equilibria to be reached by the shared policy. The idea of using RL to calibrate parameters of a system probably goes back to \cite{earl}, in the context of evolutionary algorithms. As mentioned in \cite{Avegliano2019-vx}, there is currently no consensus on how to calibrate parameters of agent-based models. Most methods studied so far build a surrogate of the multi-agent system \cite{Avegliano2019-vx,surlamp}. The term "surrogate" is very generic, and could be defined as a model that approximates the mapping between the input parameters and some output metric of the multi-agent system. \cite{surlamp} studies classifier surrogates, and in contrast to the latter and other work on calibration, our work is based on a dual-RL approach where our RL calibrator learns jointly with RL agents learning an equilibrium. In our experiments, we compare our approach to a Bayesian optimization baseline that builds such a surrogate. Inverse RL \cite{irl} could be used for calibration, but it aims at recovering unknown rewards from input expert policy: in this work we don't need the latter and assume that rewards are known for each agent type, and that the goal is to find the optimal agent type distribution.

\section{Supertype-based multi-agent simulation model}

\subsection{Agent types and supertypes}
\label{secpomg}

\textbf{Partially Observable Markov Game setting.} Throughout, we call \textit{agent class} $\kappa \in \{LP, LT\}$. The game that we consider is a game among $n_{LP}$ LPs and $n_{LT}$ LTs, of finite and deterministic time horizon $T$. We write $n_{tot}=n_{LP}+n_{LT}$ the total number of agents, and let $n_{ECN}$ be the number of ECNs. Connectivity among agents, and among agents and ECNs, is given by an input connectivity graph. Two elements of this graph can interact together if and only if they are connected, and are always disconnected if they belong to the same class. Real markets evolve in continuous time, but here, we will consider a discretized version of the game, where time $t$ is integer-valued and each time increment corresponds to a fixed value, for example one second. This will make our framework amenable to RL. We will use the related terminology throughout, where \textit{agents} observe private \textit{states}, take \textit{actions} and obtain per-timestep \textit{rewards} as a consequence. We call the simulation of a stochastic game path over $[0, T]$ an \textit{episode}, and the aggregated utility gained by each one of the agents over timesteps its \textit{cumulative reward}. All of our random variables are defined on the same probability space, and we adopt the notational convention used in the field of RL that $X \sim \mu$ indicates that the random variable $X$ has distribution $\mu$.

The schematic and informal game structure is as follows: at each time $t \in [0, T]$, LPs first stream bid and ask prices\footnote{these are price curves representing price as a function of quantity traded.} at which they are willing to trade with LTs they are connected to, and additionally decide a hedge fraction of their current inventory to trade on the ECN market. Second, LTs decide a trade size (possibly zero) and direction based on the observed prices. When that decision is made, the trade will \textit{always} occur at the best possible price among the LPs and ECN the LT is connected to. In this sense, this is a multi-stage, or stochastic Stackelberg game, since LTs play in response to LPs. The exact definition of players' actions and states in our OTC context will be given in section \ref{secrlagents}, for now we focus on defining the game as a general framework based on the concept of supertype.

We formalize the game as a $n$-player partially observable Markov game \cite{Hansen2004-un}. At each time $t$, each agent observes a fraction of the environment, proceeds to taking an action, and gets a reward which depends on all agents' actions and states. In our Stackelberg context, states of LTs include connected LPs' actions corresponding to pricing. Within each agent class, all agents share the same action and state spaces $\As^\kappa$ and $\Ss^\kappa$, typically subsets of $\R^d$. We make no specific assumption on the latter spaces unless specifically mentioned, and denote the joint action and state as $\bm{a^\kappa_t}:=(a^{(1,\kappa)}_t,...,a^{(n_\kappa,\kappa)}_t)$ and $\bm{s^\kappa_t}:=(s^{(1,\kappa)}_t,...,s^{(n_\kappa,\kappa)}_t)$. Partial observability is a consequence of our connectivity graph: an agent observes ECN information only if connected to it. Similarly, a LT observes a LP's pricing only if connected to it. Any agent-specific information such as its inventory or trade history remains private to the agent. The joint state $(\bm{s^{LP}_t},\bm{s^{LT}_t})$ constitutes all information available in the system at time $t$.

\textbf{Agent types and supertypes.} Before proceeding to the formal definition of agent types and supertypes, we give an informal description to build the reader's intuition, cf. also \cite{vadori:2020:calibration}. In addition to the agent class which specifies in particular the structure of state and action spaces, agents of a given class differ by the utility they are optimizing for, also called \textit{reward} in the context of RL. In our case and as usual in games, agents' rewards depend on players' actions and states, but also on additional agent-specific parameters capturing their characteristics and objectives, which we call \textit{agent type}. In our market framework, the type of a given agent includes its risk-aversion acting as a regularizer for PnL, its trade-off parameters between PnL, market share and trade quantity targets discussed in section \ref{sec1}, as well as its connectivity structure to other agents. This aspect will be further detailed in section \ref{secrlagents}. Agent types correspond to a parametrized family of reward functions capturing a diverse set of objectives. Thinking in terms of agent type within a given agent class will allow us to learn efficiently a spectrum of behaviors using a single neural network suitably conditioned on the type, by exploiting their generalization power. As a simple example, if two agents have risk aversion parameters differing by a small amount but are identical otherwise, it is not efficient to learn two distinct policies, rather it is better to condition the policy on the risk aversion so as to learn a family of behaviors at once. Note that this generalization trick has also been used in the context of mean-field games \cite{perrin} where they condition policies on the initial agent population so as to learn policies that can generalize well in that variable. \cite{leal} has used it to learn a family of optimal controls as a function of the trader's preferences.

It can be convenient to think of agents in terms of distributions of agents, and not at the individual level. This allows to scale the simulation in a tractable way with fewer parameters, while observing diversity in agents' characteristics. This leads to our concept of \textit{agent supertype}, which are simply parameters driving the distribution of agent type. Types will be sampled from supertypes at the beginning of each episode, and remain constant throughout a given episode. In this sense, supertypes can be seen as behavioral templates according to which agents can be cloned probabilistically. Note that introducing noise via the randomization of types can help convergence to equilibria, as noted in \cite{xu1} in a linear-quadratic framework. Typically, we create groups of agents who share the same supertype, so that the number of distinct supertypes is typically much less than the number of agents. In our setting, agent supertypes will include in particular probabilities of being connected to other agents, so that actual connectivities between agents of different supertypes can be sampled as Bernoulli random variables.

Formally, we assign to each agent $i$ a supertype $\Lambda^\kappa_i \in \Ss^{\Lambda^\kappa_i}$, with $\bm{\Lambda^\kappa}:=(\Lambda^\kappa_i)_{i \in [1,n_\kappa]}$. At the beginning of each episode, agent $i$ is assigned a type $\lambda^\kappa_i \in \Ss^{\lambda^\kappa}$ sampled probabilistically as a function of its supertype, namely $\lambda^\kappa_i \sim p_{\Lambda^\kappa_i}$ for some probability density function $p_{\Lambda^\kappa_i}$, and initial states $s^{(i,\kappa)}_0$ are sampled independently according to the initial distribution $\mu^0_{\lambda^\kappa_i}$. This is formally equivalent to extending agents' state space to $\Ss^\kappa \times \Ss^{\lambda^\kappa}$, with a transition kernel that keeps $\lambda^\kappa_i$ constant throughout an episode and equal to its randomly sampled value at $t=0$. Note that typically, $\Ss^{\Lambda^\kappa_i}$ and $\Ss^{\lambda^\kappa}$ are subsets of $\R^d$.

\textbf{Rewards, state transition kernel and type-symmetry assumption.} In the following we make symmetry assumptions whose only purpose is to guarantee that an agent's expected cumulative reward only depends on its supertype, given that all agents' policies are fixed.
Let $z^{(i,\kappa)}_t:=(s^{(i,\kappa)}_t,a^{(i,\kappa)}_t,\lambda^\kappa_i)$. At each time $t$, LP agent $i$ receives an individual reward $\RR^{LP}(z^{(i,LP)}_t,\bm{z^{(-i,LP)}_t}, \bm{z^{(LT)}_t})$, where the vector $\bm{z^{(-i,\kappa)}_t}:=(z^{(j,\kappa)}_t)_{j \neq i}$. Similarly, LT agent $i$ receives an individual reward $\RR^{LT}(z^{(i,LT)}_t,\bm{z^{(-i,LT)}_t}, \bm{z^{(LP)}_t})$. Denote $\mathcal{Y}^\kappa := \Ss^\kappa \times \As^\kappa \times \Ss^{\lambda^\kappa}$. The state transition kernel $\T: (\mathcal{Y}^{LP} )^{n_{LP}} \times (\mathcal{Y}^{LT} )^{n_{LT}} \times (\Ss^{LP})^{n_{LP}} \times (\Ss^{LT})^{n_{LT}} \to [0,1]$ is denoted $\T(\bm{z^{(LP)}_t}, \bm{z^{(LT)}_t}, \bm{s_t'^{LP}}, \bm{s_t'^{LT}})$, and represents the probability of reaching the joint state $(\bm{s_t'^{LP}}, \bm{s_t'^{LT}})$ conditionally on agents having the joint state-action-type structure $(\bm{z^{(LP)}_t}, \bm{z^{(LT)}_t})$. 

We now proceed to making assumptions on the rewards and state transition kernel that we call \textit{type-symmetry}, since they are similar to the anonymity/role-symmetry assumption in \cite{Li2020-fu}. Their only purpose is to guarantee that the expected reward of an agent in (\ref{vi2}) only depends on its supertype $\Lambda_i^\kappa$. In plain words, the latter guarantees that from the point of view of a given agent, all other agents are interchangeable, and that two agents with equal supertypes and policies have the same expected cumulative reward. 
\begin{assumption}
\label{asts}
\textbf{(Type symmetry)} For $\kappa \in \{LP, LT\}$, $\RR^{\kappa}$ is invariant w.r.t. permutations of both its second and third arguments, namely for any permutations $\rho$, $\mu$ we have $\RR^{\kappa}(\cdot,\bm{z_1}^\rho,\bm{z_2}^\mu)=\RR^{\kappa}(\cdot,\bm{z_1},\bm{z_2})$ for any $\bm{z_1}$, $\bm{z_2}$, and $\T(\bm{z_1}^\rho,\bm{z_2}^\mu,\bm{s_1}^\rho,\bm{s_2}^\mu) = \T(\bm{z_1},\bm{z_2},\bm{s_1},\bm{s_2})$ for any $\bm{z_1}$, $\bm{z_2}$, $\bm{s_1}$, $\bm{s_2}$, where superscripts denote the permuted vectors. 
\end{assumption}

\textbf{Time-horizon.} We have assumed that the time-horizon $T$ is deterministic and finite. In this context, it is well-known that the agents' policies in (\ref{pidef}) need to be time-dependent. We deal with this aspect the usual way by including time in the agents' states $s^{(i,\kappa)}_t$ \footnote{in the context of Markov semigroups, a time-dependent semigroup on a given space can be recast as a time-homogenous semigroup in the space-time domain.}. Our analysis can be extended straightforwardly to the case where $T$ is a random variable corresponding to the first time that the joint state $(\bm{s^{LP}_t},\bm{s^{LT}_t})$ hits some set (episodic case). The same goes for the case $T=+\infty$, but in that case we must have a discount factor $\zeta<1$ in (\ref{vi2}). The two latter cases do not require time to be included in the agents' states, since the problem becomes time-homogeneous.

\subsection{Efficient learning of a spectrum of agent behaviors via reinforcement learning}
\label{seclearn}

\subsubsection{Shared policy conditioned on agent type}
\label{secshared}

Multi-agent learning in partially observable settings is a challenging task. When all agents have the same action and state spaces, the work of \cite{nipsFoerster,Gupta2017-it} has shown that using a single shared policy $\pi$ across all agents represents an efficient training mechanism. This falls into the centralized training with decentralized execution paradigm (CTDE). In that case, $\pi$ is a neural network that takes as input the individual agent states $s^{(i)}_t$ and outputs (a probability over) individual agent actions $a^{(i)}_t$, hence the terminology \textit{decentralized execution}, since the network doesn't take as input the joint agent states. The network is trained by collecting all agent experiences simultaneously and treating them as distinct sequences of local states, actions and rewards experienced by the shared policy, hence the terminology \textit{centralized training}. In our case, we have two classes of homogenous agents, LPs and LTs, so we will use one shared policy per class.

We include the agent type $\lambda_i^\kappa$ in the local states and hence define the shared policy over the extended agent state space $\Ss^\kappa \times \Ss^{\lambda^\kappa}$. Denoting $\mathcal{X}^\kappa$ the space of functions $\Ss^\kappa \times \Ss^{\lambda^\kappa} \to \Delta(\As^\kappa)$, where $\Delta(\As^\kappa)$ is the space of probability distributions over actions, we then define:
\begin{align}
\label{pidef}
\begin{split}
    &\mathcal{X}^\kappa := \left[ \Ss^\kappa \times \Ss^{\lambda^\kappa} \to \Delta(\As^\kappa)\right], \\
    &\pi^\kappa(da|s,\lambda) := \PP\left[a^{(i,\kappa)}_t \in da | s^{(i,\kappa)}_t=s, \lambda_i^\kappa=\lambda\right], \hspace{5mm} \pi^\kappa \in \mathcal{X}^\kappa.
\end{split}
\end{align}

Due to the partial observability of the game (imperfect information), optimal agents' actions should depend on the whole history of their private states $(s^{(i)}_k)_{k \leq t}$ \cite{Lockhart19ED}. As often done so (\cite{Gupta2017-it}), we can use the hidden variable $h$ of a LSTM to encode the agent history of states and condition the policy on it $\pi(\cdot|s^{(i,\kappa)}_t,h^{(i,\kappa)}_{t-1},\lambda_i^\kappa)$: to ease notational burden we do not include it in the following, but this is without loss of generality since $h$ can always be encapsulated in the state $s^{(i,\kappa)}_t$.

Due to our type-symmetry assumption \ref{asts}, we see that the expected reward $V_{\Lambda_i^\kappa}$ of each agent $i$ only depends on its supertype $\Lambda_i^\kappa$ and the shared policies $\pi^{LP}$, $\pi^{LT}$\footnote{it also depends on other agents' supertypes $\bm{\Lambda_{-i}^\kappa}$ independent of their ordering, but since we work with a fixed supertype profile $\bm{\Lambda^\kappa}$ for now, $\bm{\Lambda_{-i}^\kappa}$ is fixed when $\Lambda_i^\kappa$ is.}.
\begin{align}
\label{vi2}
V_{\Lambda_i^\kappa}(\pi^{LP},\pi^{LT}):=\EE_{\substack{\lambda_i^\kappa \sim p_{\Lambda_i^\kappa}\\ a^{(i,\kappa)}_t \sim \pi^\kappa(\cdot|\cdot,\lambda_i^\kappa)}} \left[ \sum_{t=0}^T \zeta^t \RR^\kappa(z^{(i,\kappa)}_t,\bm{z^{(-i,\kappa)}_t}, \bm{z^{(-\kappa)}_t})\right], \hspace{1mm} \pi^\kappa \in \mathcal{X}^\kappa.
\end{align}

where $\zeta \in [0,1]$ is the discount factor, and the superscript $(-\kappa)$ denotes the agent class - LP or LT - which is not $\kappa$. 

Considering the expected agent rewards in (\ref{vi2}), we will aim to find Nash equilibria of the game. Our game is not analytically tractable, so we will focus on reaching empirical convergence of agents rewards, using RL. This is done in sections \ref{secgamec} and \ref{secexp}. In the case where LT's do not learn, i.e. their policy $\pi^{LT}$ is fixed\footnote{for example, they buy or sell with equal probability.}, we investigate in section \ref{secsharedeq} the nature of equilibria that can be reached by the LPs using a shared policy.

It is to be noted that we do not assume that agents know the supertype profile $\bm{\Lambda}:=(\bm{\Lambda^{LP}},\bm{\Lambda^{LT}})$, contrary to Bayesian games where the "type prior" is assumed to be of common knowledge. One could always define the type $\lambda^\kappa_i$ so that it contains the supertype $\Lambda^\kappa_i$ (or all supertypes) by construction, but this is a particular case and in general agents do not know $\Lambda^\kappa_i$, and therefore cannot condition their policy on it.

Since agents may have different states at a given point in time, sharing a network still allows different actions across agents. The rationality assumption underlying parameter sharing is the following: if two agents have equal types and equal sequences of historical states at a given point in time, then they should behave the same way, i.e. the distributions of their actions should be equal. Agents do not directly observe types nor supertypes of other agents, rather they observe these quantities indirectly via the result of their own actions $a^{(i,\kappa)}_t$, which is encoded in their (history of) states $s^{(i,\kappa)}_t$. For example, a LP agent quoting a given price will not receive the same trading activity from LTs depending on the risk aversion of other LPs. Although our focal LP cannot observe the risk aversion of his competitors, he observes the resulting trading activity and can tailor his behavior to the latter.

\subsubsection{Reinforcement learning design of OTC agents}
\label{secrlagents}

In the remainder of the paper, we assume for simplicity that there is one ECN, $n_{ECN}=1$, and we denote $P_t$ its mid price at the beginning of timestep $t$. Our framework supports $n_{ECN}>1$, in which case a reference price $P_t$ would typically be computed as a liquidity-weighted average mid-price over the $n_{ECN}$ mid-prices. The ECN, and in particular the construction of $P_t$, is discussed in section \ref{sececn}. 

Following our discussion in section \ref{sec1}, we represent the utilities of both LP and LT agents as a trade-off between a risk-penalized PnL term, and a purely quantity-related, PnL-independent term. In the case of the LP agent, the latter is market share. In the case of the LT agent, it is a pair of trade objectives on the bid and ask sides. This makes our agent formulation compact and unified as a hybrid PnL/volume objective. We summarize the RL formulation of our LP and LT agents in table \ref{tabagents} and proceed to the description of the details in the remainder of the section.

\textbf{Agent risk-penalized PnL.} Let $PnL_{t+1}$ be the cumulative PnL of a given agent at the end of time $t$ (i.e. beginning of time $t+1$), and  $\Delta PnL_{t+1} := PnL_{t+1} - PnL_{t}$ the incremental PnL. We then have the usual decomposition of PnL into inventory and spread terms, which follows by simple "accounting":
\begin{align}
\begin{split}
& PnL_{t} = PnL_{\text{inv},t} + PnL_{\text{spread},t} + \ell(q_t), \hspace{5mm} PnL_{\text{inv},t} := \sum_{i=1}^{t} q_{i}(P_{i}-P_{i-1}),\\ 
& PnL_{\text{spread},t} := \sum_{i=1}^{t} \sum_{b \in \mathcal{B}_i} q_i^{b}(P_{i-1} - P_{i-1}^{b}) + \sum_{a \in \mathcal{A}_i} q_i^{a}(P_{i-1}^{a}-P_{i-1}).
\end{split}
\end{align}

 We denote $\ell(q_t)$ the usual terminal inventory penalty capturing the fact that liquidating $q_t$ cannot usually be done exactly at price $P_t$ \cite{Gueant2017-cw}, where $q_t$ is the agent's net inventory at the beginning of time step $t$. In this work we take $\ell \equiv 0$ since we are not specifically interested in studying the impact of this penalty, but our framework allows for arbitrary $\ell$.

$PnL_{\text{inv},t}$ is the inventory PnL and captures the fluctuations of the reference mid-price $P_t$ over time.

$PnL_{\text{spread},t+1}$ is the cumulative spread PnL obtained at the end of time step $t$ as a consequence of the trades performed during time step $t$ and before. It captures the local profit of individual trades relative to the mid-price $P_t$. Let $\mathcal{A}_{t+1}$ be the set of trades performed by the agent on its ask side during time step $t$ at prices $\{P^a_{t}\}_{a \in \mathcal{A}_{t+1}}$ and absolute quantities $\{q^a_{t+1}\}_{a \in \mathcal{A}_{t+1}}$ (i.e. sold by the agent). Similarly, for trades on the agent's bid side, we use the notations $\mathcal{B}_t$, $q^b_t$, $P^b_t$. Note that with these notations, $\Delta q_t = \sum_{b \in \mathcal{B}_t} q_t^{b} - \sum_{a \in \mathcal{A}_t} q_t^{a}$. 

Trades occurring on the OTC market between LPs and LTs can be assimilated to market orders. Indeed, LPs stream prices at which they are willing to trade with LTs, which in turn decide if and how much they want to trade. In particular, LTs cannot send limit orders to LPs. Precisely, these prices take the form of a pair of price curves, bid and ask, representing price as a function of quantity traded. Equivalently, one can see this pair of curves as an order book. As specified below and to keep consistency with the LP-LT interaction, we consider that the orders performed by LP and LT agents at the ECN are market orders, but we discuss how to extend their action spaces to include the possibility of limit orders.

\textbf{Liquidity Provider Agent.} Let $m_{t+1}$ be the LP's market share during time step $t$, namely the fraction of the total quantity traded by all LTs that the LP was able to attract (buy and sell included), and $\widehat{m}_{t+1} = \frac{1}{t+1} \sum_{k=1}^{t+1} m_{k}$ its running average. Formally, for a LP $i$:
$$
m^i_t = \frac{\sum_{p \in\{bid,ask\}}\sum_{k=1}^{n_{LT}} q_{i,k,p,t}}{\sum_{i=1}^{n_{LP}}\sum_{p \in\{bid,ask\}}\sum_{k=1}^{n_{LT}} q_{i,k,p,t}},
$$ 
where $q_{i,k,p,t}$ is the total absolute quantity traded between LT $k$ and LP $i$ on side $p$, and $\sum_{i=1}^{n_{LP}} m^i_t=1$. Note that a given LP competes with all other LPs but also with ECNs. We define the term $\mathcal{M}_{t+1}$ as the distance to a target $m_{\text{target}}$:
\begin{equation*}\label{eq:LP Market Share}
    \mathcal{M}_{t+1}(m_{\text{target}}) := \left| \widehat{m}_{t+1} - m_{\text{target}} \right|.
\end{equation*}

In this work we take $m_{\text{target}}=1$, so that $\mathcal{M}_{t+1}$ will simply relate to maximizing market share (since it cannot be more than one). The parameters $\eta$, $\omega$ are trade-off parameters between PnL and market share: the former is a normalizer to make both quantities comparable, while the latter is a weight that gives more or less importance to the PnL.

In real markets, LPs aim at maximizing both PnL and market share. We reflect this in our reward formulation $\RR^{LP}_{t+1}$, which is the reward obtained by the LP agent at the end of time $t$, due to actions performed during time $t$\footnote{we remind that $\Delta X_{t+1} := X_{t+1} - X_{t}$.}:
\begin{align}
\begin{split}
\label{rewardlp}
    &\RR^{LP}_{t+1} := \omega \cdot \eta \cdot \Delta PnL_{t+1}^\gamma - (1-\omega) \cdot \Delta \mathcal{M}_{t+1}(m_{\text{target}}), \hspace{5mm} \eta>0, \hspace{2mm} \omega \in [0,1],\\
    &PnL_{t}^\gamma :=PnL_{t} - \gamma \sum_{k=1}^t\left| \Delta PnL_{\text{inv}, k}\right|.
\end{split}
\end{align}

The parameter $\gamma$ is the risk aversion which typically relates to a quadratic penalty on inventory \cite{Gueant2017-cw}. We choose a $L_1$ penalty on the inventory PnL variation since it makes the penalty homogeneous to PnL, and we chose to penalize inventory PnL rather than inventory since the penalty should be zero when the volatility of $P_t$ is zero. If $P_t$ were a Brownian motion with volatility $\sigma$, the expected value of our inventory PnL penalty would be $\EE \left| \Delta PnL_{\text{inv}, t}\right| = \sqrt{\frac{2}{\pi}}\sigma q_{t}$. Our penalty formulation allows to adapt to changes in the volatility rather than assuming it is constant.

Let $x^b_{t}(q)$, $x^a_{t}(q)$ be the quantity-normalized ECN spreads of trading $q$ on the bid and ask, corresponding to prices $P_t - x^b_{t}(q)$ and $P_t + x^a_{t}(q)$. Let $x_{t}(q) := \frac{1}{2}(x^a_{t}(q)+x^b_{t}(q))$ be the symmetrized spread, and $x_t:= 2 \lim_{q\to 0} x_{t}(q)$ the market spread, i.e. difference between best ask and best bid. The LP constructs its prices on both sides by tweaking the ECN reference price:
\begin{align}
\begin{split}
    P^a_{t}(q, \epsilon_{t,\text{spread}}, \epsilon_{t,\text{skew}}) = P_t + x_{t}(q) +  \frac{1}{2} \epsilon_{t,\text{spread}} x_t + \epsilon_{t,\text{skew}} x_t,\\
    P^b_{t}(q, \epsilon_{t,\text{spread}}, \epsilon_{t,\text{skew}}) = P_t - x_{t}(q) - \frac{1}{2} \epsilon_{t,\text{spread}} x_t + \epsilon_{t,\text{skew}} x_t.
\end{split}
\end{align}
The spread tweak $\epsilon_{t,\text{spread}} \geq -1$ controls the price difference $P^a_{t}-P^b_{t}$ and impacts both sides symmetrically. The skew tweak $\epsilon_{t,\text{skew}} \in \R$ shifts prices towards one side or another. Typically, the LP chooses the latter so as to attract flow to reduce its inventory: this is referred to as \textit{skewing} or \textit{internalization}, and we will see in section \ref{secemergent} that our agents are able to learn such behavior in an emergent manner. Note that the lower bound for $\epsilon_{t,\text{spread}}$ comes from the condition $\lim_{q \to 0} P^a_{t}(q, \epsilon_{t,\text{spread}}, \epsilon_{t,\text{skew}})-P^b_{t}(q, \epsilon_{t,\text{spread}}, \epsilon_{t,\text{skew}}) \geq 0$. We could also equivalently have defined, in the spirit of \cite{gueant2021_1}, $\epsilon_{t,\text{ask}} := \frac{1}{2} \epsilon_{t,\text{spread}} + \epsilon_{t,\text{skew}}$, $\epsilon_{t,\text{bid}} := \frac{1}{2} \epsilon_{t,\text{spread}} - \epsilon_{t,\text{skew}}$. It is a simple bijection between the two formulations, but the spread-skew formulation is more intuitive since it allows to decouple the willingness of the LP to stream competitive prices (spread) to his asymmetric treatment of the two sides (skew). Note that $\frac{1}{2}\epsilon_{t,\text{spread}} = \frac{1}{2} (\epsilon_{t,\text{ask}} + \epsilon_{t,\text{bid}})$ and $2\epsilon_{t,\text{skew}} = \epsilon_{t,\text{ask}} - \epsilon_{t,\text{bid}}$, which highlights that the former is an average spread between the two sides, and the latter a spread asymmetry.

We allow the LP to quote prices at a granularity equal to $\xi_{LP}=0.1$ bp, which is the typical granularity that we can observe in FX liquidity aggregators. In practice, $P_t^a$ and $P_t^b$ are projected on such grid.

In the RL terminology, LP's actions at time $t$ are the pricing parameters $\epsilon_{t,\text{spread}} \geq -1$, $\epsilon_{t,\text{skew}} \in \R$ and a hedge fraction $\epsilon_{t,\text{hedge}} \in [0,1]$, which results in a market order $\epsilon_{t,\text{hedge}} q_t$ at the ECN.

Due to the partial observability of the game, optimal agents' actions should depend on the whole history of the per-timestep states $s^{(i,LP)}_t$ \cite{Lockhart19ED}. These states contain the reference mid-price $P_t$, its own net inventory $q_t$, the fraction of time elapsed $\frac{t}{T}$, its market share $\widehat{m}_{t}$, liquidity available on the top $m$ levels of the ECN order book, and the cost of hedging a fraction $\epsilon_{t,\text{hedge}} q_t$ of the current inventory, for a vector of values of $\epsilon_{t,\text{hedge}} \in [0,1]$. Note that time is included in the states since we are working on a finite time-horizon, hence optimal actions should be time-dependent.

The parameters characterizing the agent type are: $\eta$, $\omega$, $\gamma$, $m_{\text{target}}$ in (\ref{rewardlp}), as well as the empirical fraction of LTs of each supertype that the agent is connected to. The latter is a vector of size the number of LT supertypes, containing, for each episode, the empirical fraction of LT agents that are connected to the LP.

In this work, we assume that orders $\epsilon_{t,\text{hedge}} q_t$ sent by the LP to the ECN are market orders. It is possible to extend this framework to include the possibility of limit orders, by extending the action space $\mathcal{A}^{LP}$ to include the distance (in ticks) from the top of book to place the order at. The same observation holds for LTs.

\textbf{Liquidity Taker Agent.} At each time $t$, the LT agent decides a trade size and direction based on the observed prices. When that decision is made, the trade will always occur, independently of the LT's policy, at the best possible price among connected LPs and ECN. In this work, we consider a simplification where each LT trades a unit quantity $q_{LT}$ (possibly different across LT agents) and only decides to buy or sell $q_{LT}$, or not to trade. In this sense its action space can be assimilated to $\{1,-1,0\}$. Note that in order to have a population of LTs trading different sizes, we can define different supertypes with specific distributions for $q_{LT}$. We define $\RR^{LT}_{t+1}$ the reward obtained by the LT agent at the end of time $t$, due to actions performed during time $t$:
\begin{equation}
\begin{split}
\label{ltreward}
   & \RR^{LT}_{t+1} := \omega \cdot \eta \cdot \Delta PnL_{t+1}^\gamma - (1-\omega) \cdot \Delta \mathcal{Q}_{t+1}(q^a, q^b), \hspace{5mm} \eta>0,\hspace{2mm} \omega \in [0,1],\\
   & \mathcal{Q}_{t+1}(q^a, q^b) := \frac{1}{2} \sum_{j \in \{a,b\}} \left|\widehat{q}^j_t  - q^{j} \right|, \hspace{5mm} q^a, q^b \in [0,1], \hspace{2mm} q^a+q^b \leq 1,\\
   & \widehat{q}^j_t:= \frac{1}{t+1}\sum_{k=0}^{t} 1_{\{a_t = j\}}.
   \end{split}
\end{equation}
Similar to the LP, $\RR^{LT}_{t+1}$ is a trade-off between PnL and a quantity-related term. Similar to the running average of the market share $\widehat{m}_t$, $\widehat{q}^b_t$ and $\widehat{q}^a_t$ represent the fraction of time the LT has bought and sold, which targets over the entire episode $[0,T]$ are $q^a$, $q^b$. The LT can decide not to trade, so we have $\widehat{q}^b_t+\widehat{q}^a_t \leq 1$. Since when the LT trades, it trades the fixed quantity $q_{LT}$, the total quantities traded on both sides over an episode of length $T$ are $\widehat{q}^a_T q_{LT} T$, $\widehat{q}^b_T q_{LT} T$, hence $\widehat{q}^b_t$, $\widehat{q}^a_t$ can alternatively be seen as trade objectives. When the weight $0<\omega<1$ and $\eta$ is high, we can interpret the utility (\ref{ltreward}) as a form of optimal execution, since we wish to reach trade objectives $q^a$, $q^b$ while maximizing PnL. We will call LTs characterized by $\omega=0$ and $\omega=1$ \textbf{flow LTs} and \textbf{PnL LTs}. 

We illustrate this reward formulation in figure \ref{fig:LT toy example}. We fix the shape of the ECN mid-price curve $P_t$ in blue, train a single shared policy network $\pi^{LT}$ by randomizing the PnL weight $\omega$ at the beginning of the episode, conditioning $\pi^{LT} \equiv \pi^{LT}(\cdot|\cdot, \omega)$ on $\omega$, and look at how the LT's trading behavior varies as a function of $\omega$. We see that as $\omega$ increases, the agent switches from achieving its $q^b=75\%$, $q^a=25\%$ bid-ask targets to maximizing PnL (buy low, sell high). When $\omega=0$, the agent exactly achieves its targets but does not necessarily trade at the most cost-efficient points in time. As $\omega$ increases, the agent tries to match its targets while trading in the best regions from a PnL point of view. As $\omega$ approaches one, the targets will stop playing any role in the agent's behavior.

In order to further show the relevance of our LT formulation, we show in figure \ref{PnlINVflow} the impact on the LP's flow response curve $\epsilon \to \mathcal{F}(\epsilon)$ of increasing the number of PnL driven LTs characterized by $\omega=1$. The definition of the curve $\epsilon \to \mathcal{F}(\epsilon)$ is given in section \ref{secemergent} and captures the trading volume received as a function of price. We see that such LTs introduce convexity in this curve, due to their eagerness to trade when the prices offered by the LP become particularly attractive. Such convexity is observed in HSBC EURUSD data in \cite{gueant2021_2}, which shows the ability of our LT formulation to generate interesting features of the market.

The LT's states $s^{(i,LT)}_t$ contain the reference mid-price $P_t$, the LP's net inventory $q_t$, the fraction of time elapsed $\frac{t}{T}$, its running trade fractions $\widehat{q}^a_t$, $\widehat{q}^b_t$, and the cost of buying and selling $q_{LT}$ from the set of LPs and ECNs it is connected to. 

The parameters characterizing the agent type are: $\eta$, $\omega$, $\gamma$, $q^a$, $q^b$, $q_{LT}$ in (\ref{ltreward}) and the fraction of LPs of each supertype that the agent is connected to.

\begin{figure}[ht]
    \centering
    \begin{subfigure}[b]{0.49\textwidth}
        \centering
        \includegraphics[width=\textwidth]{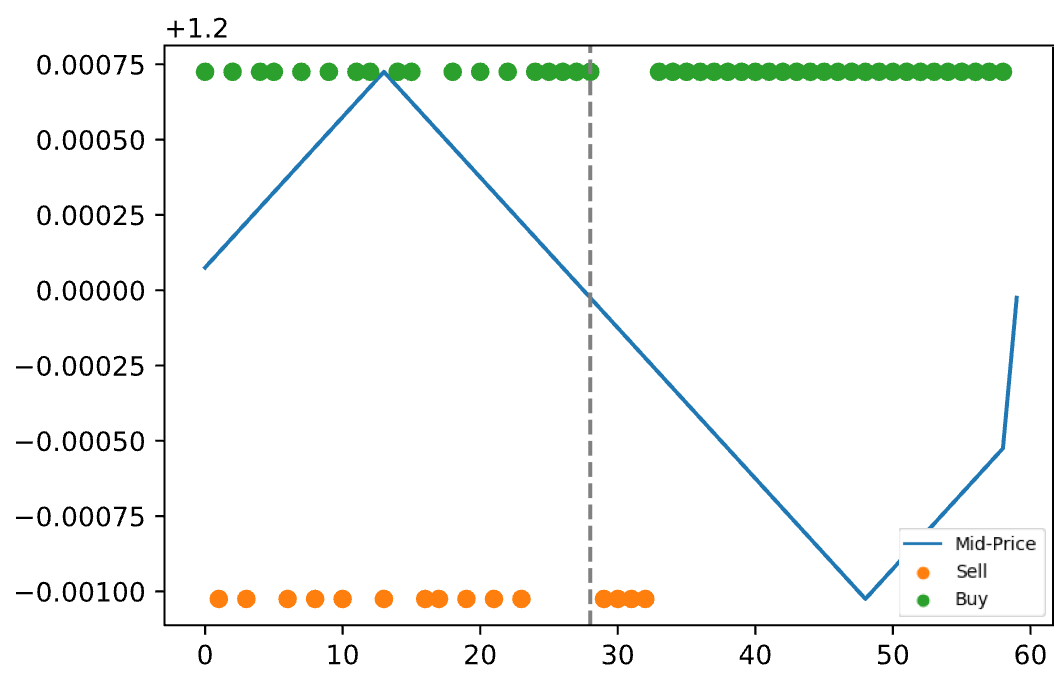}
        \caption{$\omega=0.00$}
        \label{fig:w=0.0}
    \end{subfigure}
    \hfill
    \begin{subfigure}[b]{0.49\textwidth}
        \centering
        \includegraphics[width=\textwidth]{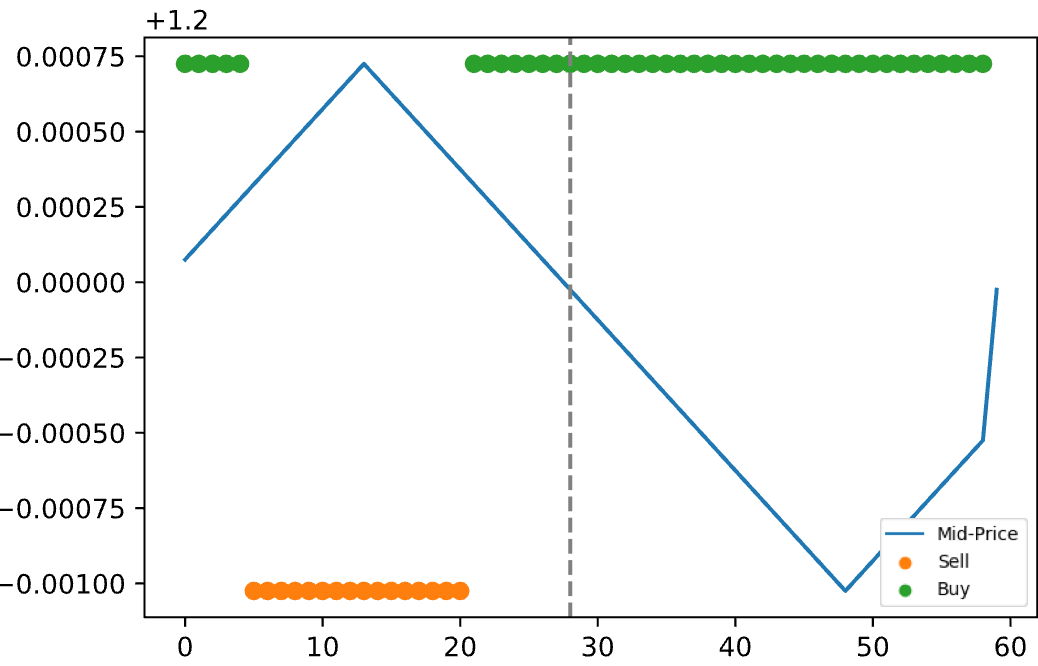}
        \caption{$\omega=0.25$}
        \label{fig:w=0.25}
    \end{subfigure}
    \hfill
    \begin{subfigure}[b]{0.49\textwidth}
        \centering
        \includegraphics[width=\textwidth]{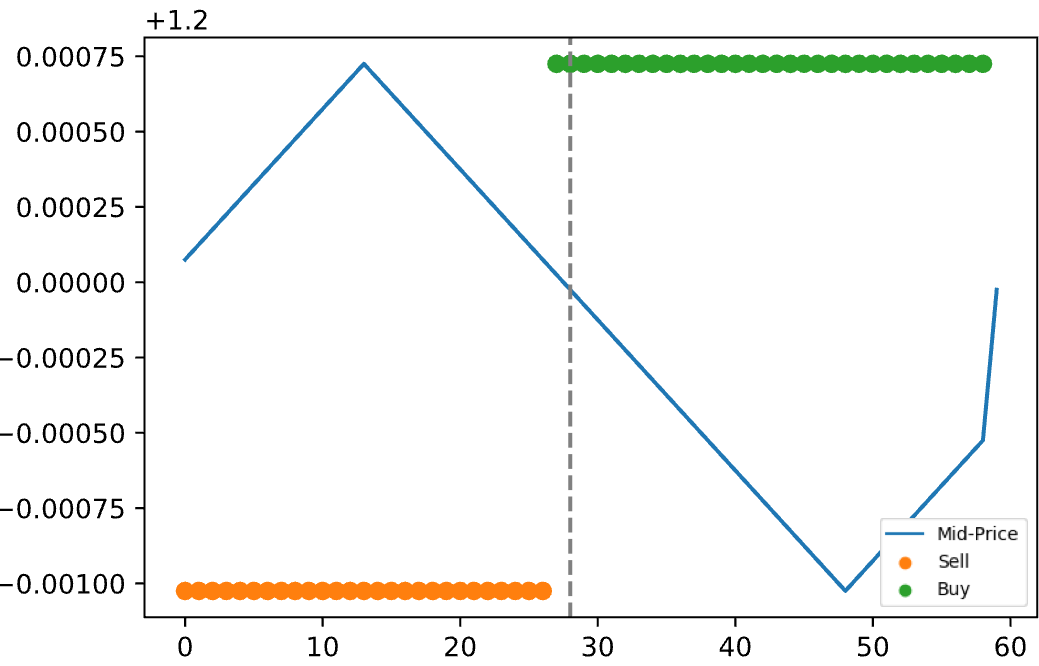}
        \caption{$\omega=0.75$}
        \label{fig:w=0.75}
    \end{subfigure}
    \hfill
    \begin{subfigure}[b]{0.49\textwidth}
        \centering
        \includegraphics[width=\textwidth]{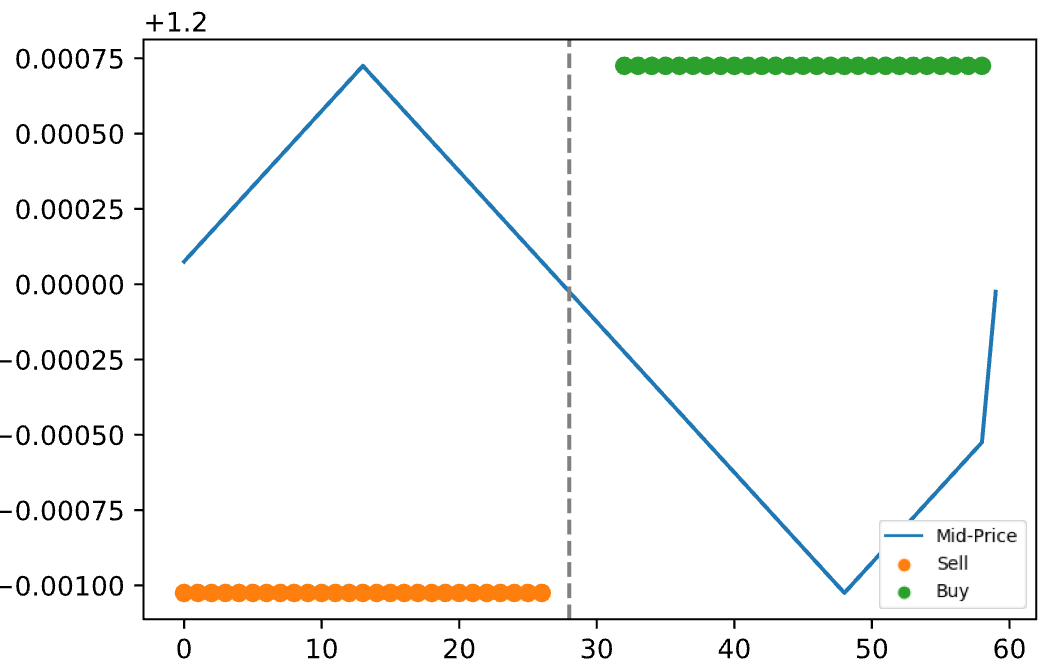}
        \caption{$\omega=1.0$}
        \label{fig:w=1.0}
    \end{subfigure}
    \caption{Spectrum of learnt actions by the shared LT agent policy $\pi^{LT}$, as a function of PnL weight $\omega$ (from left to right, from top to bottom: 0, 0.25, 0.75, 1). Mid-price as a function of time (blue), LT buy actions (green), sell actions (orange). Quantity targets $q^a=25\%$, $q^b=75\%$. The agent gradually shifts from exactly achieving its quantity targets to maximizing PnL (buy low, sell high).}
    \label{fig:LT toy example}
\end{figure}

\begin{figure}[ht]
  \centering
  \centerline{\includegraphics[width=\textwidth]{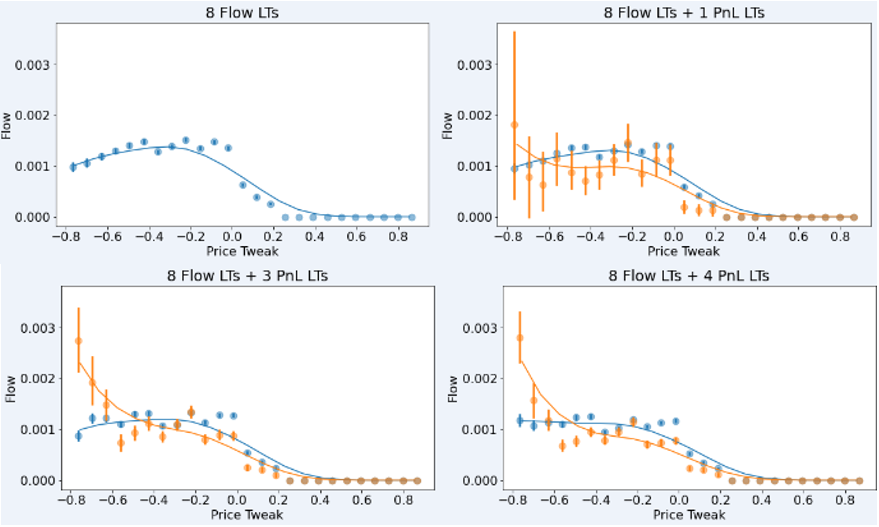}}
  \caption{Flow response curve $\epsilon \to \mathcal{F}(\epsilon)$. Impact on the flow received by the LP of increasing the number of LTs with PnL weight $\omega=1$ (from left to right, from top to bottom: 0, 1, 3, 4). Flow coming from flow LT's $\omega=0$ (blue) and PnL LT's $\omega=1$ (orange). The flow received from each agent class (flow, PnL) is normalized by the number of agents in that class, so that all curves are comparable and can be interpreted as the average flow received from a typical agent of that class.}
  \label{PnlINVflow}
\end{figure}

\begin{table}[!ht]
\caption{Summary of the RL formulation of market agents.}
\label{tabagents}
\begin{center}
\begin{scriptsize}
\begin{tabular}{ccc}
 \makecell{Agent class} & \makecell{Liquidity Provider} & \makecell{Liquidity Taker}  \\ 
\hline
States $s_t$ & \makecell{ reference price $P_t$, net inventory $q_t$, \\ time $\frac{t}{T}$, market share $\widehat{m}_{t}$, \\ ECN liquidity for top $m$ levels,\\ cost of hedging $\epsilon q_t$ for various $\epsilon$.}  &  \makecell{ reference price $P_t$, net inventory $q_t$, \\ time $\frac{t}{T}$, running quantity traded $\widehat{q}^a_t$, $\widehat{q}^b_t$, \\ cost of trading $q_{LT}$ from the best LP or ECN.} \\ \\

Type $\lambda$ & \makecell{$\gamma$, $\eta$, $\omega$, $m_{\text{target}}$, \\ fraction of LTs of each supertype it is connected to\\ (vector of size the number of LT supertypes).} & \makecell{$\gamma$, $\eta$, $\omega$, $q^a$, $q^b$, $q_{LT}$, \\ 
fraction of LPs of each supertype it is connected to\\ (vector of size the number of LP supertypes).} \\ \\

Actions $a_t$ &  \makecell{Pricing parameters $\epsilon_{\text{spread}} \geq -1$, $\epsilon_{\text{skew}} \in \R$, \\ and hedge fraction $\epsilon_{\text{hedge}} \in [0,1]$.} & \makecell{$\{1,-1,0\}$: buy or sell a unit \\ quantity $q_{LT}$, or no trade.}\\ \\

Rewards $\RR_{t+1}$ & \makecell{$\omega \cdot \eta \cdot \Delta PnL_{t+1}^\gamma - (1-\omega) \cdot \Delta \mathcal{M}_{t+1}(m_{\text{target}})$}  &  \makecell{$\omega \cdot \eta \cdot \Delta PnL_{t+1}^\gamma - (1-\omega) \cdot \Delta \mathcal{Q}_{t+1}(q^a, q^b)$}\\ 
 \hline
\end{tabular}\end{scriptsize}\end{center}
\end{table}

\subsection{ECN model}
\label{sececn}

\subsubsection{Vanilla model}
\label{sececnvanilla}

The ECN is not the main focus of our study, rather it is the OTC interactions between LPs and LTs. We however need an ECN engine so as to provide a reference price to market participants, and allow LPs to hedge a fraction of their inventory or LTs to trade if need be. This ECN plays the same role as the open market of \cite{Bank2018-ip}, or the inter-dealer segment of \cite{gueant2021_2}. We code an ECN engine that supports market, limit and cancel orders on the bid and ask sides, in particular it keeps track of individual orders and their originators, with a classical first-in-first-out mechanism.

Our two desiderata for the ECN is that i) LP and LT agents can impact the ECN limit order book when sending orders to it, and ii) in the absence of LP and LT agents orders, the ECN evolves realistically over the input RL simulation timestep $dt$, in particular its volume remains stable over time and does not explode nor vanish. The former is a structural requirement that follows from the ability of agents to send orders as explained in section \ref{secrlagents}. We solve the latter by equipping the ECN with a so-called ECN agent in charge of sending orders to it (and to it only) at every timestep. In this section we proceed to specifying the nature of this ECN agent, i.e. how it constructs a list of orders at each timestep of the simulation. We emphasize that although the ECN agent will use the model in (\ref{snapvarr})-(\ref{dlob}) to build a list of orders, the ECN dynamics are not given by these equations, rather the limit order book evolves solely as a consequence of orders sent by LPs, LTs, and the ECN agent. This can be seen as a hybrid impact model between RL-agents and the "rest of the world", where the LT trade sizes $q_{LT}$ and the ECN agent trade sizes drive the extent to which RL-based agents and the ECN agent impact the ECN.

We equip the ECN with a fixed price grid with a fixed granularity $\xi$, also called tick size. In the case of eurodollar on EBS, $\xi=0.5$ bp at the time of the writing. This price grid is upper and lower bounded by $P_{max}$, $P_{min}$ and we let $K_{max}:=\left \lfloor \xi^{-1}(P_{max}-P_{min}) \right \rfloor$ be the total number of ticks. We call \textit{limit order book snapshot} the associated pair of vectors of size $K_{max}$ representing the volumes available for each price, on each side. We denote $P_t$, $P^a_t$, $P^b_t$ the ECN mid-price, best ask and best bid prices, i.e. $2P_t = P^a_t+P^b_t$.

The training data $\mathcal{D}$ that we use to calibrate the ECN agent is level two limit order book data, namely a dataset containing the limit order book snapshot at different times. Such data doesn't contain volumes for all prices from $P_{min}$ to $P_{max}$, typically it contains volume and price data for the top $m$ non-empty levels on each side, we will typically use $m=5$. Our ECN agent model will therefore consider the associated volumes $(V_{i,t})_{i \in [1,2m]}$. In order to always have a snapshot for prices from $P_{min}$ to $P_{max}$, we extrapolate volumes for levels $k$ further than $m$ with exponential decay $e^{-\alpha k }$, where $\alpha>0$ is fitted to the top $m$ levels using our dataset $\mathcal{D}$. To further reduce the dimensionality of the model, we assume that the price difference between two consecutive non-empty levels is always one tick, except for the difference between the best ask and best bid - also called market spread - which can be greater, typically 1 to 3 ticks. This is typical of liquid markets. With these assumptions, in order to construct the order book at time $t+1$ from the order book at time $t$, it is enough to know the mid price change $\xi^{-1}(P_{t+1}-P_t)$, the new market spread $\xi^{-1}(P^a_{t+1}-P^b_{t+1})$, and the change in volume $(V_{i,t+1} - V_{i,t})_{i \in [1,2m]}$ for the top $m$ levels on both sides. The vanilla version of the model that we consider in this section will, in short, assume that the vectors of these $2m+2$ quantities at two distinct times are independent and identically distributed, although we account for correlation among the $2m+2$ random variables at a fixed time. Precisely, the vectors $\mathcal{S}_n$, $\mathcal{S}_{m}$ in (\ref{snapvarr}) are independent for $n \neq m$. We revisit this independence assumption in the neural extension of the vanilla model in section \ref{sececnneural}, where we will allow the order book variation to depend on its history.

The mechanism of the ECN agent can be broken down into the following steps. First, at the beginning of the simulation, it creates an initial limit order book snapshot. Then, at each timestep, i) it generates a limit order book snapshot variation, precisely two vectors of size $K_{max}$, one for each side, containing volume variations associated to the price grid; ii) it creates a sequence of orders associated to the latter snapshot variation, in the sense that sending those orders to the ECN would reproduce exactly that variation, \textit{in the absence of other orders}.

\textbf{Initial limit order book snapshot.} We fix the initial mid-price $P_0$ and model the initial limit order book snapshot as a vector $\mathcal{S}_0$ of size $2m+1$ consisting of log-quantities $(\ln V_{i,0})_{i \in [1,2m]}$ plus the market spread, in ticks:
$$
\mathcal{S}_0 := \left((\ln V_{i,0})_{i \in [1,2m]}, \xi^{-1}(P^a_0-P^b_0) \right).
$$
From our training dataset $\mathcal{D}$, we construct samples of $\mathcal{S}_0$ that we fit to a multivariate Gaussian mixture distribution with typically $5$ components, so as to capture higher moments and cross-moments between log-volumes and market spread. We then sample from this fitted distribution to generate the initial snapshot.

\textbf{Breaking down the limit order book snapshot variation into smaller orders of typical size.} For now, assume that we have a way to draw samples of order book snapshot variation over a timestep. Remember that such variation consist of precisely two vectors of size $K_{max}$, one for each side, containing volume variations associated to the price grid. Given such variation, we constitute a list of market, limit, cancel orders such that sending those orders to the ECN reproduces exactly that variation, in the absence of other orders. In this sense it is a purely mechanical process. We describe below such transformation for the ask side (the bid side being similar).

We first build a list of meta-orders which we then break up into smaller pieces of typical size. Let $(\Delta V_i)_{i=1..K_{max}}$ be the vector of signed volume variation on the ask size, where entries of the latter vector correspond to increasing prices. If the first non-zero entry of $\Delta V$ is negative, we create a meta market order, and continue incrementing the size of the order until we reach a nonnegative variation. The subsequent meta orders will be set as limit or cancellation depending if the variation is positive or negative. Otherwise, if the first non-zero entry is positive, we simply skip the market order step. Precisely, let $i_0 := \inf\{i: \Delta V_i \neq 0\}$. If $\Delta V_{i_0}>0$, then we create a list of size $K_{max}-i_0$ where the $k^{th}$ element is a meta limit or cancellation order of size $|\Delta V_{k}|$, depending on the sign of $\Delta V_{k}$. On the other hand if $\Delta V_{i_0}<0$, let $i_1 := \inf\{i>i_0: \Delta V_i \geq 0\}$, then we create a meta market order of size $\sum_{k=i_0}^{i_1-1} |\Delta V_{k}|$, and create meta limit and cancellation orders using $(\Delta V_{k})_{i \geq i_1}$ as specified previously. Then, from our training data $\mathcal{D}$, we compute the order book volume differences for the smallest time interval available, which gives us an empirical distribution of typical individual orders, and then break down each large meta order into a sequence of small orders of random size sampled from this empirical distribution.

\textbf{Limit order book snapshot variation.} We now proceed to specifying our model for generating order book snapshot variation over a timestep from time $n$ to time $n+1$, from which the ECN agent constructs its orders. Similarly to the initial snapshot, we consider the vector $\mathcal{S}_{n+1}$ of dimension $2m+2$, where the last 2 entries are the market spread at time $n+1$ (in ticks) and the mid price change (in ticks) between times $n$ and $n+1$. The first $2m$ entries consist of relative volume variations $\delta_{i,n}$ at book level $i$ related to the dynamics (\ref{dlob}). The knowledge of $\mathcal{S}_{n+1}$ is enough to construct the limit order book at time $n+1$ from the order book at time $n$.
\begin{align}
\label{snapvarr}
    \mathcal{S}_{n+1} := \left((\delta_{i,n})_{i \in [1,2m]}, \xi^{-1}(P^a_{n+1}-P^b_{n+1}), \xi^{-1}(P_{n+1}-P_{n}) \right), \hspace{3mm} n \geq 0.
\end{align}

From our training dataset $\mathcal{D}$, we construct samples of $\mathcal{S}_{n+1}$ that we fit to a multivariate Gaussian mixture distribution with typically $5$ components, so as to capture the correlation structure and higher order moments/cross-moments between volume variations, market spread and mid price change. We then sample from this fitted distribution independently at each timestep to generate the snapshot variation.

We consider the following discrete-time dynamics for the volume $V_{i,n}$ at time $n$ and level $i$:
\begin{align}
\label{dlob}
\begin{split}
&\Delta V_{i,n}:=V_{i,n+1}-V_{i,n} = -\delta^{-}_{i,n} V_{i,n} +\delta^{+}_{i,n},\\
&\delta^{+}_{i,n}:=\max(\delta_{i,n},0) \geq 0, \hspace{4mm} \delta^{-}_{i,n}:=\max(-\delta_{i,n},0) \in [0,1].
\end{split}
\end{align}

We treat positive and negative volume variations differently since when a volume variation is negative, it cannot be more than the current volume at that level, so it is coherent to express the variation as a multiple of the latter quantity. On the other hand, when it is positive, it is better to express it in absolute terms. If it were expressed multiplicatively, a close to empty book would remain empty, and a book with more liquidity would increase its liquidity exponentially: this hybrid formulation allows us to obtain a provably stable ECN evolution over arbitrary time horizons, i.e. $\lim_{n \to \infty} \EE [V_{i,n}] < + \infty$, as will be seen in theorem \ref{dam} and proposition \ref{corvar}. This would not be true without this hybrid "multiplicative-absolute" $\delta^+ - \delta^-$ formulation, for example in the purely multiplicative or absolute cases $\Delta V= (\delta^{+}-\delta^{-}) V$, $\Delta V= \delta^{+}-\delta^{-}$ (unless in the very specific case $\EE[\delta^{+}_{i,n}]=\EE[\delta^{-}_{i,n}]$, which is usually not true when calibrating to empirical data). We illustrate these dynamics in figure \ref{ecnexample}.

\begin{figure}[ht]
  \centering
    \begin{subfigure}[b]{0.32 \textwidth}
        \centering
        \includegraphics[width=\textwidth]{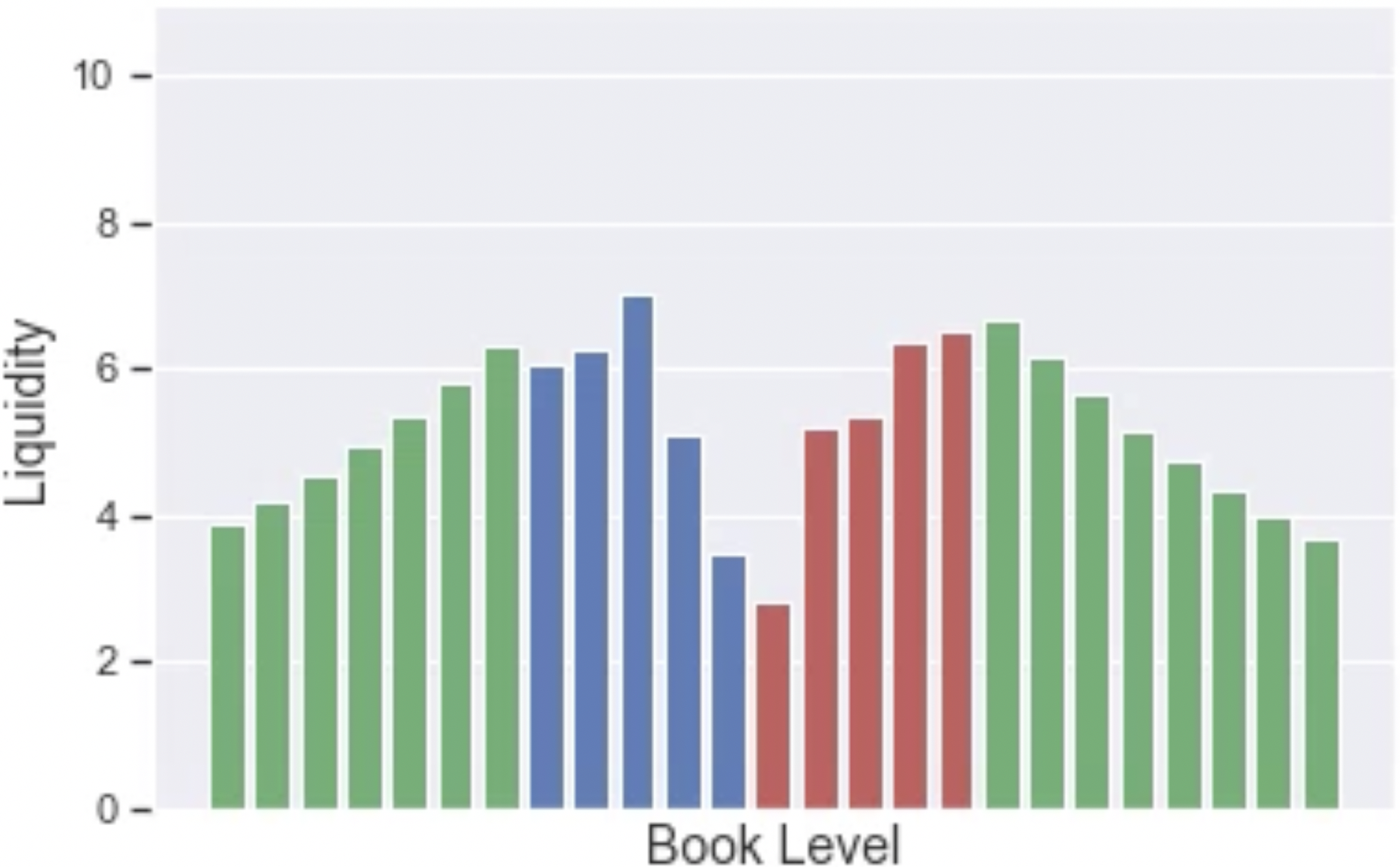}
    \end{subfigure}
    \hfill
    \begin{subfigure}[b]{0.32 \textwidth}
        \centering
        \includegraphics[width=\textwidth]{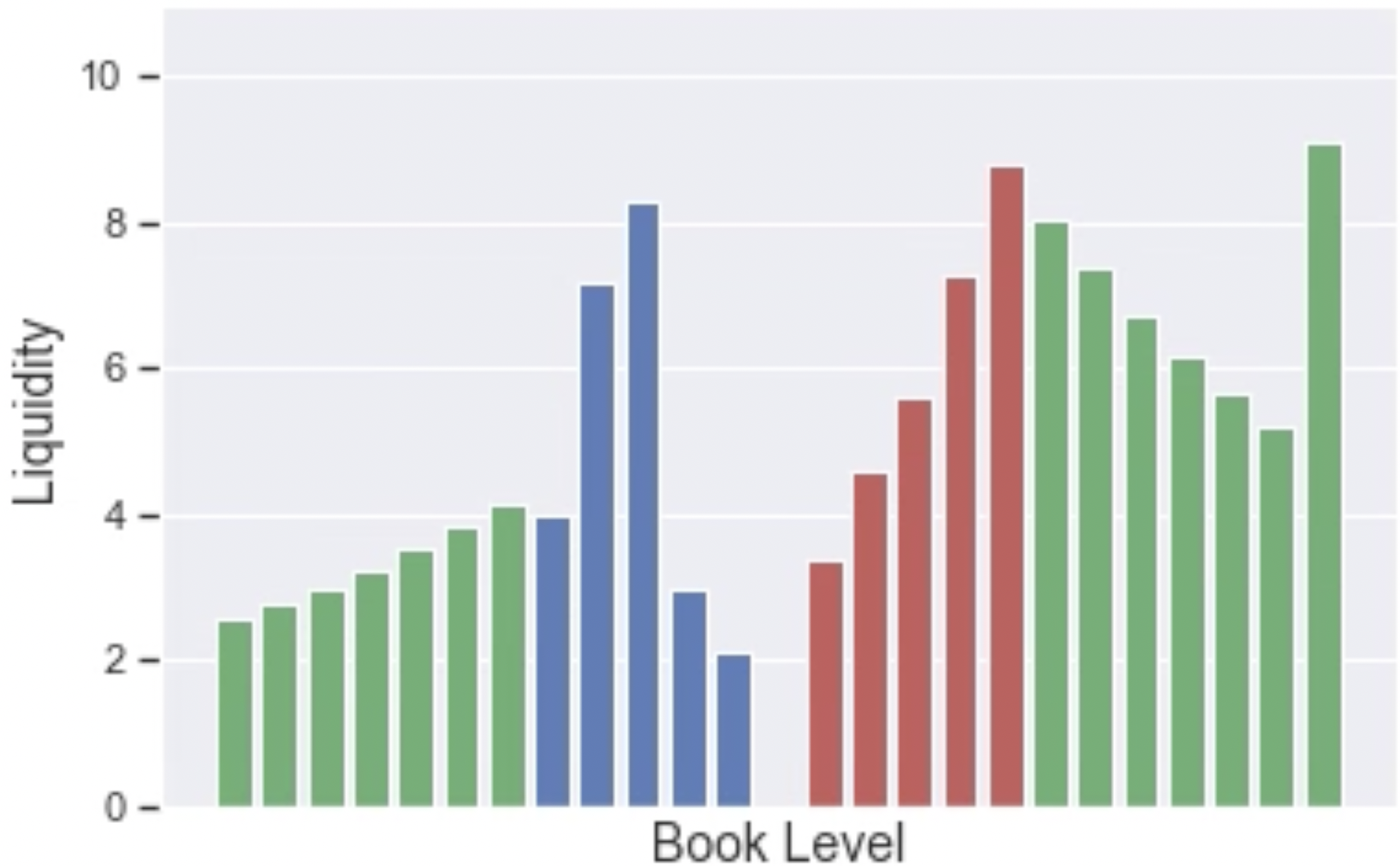}
    \end{subfigure}
    \begin{subfigure}[b]{0.32 \textwidth}
        \centering
        \includegraphics[width=\textwidth]{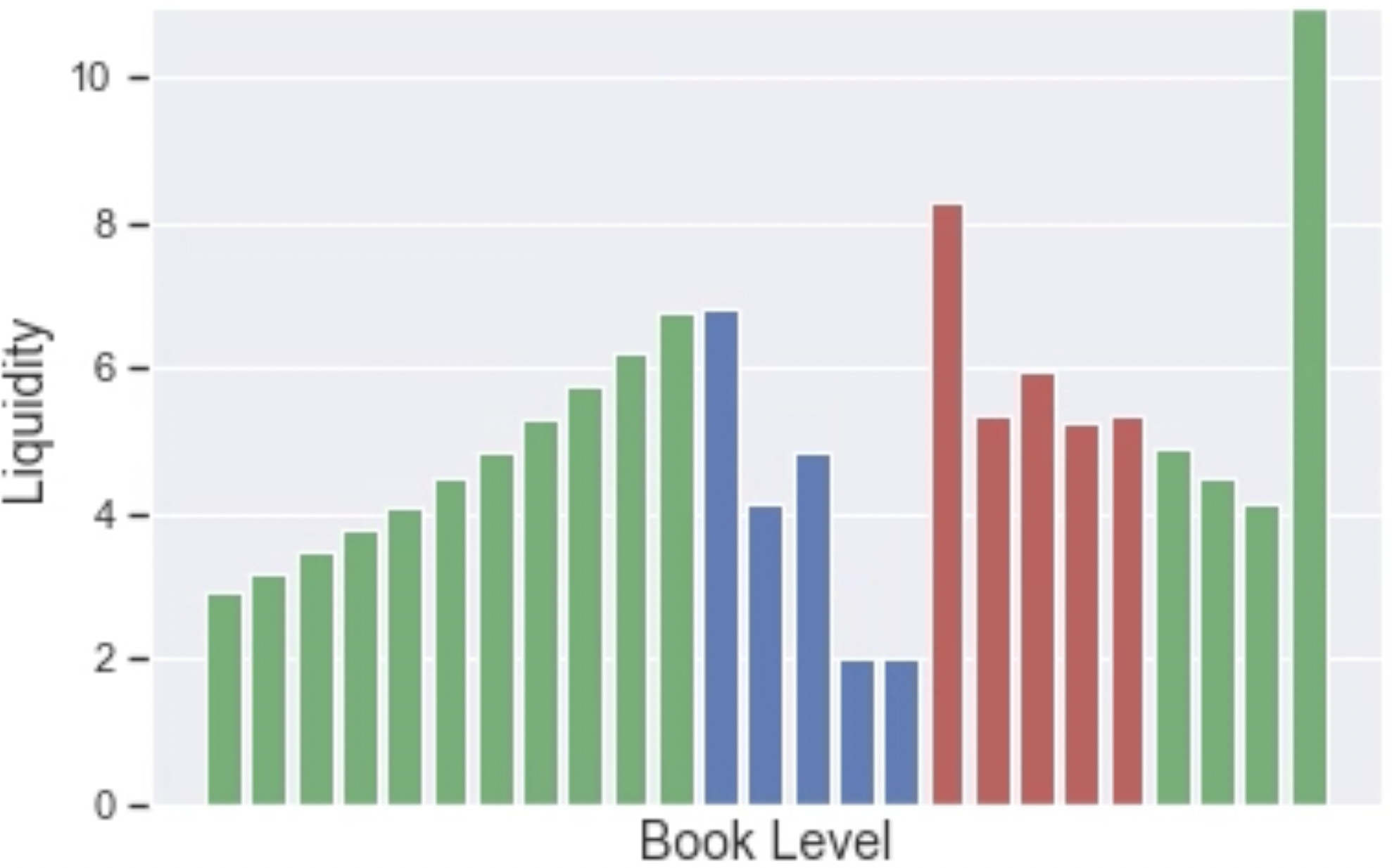}
    \end{subfigure}
    \caption{Example of ECN order book evolution associated to the dynamics (\ref{snapvarr})-(\ref{dlob}) over time an episode, in the absence of LP and LT orders (from left to right). Top 5 levels on the ask (red), bid (blue), levels in green are extrapolated. The middle plot displays a market spread of 2 ticks, contrary to 1 tick in the other cases. In accordance with the theory, running these dynamics over arbitrary time horizons keeps the order book stable, in the sense that volumes do not vanish nor explode.}
  \label{ecnexample}
\end{figure}

Our hybrid formulation is similar in spirit to that of \cite{Cont:2021}, which models the order book volume by means of a SDE. Their drift contains a term $-\alpha_x V_{x,t} + \beta_x$ with $\alpha, \beta \geq 0$ deterministic and $x$ the (continuous) book level. Their dynamics assume multiplicative volatility for the volume $\sigma^- V_{x,t}$, which they justify by \textit{"the multiplicative nature of the noise accounts for the high-frequency cancellations associated with HFT orders"}. Interestingly, we will see that the continuous-time diffusion approximation limit of our discrete-time dynamics (\ref{dlob}) generalizes their volatility term in that when the volatility $\sigma^{+}$ of positive volume variations $\delta^{+}$ is zero, we recover multiplicative book volume volatility $\sigma^- V_{i,t}$ generated by negative volume variations (remark \ref{contgen}). When $\sigma^{+}$ is not zero, the book volume variance will be shown to be a quadratic polynomial in $V_{i,t} - \mu_i^\infty$, where $\mu_i^\infty:=\lim_{t \to \infty} \EE [V_{i,t}]$ is the long-range mean. The deviation of the book volume from its long-range mean can be seen as a proxy for market activity, hence we can define regions for the book volume where its variance is higher or lower than the (constant) long-range variance. This self-exciting behavior is studied in proposition \ref{damr}. \cite{Cont:2021} further considers convection and diffusion terms in their SDE drift, which relate to correlation across their infinite set of book levels. In our case we instead model a finite number of book levels $i$ which are correlated through the joint distributions of $(\delta^{+}_{i}, \delta^{-}_{i})_{i \in [1,2m]}$. This can be seen as a discrete counterpart of the convection and diffusion terms.

In theorem \ref{dam} and corollary \ref{dam1}, we look at the continuous-time limit of our discrete-time dynamics (\ref{dlob}). As mentioned previously, the covariance structure $Q$ generalizes that of \cite{Cont:2021}, which corresponds to the specific case $\sigma^+=0$. In the general case, we see that $Q$ is modulated by the position of the volume $V_{i,t}$ relative to its long-range mean $\mu_i^\infty$. This aspect is further analyzed in proposition \ref{damr}. Interestingly, looking at the continuous-time limit allows us to quantify the impact of the model parameters $\sigma_\pm$, $\mu_\pm$, $\rho$ on the nature of the fluctuations of the order book volume, which is compactly captured by the polynomial $Q$.

\begin{theorem}
\label{dam}
Let $((\delta^{+}_{i,n},\delta^{-}_{i,n})_{i \in [1,2m]})_{n \geq 0}$  be a sequence of independent and identically distributed vectors of random variables taking value in $(\R \times [0,1])^{2m}$ where $\delta^{\pm}_{i,n}$ has mean $\mu_{i}^\pm$, standard deviation $\sigma_{i}^\pm$, and such that $\mu_{i}^->0$. Let
$\corr(\delta_{i,n}^-,\delta_{j,n}^+)=:\rho_{ij}$,
$\corr(\delta_{i,n}^+,\delta_{j,n}^+)=:\rho_{ij}^+$, $\corr(\delta_{i,n}^-,\delta_{j,n}^-)=:\rho_{ij}^-$. Assume that $V_i$ satisfies the discrete-time recursion:
$$
V_{i, n+1}=(1-\delta^-_{i,n})V_{i,n}+\delta^+_{i,n} \hspace{3mm} a.s., \hspace{3mm} V_{i,0} \in \R_+, \hspace{3mm} n \geq 0.
$$

Denote $V^\epsilon_{i, n+1}$ the process associated to the scaling $\delta^\pm_{i,n,\epsilon}:= \epsilon \mu_{i}^\pm + \sqrt{\epsilon}(\delta^\pm_{i,n} - \mu_{i}^\pm)$, and $V^\epsilon_{i, t} := V^\epsilon_{i, \lfloor \epsilon^{-1}t \rfloor}$. Under the assumption that $\delta^{+}_{i,n}$ are bounded in $L^{2+\eta}$ for some $\eta>0$, $V^\epsilon$ converges weakly in the Skorokhod topology to a multivariate Ornstein–Uhlenbeck process $V^*$ with quadratic covariance structure:
$$
dV^{*}_{i,t}=\mu_{i}^- \left(\mu_{i}^\infty - V^{*}_{i,t}\right)dt + \left[q(V^{*}_{t}-\mu^\infty)dB_t\right]_i,
$$
where $B$ is a standard vector Brownian motion and $q$ is a square-root of the matrix $Q$, $q(V^{*}_{t}-\mu^\infty)q(V^{*}_{t}-\mu^\infty)^T=Q(V^{*}_{t}-\mu^\infty)$, with $[Q(V^{*}_{t}-\mu^\infty)]_{ij} =: \widetilde{Q}_{ij}(V^{*}_{i,t}-\mu_{i}^\infty,V^{*}_{j,t}-\mu_{j}^\infty)$, and $\widetilde{Q}_{ij}$ the quadratic polynomial:
\begin{align*}
\begin{split}
\widetilde{Q}_{ij}(x,y):=&\sigma_{ij}^{\infty} + 
\rho_{ij}^- \sigma_{i}^- \sigma_{j}^- xy
 +  \sigma_{i}^- \left( \mu_{j}^\infty \rho_{ij}^- \sigma_{j}^- - \rho_{ij}\sigma_{j}^+ \right) x
 +  \sigma_{j}^- \left( \mu_{i}^\infty \rho_{ij}^- \sigma_{i}^- - \rho_{ji}\sigma_{i}^+ \right) y,
\end{split}
\end{align*}

where:
\begin{align*}
\begin{split}
&\sigma_{ij}^{\infty}:=\rho_{ij}^+ \sigma_{i}^+ \sigma_{j}^+ + \mu_{i}^\infty \mu_{j}^\infty \rho_{ij}^- \sigma_{i}^- \sigma_{j}^- -\mu_{i}^\infty \rho_{ij} \sigma_{j}^+ \sigma_{i}^- -\mu_{j}^\infty \rho_{ji} \sigma_{i}^+ \sigma_{j}^-, \hspace{4mm} \mu_{i}^\infty := \frac{\mu_{i}^+}{\mu_{i}^-}.
\end{split}
\end{align*}
\end{theorem}

The long-range moments in proposition \ref{corvar} follow easily by direct computation\footnote{Note that the term $\mu^-_i\mu^-_j$ disappears from the covariance denominator when passing to the continuous-time limit, since $n$ is of order $\epsilon^{-1}$, $\mu_i^\pm$ or order $\epsilon$ and $\sigma^\pm_i$ of order $\sqrt{\epsilon}$, hence $\sigma^-_i\sigma^-_j$, $\mu^-_i$, $\mu^-_j$ are all of order $\epsilon$, but $\mu^-_i\mu^-_j$ is of order $\epsilon^2$ and hence vanishes.}.
\begin{proposition}
\label{corvar}
We have for every $i$, $j \in [1,2m]$:
$$
\lim_{n \to \infty} \EE[V_{i,n}] = \lim_{t \to \infty} \EE[V^{*}_{i,t}] = \mu_{i}^\infty,
$$
$$
\lim_{n \to \infty} \cov[V_{i,n},V_{j,n}] = \frac{ \sigma_{ij}^{\infty}}{\mu^-_i+\mu^-_j-\mu^-_i\mu^-_j-\rho^-_{ij}\sigma_i^-\sigma_j^-},
$$
$$
\lim_{t \to \infty} \cov[V^{*}_{i,t} V^{*}_{j,t}] = \frac{ \sigma_{ij}^{\infty }}{\mu^-_i+\mu^-_j-\rho^-_{ij}\sigma_i^-\sigma_j^-}.
$$
\end{proposition}

In particular when looking at a single level $i$ in theorem \ref{dam}, we get immediately the corresponding univariate version presented in corollary \ref{dam1}.

\begin{corollary}
\label{dam1}
Under the same notations as theorem \ref{dam}, the rescaled volume $V^\epsilon$ at a single book level converges weakly in the Skorokhod topology to a Ornstein–Uhlenbeck process $V^*$ with quadratic variance:
$$
dV^{*}_{t}=\mu^- \left(\mu^\infty - V^{*}_{t}\right)dt + \sqrt{Q(V^{*}_{t} - \mu^\infty)}dB_t,
$$
where $B$ is a standard Brownian motion and $Q$ the quadratic polynomial:
\begin{align*}
\begin{split}
Q(x):=&\sigma^{\infty 2} + 
\sigma^{-2} x^2
 +  2\sigma^- \left( \mu^\infty \sigma^- - \rho\sigma^+ \right) x,
\end{split}
\end{align*}

where:
\begin{align*}
\begin{split}
&\sigma^{\infty 2}:= \sigma^{+2} + \mu^\infty \sigma^- \left( \mu^{\infty } \sigma^{-} - 2\rho \sigma^+\right),\hspace{4mm}  \mu^\infty := \frac{\mu^+}{\mu^-}.
\end{split}
\end{align*}
\end{corollary}

\begin{remark}
\label{contgen}
If $\sigma^+=0$ in corollary \ref{dam1}, $\sqrt{Q(V^{*}_{t} - \mu^\infty)}=\sigma^-V^{*}_{t}$ as in \cite{Cont:2021} (section 1, equation 1.2).
\end{remark}

The equivalent vanilla Ornstein–Uhlenbeck process would have constant volatility $\sigma^{\infty}$. In our case, the quadratic variation $Q$ is a second order polynomial in $V^{*}_{t} - \mu^\infty$. The latter quantity can be seen as a proxy for market activity, since when the market is calm, the volume is close to its long-range mean, whereas when it is volatile, it is perturbed away from it. Depending on the parameters $\sigma_\pm$, $\mu_\pm$, $\rho$ of corollary \ref{dam1}, we can define in proposition \ref{damr} regimes where the variance is higher, or lower than the equivalent "flat" Ornstein–Uhlenbeck variance $\sigma^{\infty 2}$, which is also the long-range variance of $V^*$ by proposition \ref{corvar}.

\begin{proposition}
\label{damr}
Under the notations of corollary \ref{dam1}, let $\mathcal{V}(x):=Q(x)-Q(0)$ be the variance impact due to the polynomial nature of $Q$. We define the self-exciting and self-inhibiting regimes of $V^*$ as the regions $\mathcal{V}>0$, $\mathcal{V}<0$. If $\sigma^-=0$, $V^*$ is not self-exciting nor self-inhibiting since $\mathcal{V} \equiv 0$. If $\sigma^->0$, the self-exciting regimes of $V^*$ are $(-\infty,\gamma_*)$, $(\gamma^*,+\infty)$, and its self-inhibiting regime is $(\gamma_*,\gamma^*)$, with:
$$
\gamma^* = \max \left[\mu^\infty, 2\frac{\sigma^+}{\sigma^-} \rho - \mu^\infty \right], \hspace{5mm} \gamma_* = \min \left[\mu^\infty, 2\frac{\sigma_+}{\sigma_-} \rho - \mu^\infty \right].
$$
In particular, $V^*$ is never self-inhibiting if and only if $\sigma^+ \mu^-\rho = \sigma^-\mu^+$. The parameters $\mu^\pm$, $\sigma^\pm$ and $\rho$ impact the regime change only through the ratios $\frac{\sigma^+}{\sigma^-} \rho$ and $\frac{\mu^+}{\mu^-}$. Further, we have $\partial_\rho \mathcal{V}(x) >0$ if and only if $x<0$.
\end{proposition}

\subsubsection{Neural extension of the vanilla model}
\label{sececnneural}

In this section we show how the vanilla model in section \ref{sececnvanilla} can be extended. The main drawback of this model is that the $p=2m+2=12$ dimensional snapshot variation vectors $\mathcal{S}_{t}$ in (\ref{snapvarr}) for different times $t$ are assumed to be independent, identically distributed, sampled from a multivariate Gaussian mixture with constant, non history-dependent parameters. These parameters consist of, for each component of the mixture, a nonnegative weight, as well as the mean (vector of size $p$), variance (vector of size $p$) and $p \times p$ correlation matrix of a multivariate normal distribution, associated to the vector of random variables $\mathcal{S}_{t}$ in (\ref{snapvarr}).

More realistically, we would like the parameters of this Gaussian mixture to depend on the book history, namely on the history of previous $\mathcal{S}_t$'s.

We achieve this using the architecture in figure \ref{neural_ecn_architecture} similar to \cite{worldmodels}, built with a long short-term memory network (LSTM) and a mixture density network (MDN). A LSTM layer first transforms a batch $(\mathcal{S}_j)_{j\in [t-k+1,t]} \in \R^{k \cdot p}$ of historical window size $k=20$ into a latent space of dimension $\mathcal{L}=32$. This encoding is obtained by taking the hidden state of the LSTM. Using the hidden state as a projection method for the history of a vector-valued stochastic process is a well-known method and is typically used in modern game theory to represent the history of actions and states experienced by an agent in imperfect-information games \cite{Gupta2017-it}. Then, the hidden state is concatenated with the most recent book snapshot variation $\mathcal{S}_t$ and fed into a fully connected neural network with one input layer (width 32), 2 hidden layers (width 64) and multiple output layers, one per parameter type of the Gaussian mixture with $n_{mix}$ components, cf. figure \ref{neural_ecn_architecture}. This consists of, for each component, its weight $\pi \in \R_+$, mean $\mu \in \R^p$, variance $\sigma^2 \in \R^p_+$ and $p \times p$ correlation matrix $\rho$, a total of $n_{mix} \left(1+2p+\frac{p(p-1)}{2}\right)$ parameters. This is because the correlation matrix is symmetric and contains ones on the diagonal. Note that we use suitable output activation functions (exponential, tanh, softmax) to ensure that variances are non-negative, correlations are in $[-1,1]$ and mixture weights are nonnegative and sum to one. The network is trained to maximize the log-likelihood between the training data and that of the Gaussian mixture distribution.

In our case, $p=12$, so the total number of Gaussian mixture parameters to learn is dominated by the number of correlation parameters $\mathcal{O}(n_{mix} p^2)$. In order to reduce the number of parameters to learn, we consider two cases. For the first case, \textit{fixed correlation}, a correlation matrix between the $p$ random variables $\mathcal{S}_{t}$ is precomputed from our training data and used by each mixture component. Consequently, our network solely outputs component weights, means, variances, a total of $n_{mix} \cdot (1+2p)$ parameters. The training loss is evaluated using the precomputed correlation matrix. For the second case, \textit{shared correlation}, the network learns one correlation matrix that is shared across all mixture components, i.e. each component uses the same learnt correlation matrix. Note that \textit{shared correlation} includes \textit{fixed correlation} as a particular case, however distinguishing these two cases is useful to quantify the benefit of learning the correlations vs. precomputing it using the training data.

We train the network using an Adam optimizer with learning rate $10^{-3}$ up to epoch 110, then $10^{-4}$, and a minibatch size of 100. We initialize the network weights using the Xavier initialization, also known as Glorot initialization, associated to the uniform distribution. We use the same training, validation and test data when assessing performance of the vanilla model and of its neural extension. The vanilla model fit is performed using the standard expectation-maximization (EM) algorithm, using \textit{sklearn}. We do not use dropout, but use early stopping as a way to prevent overfitting. That is, we stop training when the validation loss starts increasing. In total we trained on approximately 20,000 epochs.

We present in table \ref{tabnecn} the scores (log-likelihood) of all three models. We see that the neural architecture outperforms significantly the vanilla model, and that learning correlations improves performance. 

To conduct the experiment in table \ref{tabnecn}, we used Lobster INTC data from Jan 28th 2015 \footnote{www.lobsterdata.com/tradesquotesandprices}. We considered the last two hours before close as it is highly liquid, the timestep value $dt=1s$, and split the data into training, validation and test sets of sizes 70\%, 15\%, 15\%. The data that we feed as input to our network is standardized, namely we substract the mean and divide by the standard deviation.

\begin{figure}[H]
  \centering
  \includegraphics[scale=0.33]{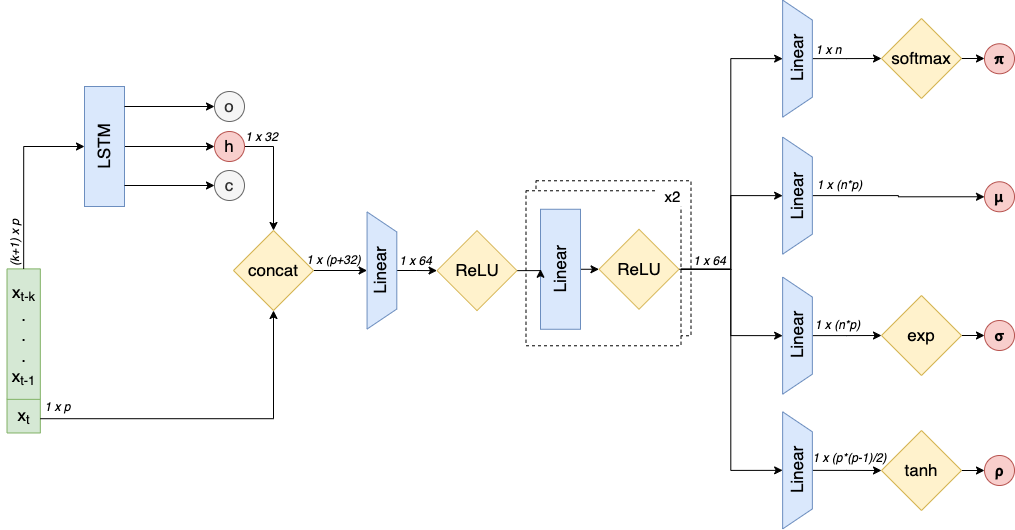}
  \caption{Architecture of the Neural ECN model. We output parameters of a multivariate Gaussian mixture distribution (component weights, means, variances and correlations) with $n_{mix}$ components fitted to the $p=2m+2=12$ dimensional vector $\mathcal{S}_t$ in (\ref{snapvarr}). An order book history $h_t:=(\mathcal{S}_j)_{j\in [t-k+1,t]}$ of size $k=20$ is fed as input to the network and first encoded into a latent space of size 32. We concatenate the latter with the most recent book information $\mathcal{S}_t$, and feed it to a fully connected neural network with 2 hidden layers of width 64.
  }
  \label{neural_ecn_architecture}
\end{figure}

\begin{table}[H]
\caption{Summary of ECN model scores (log-likelihood). Gaussian mixture with 5 components.}
\label{tabnecn}
\begin{center}
\begin{tabular}{c c c c}
\hline
Model & Train Score & Validation Score & Test Score \\ \hline
Vanilla                      & 8.89  & 6.23  & 8.94  \\
Neural w/ fixed correlation  & 18.32 & 15.79 & 15.64 \\
Neural w/ shared correlation & \textbf{23.40} & \textbf{21.90} & \textbf{21.63} \\ \hline
\end{tabular}
\end{center}
\end{table}

\section{Game theoretical analysis and convergence properties}
\label{secgt}

\subsection{Shared equilibria: convergence of shared-policy learning in the case of stationary LTs}
\label{secsharedeq}

In this section, we focus on LPs. That is, we assume that LTs do not learn, i.e. $\pi^{LT}$ is a fixed distribution and can therefore be considered as part of the environment, i.e. transition dynamics $\mathcal{T}$. We will therefore omit the dependency on $\pi^{LT}$ in our notations. Our goal is to understand what are the game theoretic implications of LP agents of different types learning under a shared policy. 

Intuitively, assume 2 players are asked to submit algorithms to play chess that will compete against each other. Starting with the white or dark pawns presents some similarities as it is chess in both cases, but also fundamental differences, hence the algorithms need to be good in all cases, whatever the type (white or dark) assigned by the random coin toss at the start of the game. The 2 players are playing a higher-level game on the space of algorithms that requires the submitted algorithms to be good in all situations. This also means we will consider games where there are "good" strategies, formalized by the concept of extended transitivity in assumption \ref{et} and proposition \ref{et2}.

 We will actually need the following definition, which is slightly more general than (\ref{vi2}) in that it allows LP agents $j \neq i$ to use a different policy $\pi_2 \in \mathcal{X}^{LP}$, where $i \neq j$, $\pi_1,\pi_2 \in \mathcal{X}^{LP}$:
\begin{align}
\label{vi}
\widetilde{V}_{\Lambda^{LP}_i}(\pi_1,\pi_2):=\EE_{\substack{\lambda^{LP}_i \sim p_{\Lambda^{LP}_i}, \hspace{1mm} a^{(i,LP)}_t \sim \pi_1(\cdot|\cdot,\lambda^{LP}_i) \\ \lambda^{LP}_j \sim p_{\Lambda^{LP}_j}, \hspace{1mm} a^{(j,LP)}_t \sim \pi_2(\cdot|\cdot,\lambda^{LP}_j)}} \left[ \sum_{t=0}^T \zeta^t \RR^{LP}(z^{(i,LP)}_t,\bm{z^{(-i,LP)}_t}, \bm{z^{(LT)}_t})\right].
\end{align}

$\widetilde{V}_{\Lambda^{LP}_i}(\pi_1,\pi_2)$ is to be interpreted as the expected reward of a LP agent of supertype $\Lambda^{LP}_i$ using $\pi_1$, while all other LP agents are using $\pi_2$. $\widetilde{V}_{\Lambda^{LP}_i}$ also depends on $\pi^{LT}$, but since we assume it to be fixed in this section, we do not write the dependence explicitly. This method of having an agent use $\pi_1$ and all others use $\pi_2$ is mentioned in \cite{sofa} under the name "symmetric opponents form approach" (SOFA) in the context of symmetric games. Our game as we formulated it so far is not symmetric since different supertypes get different rewards, however we will see below that we will introduce a symmetrization of the game via the function $\widehat{V}$.

\textbf{Shared policy gradient and the higher-level game $\widehat{V}$.} In the parameter sharing framework, $\pi \equiv \pi_{\theta}$ is a neural network with weights $\theta$, and the gradient $\nabla^{shared}_{\theta,B}$ according to which the shared policy $\pi_\theta$ is updated (where $B$ is the number of episodes sampled) is computed by collecting all agent experiences simultaneously and treating them as distinct sequences of local states, actions and rewards experienced by the shared policy \cite{Gupta2017-it}, yielding the following expression under vanilla policy gradient, similar to the single-agent case:
\begin{align}
\label{grad}
\begin{split}
& \nabla^{shared}_{\theta,B} = \frac{1}{n_{LP}} \sum_{i=1}^{n_{LP}} g_i^B, \\
& g_i^B:=\frac{1}{B}\sum_{b=1}^{B} \sum_{t=0}^{T} \nabla_{\theta} \ln \pi_{\theta} \left(a^{(i,LP)}_{t,b}|s^{(i,LP)}_{t,b},\lambda_{i,b}^{LP} \right) \sum_{t'=t}^{T} \zeta^{t'} \RR^{LP}(z^{(i,LP)}_{t',b},\bm{z^{(-i,LP)}_{t',b}}, \bm{z^{(LT)}_{t',b}}).
\end{split}
\end{align}

Note that one may use a critic network in equation (\ref{grad}) in place of the sampled rewards $\RR^{LP}$, but this is related to sample efficiency and doesn't change the methods and observations developed subsequently. By the strong law of large numbers, taking $B=+\infty$ in (\ref{grad}) simply amounts to replacing the average by an expectation as in (\ref{vi}) with $\pi_1=\pi_2=\pi_\theta$. Proposition \ref{pk} is a key observation of this paper and sheds light upon the mechanism underlying parameter sharing in  (\ref{grad}): in order to update the shared policy, we first set all agents to use the same policy $\pi_\theta$, then pick one agent at random and take a step towards improving its individual reward while keeping other agents on $\pi_\theta$, and then have all agents use the new policy: by $(\ref{grad2})$, this yields an unbiased estimate of the gradient $\nabla^{shared}_{\theta,\infty}$. Sampling many agents at random $\alpha \sim U[1,n_{LP}]$ in order to compute the expectation in (\ref{grad2}) will yield a less noisy gradient estimate but will not change its bias. In (\ref{grad2}), $\widehat{V}$ is to be interpreted as the utility received by a randomly chosen agent using $\pi_1$ while all other agents use $\pi_2$.
\begin{proposition}
\label{pk}
For a function $f(\theta_1,\theta_2)$, let $\nabla_{1} f(\theta_1,\theta_2)$ be the gradient with respect to the first argument, evaluated at $(\theta_1,\theta_2)$. We then have:
\begin{align}
\label{grad2}
    \nabla^{shared}_{\theta,\infty} = \nabla_{1} \widehat{V}(\pi_{\theta},\pi_{\theta}),\hspace{4mm} \widehat{V}(\pi_{1},\pi_{2}):=\EE_{\alpha \sim U[1,n_{LP}]} \left[ \widetilde{V}_{\Lambda^{LP}_\alpha}(\pi_{1},\pi_{2})\right],\hspace{2mm} \pi_1, \pi_2 \in \mathcal{X}^{LP},
\end{align}
where $\EE_{\alpha \sim U[1,n_{LP}]}$ indicates that the expectation is taken over $\alpha$ random integer in $[1,n_{LP}]$.
\end{proposition}
\textit{Proof.} It is known (although in a slightly different form in \cite{Lockhart19ED} or \cite{Srinivasan18RPG} appendix D) that the term $g_i^\infty$ in (\ref{grad}) is nothing else than
$\nabla_{1} \widetilde{V}_{\Lambda^{LP}_i}(\pi_{\theta},\pi_{\theta})$, that is the sensitivity of the expected reward of an agent of supertype $\Lambda^{LP}_i$ to changing its policy while all other agents are kept on $\pi_\theta$, cf. (\ref{vi}). The latter can be seen as an extension of the  likelihood ratio method to imperfect information games, and allows us to write concisely, using (\ref{grad}):
\begin{align*}
\nabla^{shared}_{\theta,\infty} & =  \frac{1}{n_{LP}} \sum_{i=1}^{n_{LP}} \nabla_{1} \widetilde{V}_{\Lambda^{LP}_i}(\pi_{\theta},\pi_{\theta}) \\
& =  \nabla_{1} \frac{1}{n_{LP}} \sum_{i=1}^{n_{LP}}  \widetilde{V}_{\Lambda^{LP}_i}(\pi_{\theta},\pi_{\theta}) = \nabla_{1} \EE_{\alpha \sim U[1,n_{LP}]} \left[ \widetilde{V}_{\Lambda^{LP}_\alpha}(\pi_{\theta},\pi_{\theta})\right].
\end{align*}
\qed

\textbf{Shared Equilibria.} We remind that a 2-player game is said to be symmetric if the utility received by a player only depends on its own strategy and on its opponent's strategy, but not on the player's identity, and that a pure strategy Nash equilibrium $(\pi_1^*,\pi_2^*)$ is said to be symmetric if $\pi_1^*=\pi_2^*$ (see \cite{2sym}). Here, "pure strategy" is in the terminology of game theory, as opposed to "mixed strategy", which is a probability distribution over the set of pure strategies. We now introduce the terminology "payoff", relevant in the symmetric case.

\begin{definition}
\label{pay}
\textbf{(Payoff in symmetric games)} Consider a 2-player symmetric game where $u_i(\pi_1,\pi_2)$ represents the utility of player $i$ when each player $j$ plays $\pi_j$. Due to symmetry, we have $u_1(\pi_1,\pi_2) = u_2(\pi_2,\pi_1)$, and therefore we call payoff $u(\pi,\eta):=u_1(\pi,\eta)$ the utility received by a player playing $\pi$ while the other player plays $\eta$. 
\end{definition}

Equation (\ref{grad2}) suggests that the shared policy is a Nash equilibrium of the 2-player symmetric game with payoff $\widehat{V}$, where by our definition \ref{pay} of the term "payoff", the first player receives $\widehat{V}(\pi_1,\pi_2)$ while the other receives $\widehat{V}(\pi_2,\pi_1)$. This is because $\nabla_{1}\widehat{V}(\pi_\theta,\pi_\theta)$ in (\ref{grad2}) corresponds to trying to improve the utility of the first player while keeping the second player fixed, starting from the symmetric point $(\pi_\theta,\pi_\theta)$. If no such improvement is possible, we are facing by definition a symmetric Nash equilibrium, since due to symmetry of the game, no improvement is possible either for the second player starting from the same point $(\pi_\theta,\pi_\theta)$. Since the 2 players are not part of the $n_{LP}$ agents, the game with payoff $\widehat{V}$ can be seen as an abstract game where each pure strategy is a policy $\pi \in \mathcal{X}^{LP}$ defined in (\ref{pidef}). This type of game has been introduced in \cite{pmlr-v97-balduzzi19a} as a \textit{Functional Form Game} (FFG), since pure strategies of these games are stochastic policies themselves (but of the lower-level game among the $n_{LP}$ agents). This motivates the following definition.

\begin{definition}
\label{se}
\textbf{(Shared Equilibrium)} A shared (resp. $\epsilon-$shared) equilibrium $\pi^*$ associated to the supertype profile $\bm{\Lambda}:=(\bm{\Lambda^{LP}},\bm{\Lambda^{LT}})$ is defined as a pure strategy symmetric Nash (resp. $\epsilon-$Nash) equilibrium $(\pi^*,\pi^*)$ of the 2-player symmetric game with pure strategy set $\mathcal{X}^{LP}$ and payoff $\widehat{V}$ in (\ref{grad2}), cf. definition \ref{pay}.
\end{definition}

Note that the previously described mechanism occurring in parameter sharing is exactly what is defined as \textit{self-play} in \cite{pmlr-v97-balduzzi19a} (algorithm 2), but for the game of definition \ref{se} with payoff $\widehat{V}$. That is, we repeat the following steps for training iterations $n$: first set all agents on $\pi_{\theta_n}$, then pick one agent at random and improve its utility according to the gradient update (\ref{grad2}), thus finding a new policy $\pi_{\theta_{n+1}}$, then set all agents on $\pi_{\theta_{n+1}}$.

Consider for a moment only the case where types are defined to be agents' indexes: $\Lambda^{LP}_i:=\lambda^{LP}_i:=i$: this is a particular case of our formalism. If each type $i$ has its own set of parameters $\theta_i$ (in which case $\theta$ is a concatenation of the disjoint $\theta_i$), the gradient update (\ref{grad}) reduces to independent policy gradient, where each player updates its own policy independently of the others. Define the "underlying game" among $n_{LP}$ LP agents to be the game where each agent $i$ chooses a policy $\pi_i$, and receives a utility equal to its value function $V_i(\pi_i, \pi_{-i})$, defined as the expected value of its cumulative reward (cf. (\ref{vi2})). This definition is classical, cf. \cite{markovpg} for the study of Markov Potential games, and amounts to recasting a Markov game whereby agents take actions and transition from states to states over multiple timesteps as a one-shot game with utilities $V_i$ on a larger (pure) strategy space, namely the space of stochastic policies. This trick was also used in PSRO \cite{psro}, where they call the larger one-shot game a "meta-game". In this case, Nash equilibria of the underlying game among the $n_{LP}$ agents exactly coincide with shared equilibria. Indeed, assume that $(\pi^*_i)_{i \in [1,n_{LP}]}$ is any Nash equilibrium of the underlying game, i.e. $V_i(\eta, \pi^*_{-i}) \leq V_i(\pi^*_i, \pi^*_{-i})$ $\forall i$, $\forall \eta$. We can simply define $\pi(\cdot|\cdot, i):=\pi^*_i$ to get a shared equilibrium, where $\pi$ is the shared policy being used by all agents (but seen through a different prism $\pi(\cdot|\cdot, i)$ for each agent $i$). Conversely, every shared equilibrium is a Nash equilibrium of the underlying game. Indeed, assume $\pi$ is not a shared equilibrium, i.e. there exists a shared policy $\widetilde{\pi}$ such that $\widehat{V}(\widetilde{\pi},\pi)>\widehat{V}(\pi,\pi)$, i.e. $\sum_{i=1}^{n_{LP}} \widetilde{V}_{i}(\widetilde{\pi},\pi) - \widetilde{V}_{i}(\pi,\pi)>0$. If a sum is positive, at least one term of this sum is positive, but this is not possible by definition of the Nash, since by definition $\widetilde{V}_{i}(\widetilde{\pi},\pi)$ is the utility of agent $i$ when he plays $\widetilde{\pi}(\cdot|\cdot, i)$ while other agents $j$ play $\pi(\cdot|\cdot, j)=\pi^*_j$. Conversely, assume that $\pi^*_i:=\pi(\cdot|\cdot, i)$ is not a Nash of the underlying game, i.e. there exists a player $j$ and a strategy $\eta$ such that $V_j(\eta, \pi^*_{-j}) > V_j(\pi^*_j, \pi^*_{-j})$. Define the shared policy $\widetilde{\pi}$ as $\widetilde{\pi}(\cdot|\cdot, i):=\pi(\cdot|\cdot, i)$ for $i \neq j$, and $\widetilde{\pi}(\cdot|\cdot, j):=\eta$. Then, $\widehat{V}(\widetilde{\pi},\pi)-\widehat{V}(\pi,\pi)= \frac{1}{n_{LP}} \sum_{i=1}^{n_{LP}} \widetilde{V}_{i}(\widetilde{\pi},\pi) - \widetilde{V}_{i}(\pi,\pi) = \frac{1}{n_{LP}} (\widetilde{V}_{j}(\widetilde{\pi},\pi) - \widetilde{V}_{j}(\pi,\pi))>0$, a contradiction since $\pi$ is a shared equilibrium.

We now discuss the case where $\Lambda^{LP}_i$ and $\lambda^{LP}_i$ are general quantities rather than the agents' indexes $i$. In this case, the reasoning is the same except that shared equilibria will now essentially coincide with symmetric Bayesian Nash equilibria of the underlying game. One-shot Bayesian games are a natural generalization of normal form games in which at every instance of the game, each player is given a type $\lambda^{LP}_i$ randomly sampled from an exogenous distribution of types, which we call supertype. A policy for agent $i$ in a Bayesian game is a mapping from types to state-action policies of the original game, $\pi_i \equiv \lambda \to \pi_i(\cdot|\cdot, \lambda)$, i.e. agents need to specify their behavior for every type they can possibly be given. Therefore, Bayesian Nash equilibria are Nash equilibria on a larger pure strategy space where policies have been augmented with the agent type. We refer to \cite{wellman1} for the definition of Bayesian games and Bayesian Nash equilibria. One difference between Bayesian games and our game is that we do not assume that agents know the supertype profile $\bm{\Lambda^{LP}}$, called the "type prior" in Bayesian games and assumed to be of common knowledge. Each agent only knows its private type sampled at the beginning of each episode. In our case, all agents use the same shared policy $\pi$, but seen through a different prism $\pi(\cdot|s^{(i,\kappa)}_t, \lambda^\kappa_i)$ depending on the type. This means that the Bayesian Nash equilibria we are looking for are symmetric. As mentioned in section \ref{secshared}, our rationality assumption underlying parameter sharing is the following: if two agents have equal types and equal sequences of historical states at a given point in time, then they should behave the same way. From the definition of $\widehat{V}$, one could see our game as a symmetric Bayesian game where agents do not know the supertypes, and where the type distribution is $p(d\lambda):=\frac{1}{n_{LP}} \sum_{i=1}^{n_{LP}} p_{\Lambda^{LP}_i}(d\lambda)$. In this sense $\widehat{V}$ constitutes a symmetrization of our original game among supertypes $\Lambda^{LP}_i$. Together with the extended transitivity assumption \ref{et} which is similar to generalized ordinal potential games, our overall setting can be seen as a symmetric Bayesian Markov generalized ordinal potential game. 

Why is the abstract game $\widehat{V}$ useful? The abstract game characterizes the learning mechanism according to which an equilibrium is reached, rather than the equilibrium itself. Such an abstract game is used as a mean to find equilibria of the underlying game among the $n_{LP}$ agents. In our case, the learning mechanism is that at each step, a randomly selected agent tries to find a profitable deviation. This is by definition of $\widehat{V}$. Such sequences of improvements by individual players are known as "improvement paths" since the seminal work on potential games \cite{potential}. One of our contributions is to give a rigorous, game-theoretical explanation to the policy sharing algorithms that have been used in the literature \cite{Gupta2017-it} via (variants of) equation (\ref{grad}), but without theoretical grounding. Importantly, such characterization allows to relate policy sharing to potential games, giving us insights on conditions required for the learning mechanism to converge, namely our "extended transitivity" assumption \ref{et}. This learning mechanism, from a practical point of view, further allows us to use all agents' experiences to interpolate policies $\pi(\cdot|\cdot, \lambda)$ in the type variable $\lambda$, using the generalization power of neural nets. For example, if in a specific instance of the game, the sampled risk aversions $\gamma$ for 2 LP agents are 0.5 and 1, the related experience will be used at the next stage of the game for an agent which risk aversion is 0.6.

The natural question is now \textit{under which conditions do Shared equilibria exist, and can the self-play mechanism in (\ref{grad2}) lead to such equilibria?} We know  \cite{pmlr-v97-balduzzi19a} that self-play is related to transitivity in games, so to answer this question, we introduce a new concept of transitivity that we call \textit{extended transitivity} as it constitutes a generalization to 2-player symmetric general sum games of the concept of transitivity for the zero-sum case in \cite{pmlr-v97-balduzzi19a}. There, such a transitive game has payoff $u(x,y):=t(x)-t(y)$. One can observe that this game satisfies extended transitivity in assumption \ref{et} with $\delta_\epsilon:=\epsilon$ and $\Phi(x):=t(x)$. Note also that their monotonic games for which $u(x,y):=\sigma(t(x)-t(y))$ (where $\sigma$ is increasing) satisfy extended transitivity as well with $\delta_\epsilon:=\sigma^{(-1)}(\epsilon+\sigma(0))$ and $\Phi(x):=t(x)$.

\begin{assumption}
\label{et}
\textbf{(extended transitivity)} A 2-player symmetric game with pure strategy set $S$ and payoff $u$ is said to be extended transitive if there exists a bounded function $\Phi: S \to \R$ such that:
$$
\forall \epsilon > 0, \exists \delta_\epsilon>0: \forall x,y \in S: \text{ if } u(y,x)-u(x,x)>\epsilon, \text{ then } \Phi(y)-\Phi(x)> \delta_\epsilon.
$$
\end{assumption}

The intuition behind assumption \ref{et} is that $\Phi$ can be seen as the game "skill" that is being learnt whenever a player finds a profitable deviation from playing against itself. It will be required in theorem \ref{n11} to prove the existence of shared equilibria. Actually, it will be proved that such equilibria are reached by following self-play previously discussed, thus showing that policy updates based on (\ref{grad2}) with per-update improvements of at least $\epsilon$ achieve $\epsilon$-shared equilibria within a finite number of steps. In order to do so, we need definition \ref{defsp} of a \textit{self-play sequence}, which is nothing else than a rigorous reformulation of the mechanism occurring in self-play \cite{pmlr-v97-balduzzi19a} (algorithm 2). For $\epsilon$-shared equilibria, assumption \ref{et} is sufficient, but for shared equilibria, we need the continuity result in lemma \ref{lvf2}.

\begin{definition}
\label{defsp}
A $\bm{(f,\epsilon)}$-\textbf{self-play sequence} $(x_n,y_n)_{0\leq n \leq 2N}$ of size $0\leq 2N\leq +\infty$ generated by $(z_n)_{n \geq 0}$ is a sequence such that for every $n$, $x_{2n}=y_{2n}=z_n$, $(x_{2n+1},y_{2n+1})=(z_{n+1},z_n)$ and $f(x_{2n+1},y_{2n+1})>f(x_{2n},y_{2n})+\epsilon$.
\end{definition}

\begin{lemma}
\label{lvf2}
Assume that the rewards $\RR^{LP}$ are bounded, and that $\Ss^{LP}$, $\As^{LP}$ and $\Ss^{\lambda^{LP}}$ are finite. Then $\widetilde{V}_{\Lambda^{LP}_i}$ is continuous on $\mathcal{X}^{LP} \times \mathcal{X}^{LP}$ for all $i$, where $\mathcal{X}^{LP}$ is equipped with the total variation metric.
\end{lemma}

\begin{theorem}
\label{n11}
Let $\bm{\Lambda^{LP}}$ be a supertype profile. Assume that the symmetric 2-player game with pure strategy set $\mathcal{X}^{LP}$ and payoff $\widehat{V}$ is extended transitive. Then, there exists an $\epsilon-$shared equilibrium for every $\epsilon>0$, which further can be reached within a finite number of steps following a $(\widehat{V},\epsilon)$-self-play sequence. Further, if $\Ss^{LP}$, $\As^{LP}$ and $\Ss^{\lambda^{LP}}$ are finite and the rewards $\RR^{LP}$ are bounded, then there exists a shared equilibrium.
\end{theorem}

We should comment on the relationship between our extended transitivity and potential games \cite{potential}. Extended transitivity may seem a bit abstract, but it is well-rooted in the game theory literature. In simple terms, it corresponds to games of skill where there are some universally good strategies \cite{spin}. For example in chess or tennis, some players are universally "good". This is because there is a skill underlying the game that players need to master. Mathematically, that skill is the function $\Phi$ in our extended transitivity assumption. This is in contrast to cyclic games like rock-paper-scissors where there is no dominating strategy. Precisely, our extended transitivity is very closely related to \textit{generalized ordinal potential games} in the seminal work of \cite{potential}. A 2 player symmetric game $u$ as in assumption \ref{et} is said to be generalized ordinal potential with potential function $\mathcal{P}$ if \cite{potential}: 
$$
\textbf{(GOP)} \hspace{3mm} \forall x,y,z \in S: \text{ if } u(y,z)-u(x,z)>0, \text{ then } \mathcal{P}(y,z)-\mathcal{P}(x,z)> 0.
$$

The first comment is that our $\epsilon-\delta_\epsilon$ requirement is a "uniform" version of the $>0$ requirement in (GOP), like continuity vs. uniform continuity. We need it for technical reasons in Lemma \ref{f2sis} in appendix. In the discussion below, we omit this technical aspect. The second comment is that extended transitivity only assumes deviations from symmetric points $(x,x)$, contrary to all points $(x,z)$ in (GOP). There are two ways to connect extended transitivity and (GOP). The first way is to assume that $\mathcal{P}(y,z)-\mathcal{P}(x,z)$ in (GOP) does not depend on $z$, which occurs for example if $\mathcal{P}$ is a separable function, i.e. $\mathcal{P}(x,y)=p_1(x)+p_2(y)$. Then (GOP) implies extended transitivity by taking $z=x$. The second way is to assume that: 
\begin{align}
\label{a1gop}
\hspace{3mm}\forall x,y \in S: \text{ if } u(y,x)-u(x,x)>0, \text{ then } u(y,y)-u(x,y)>0.
\end{align}

(\ref{a1gop}) is an intuitive assumption, always true in the zero-sum symmetric case where $u$ is an antisymmetric function. If $u(y,x)-u(x,x)>0$, it means that the strategy $y$ is "good" when the other player plays $x$. If $y$ is good in some universal way, then it is also good when the other player plays $y$, i.e.  $u(y,y)-u(x,y)>0$. If $\mathcal{P}$ is a symmetric function \footnote{If the 2 player symmetric game $u$ is exact potential - which is stronger than (GOP) - then it is well-known that $\mathcal{P}$ is symmetric, namely $\mathcal{P}(x,y)=\mathcal{P}(y,x)$.}, we have that (GOP) together with (\ref{a1gop}) imply extended transitivity with $\Phi(x):= \mathcal{P}(x,x)$. Indeed, assume that $u(y,x)-u(x,x)>0$. Then by (GOP), $\mathcal{P}(y,x)-\mathcal{P}(x,x)> 0$. By (\ref{a1gop}), $u(y,y)-u(x,y)>0$, which by (GOP) yields $\mathcal{P}(y,y)-\mathcal{P}(x,y)> 0$. This implies, using the symmetry of $\mathcal{P}$, that $\Phi(y)-\Phi(x)=\mathcal{P}(y,y)-\mathcal{P}(x,x)=\mathcal{P}(y,y)-\mathcal{P}(x,y) + \mathcal{P}(y,x) -\mathcal{P}(x,x) >0$, which proves extended transitivity. However, for extended transitivity to be true, we only need the weaker $\mathcal{P}(y,y) > \mathcal{P}(x,x)$.

It is possible to check empirically that extended transitivity holds for the game $\widehat{V}$: for this, we need to check that our learning mechanism based on a random player finding a profitable deviation at each training iteration makes the cumulative reward of the shared policy approximately monotonically increasing during training. If extended transitivity were false, we would observe a cyclic behavior of the shared policy reward, as as in the case of rock-paper-scissors for example. We observe that LPs' shared policy reward indeed has the correct behavior in figures \ref{f1} and \ref{figgn1}.

We now proceed to proving that the gradient update (\ref{grad2}) converges to $\epsilon-$shared equilibria. In order to prove theorem \ref{gradet}, we reformulate extended transitivity in proposition \ref{et2}.

\begin{proposition}
\label{et2}
A 2-player symmetric game with pure strategy set $S$ and payoff $u$ is extended transitive if and only if there exists a bounded function $\Phi: S \to \R$ and a continuous and strictly increasing function $\varphi: [0,+\infty) \to [0,+\infty)$ with $\varphi(0)=0$ such that:
$$
\Phi(y) - \Phi(x) \geq \varphi(u(y,x)-u(x,x)) \hspace{4mm} \forall (x,y) \in S_0,
$$
where $S_0 := \{(x,y) \in S: u(y,x)-u(x,x)>0\}$. In particular, by considering the convex Hull of $\varphi$, it can be chosen convex. Consequently, if $x \in S$ and $Y$ is a random variable such that $(x,Y)$ takes value in $S_0$ almost surely, then Jensen's inequality yields:
$$
\EE[\Phi(Y)] - \Phi(x) \geq \varphi(\EE[u(Y,x)]-u(x,x)).
$$
\end{proposition}

We are now ready to prove that the gradient update (\ref{grad2}) converges to $\epsilon-$shared equilibria under direct policy parametrization, i.e. when each action-state-type tuple has its own parameter. Under such parametrization it is needed to project back onto the probability simplex after each gradient step. Our proof will follow the line of reasoning in \cite{convpg}, theorem 4.1, in the single-agent case. It uses a "gradient domination" argument where the optimality gap at stage $n$ $\widehat{V}(\pi_\theta,\pi_{\theta_n}) - \widehat{V}(\pi_{\theta_n},\pi_{\theta_n})$ is dominated by the gradient. This approach was also used in \cite{markovpg} in the context of Markov (exact) potential games. Proposition \ref{et2} will turn out to be key to adapt the proof to the multi-agent case, and our proof shows that we do not need any smoothness on the extended transitivity map $\Phi$, rather we use the properties of $\varphi$ (increasing and invertible) to upper bound positive agent deviations $\widehat{V}(y,x)-\widehat{V}(x,x)$ by $\Phi$ and to move between the payoff space and the potential space. In particular, this answers Remark 5 in \cite{markovpg}, where it is noticed that obtaining convergence rates in the case of ordinal potential games is not immediate due to the lack of smoothness of the potential function \footnote{we refer to the discussion below theorem \ref{n11} for the close link between generalized ordinal potential games and extended transitivity.}. The results below use the notion of $\beta$-smoothness, which we recall means that the gradient is $\beta$-Lipschitz.

Remember that the initial state distribution of an agent of supertype $\Lambda^{LP}_i$ is $\int \mu^0_{\lambda}(\cdot) p_{\Lambda^{LP}_i}(d\lambda)$. In the following, we will use the notation $\widehat{V}_{\mu^0}$ to make the dependence explicit on $\mu^0$. We introduce the following visiting distribution:
$$
d^\pi_{\mu^0,i}(ds):= \frac{1}{T_\zeta} \int \sum_{t=0}^T \zeta^t \PP_\pi[s^{(i,LP)}_t \in ds | s^{(i,LP)}_0=s_0] \mu^0_{\lambda}(d s_0) p_{\Lambda^{LP}_i}(d\lambda),
$$
where the normalization factor $T_{\zeta}:=T+1$ if $\zeta=1$, and $T_{\zeta}:= \frac{1-\zeta^{T+1}}{1-\zeta}$ otherwise. Define now $\mathcal{N}^\epsilon_{\mu^0}$ the set of $\epsilon$-shared equilibria associated to $\mu^0$ (which is not empty by theorem \ref{n11}), and $BR_{\mu^0,i}(\pi)$ the set of policies $\pi'$ such that $\widetilde{V}_{\Lambda^{LP}_i,\mu^0}(\pi',\pi)-\widetilde{V}_{\Lambda^{LP}_i,\mu^0}(\pi,\pi)$ is maximal. By definition of the equilibrium, the latter gap is bounded by $\epsilon$ whenever $\pi \in \mathcal{N}^\epsilon_{\mu^0}$. Let $\widetilde{\mu}^0:= \lambda \to \widetilde{\mu}^0_\lambda(\cdot)$ any initial state distribution. In the case where $\Ss^{LP}$ is finite, define the distribution mismatch coefficient:
\begin{align}
\label{dmis}
\mathcal{D}_\epsilon(\mu^0, \widetilde{\mu}^0) := \max_{i \in [1,n_{LP}]} \sup_{\pi \in \mathcal{N}^\epsilon_{\mu^0}} \sup_{\pi' \in BR_{\mu^0,i}(\pi)} \max_{s \in \Ss^{LP}} \frac{d^{\pi'}_{\mu^0,i}(s)}{d^{\pi}_{\widetilde{\mu}^0,i}(s)}.
\end{align}

Note that:
$$
\mathcal{D}_\epsilon(\mu^0, \widetilde{\mu}^0) \leq \frac{T_{\zeta}}{\min_{i \in [1,n_{LP}]} \min_{s \in \Ss^{LP}}  \int \widetilde{\mu}^0_{\lambda}(s) p_{\Lambda^{LP}_i}(d\lambda)}.
$$

This is similar to the mismatch coefficient in \cite{convpg}, but adapted to our multi-agent case. The rationale behind using $\mu^0$ and $\widetilde{\mu}^0$ is that although we may be interested in computing equilibria associated to $\mu^0$, it can be useful to run gradient ascent on the value function associated to $\widetilde{\mu}^0$. One is always free to take $\widetilde{\mu}^0 = \mu^0$. The typical case is when $\mu^0$ is the indicator function of a particular state, and $\widetilde{\mu}^0$ is uniformly spread across states.

\begin{theorem}
\label{gradet}
Let $\bm{\Lambda^{LP}}$ be a supertype profile, $\mu^0$, $\widetilde{\mu}^0$ be arbitrary initial state distributions and $\epsilon>0$. Assume that $\Ss^{LP}$, $\As^{LP}$ and $\Ss^{\lambda^{LP}}$ are finite, that the rewards $\RR^{LP}$ are bounded by $\RR_{max}$, and that the symmetric 2-player game with pure strategy set $\mathcal{X}^{LP}$ and payoff $\widehat{V}_{\widetilde{\mu}^0}$ is extended transitive. Consider the infinite time-horizon setting $T=+\infty$ with discount factor $\zeta<1$, as well as the direct policy parametrization $\pi_{\theta}(a|s,\lambda) := \theta^{s,\lambda}_{a}$, where $(\theta^{s,\lambda}_{a})_{a \in \As^{LP}}$ is in the probability simplex for every $s \in  \Ss^{LP}$, $\lambda \in \Ss^{\lambda^{LP}}$. Consider the so-called Projected Gradient Ascent (PGA) update, whereby at each iteration $n$:
$$
\theta_{n+1} = P_{\mathcal{X}^{LP}} \left[\theta_n + \alpha \nabla_{1} \widehat{V}_{\widetilde{\mu}^0}(\pi_{\theta_n},\pi_{\theta_n}) \right],
$$
where $0 < \alpha < \frac{2}{\beta}$ is the learning rate, $\beta := \frac{2 \zeta \RR_{max} |\As^{LP}|}{(1-\zeta)^3}$, and $P_{\mathcal{X}^{LP}}$ denotes the euclidean projection onto $\mathcal{X}^{LP}$. Assume $\mathcal{D}_\epsilon(\mu^0, \widetilde{\mu}^0)<+\infty$ in (\ref{dmis}), and set the number of iterations $N$ to be:
$$
N \geq \frac{diam(\Phi)}{\varphi \left(\frac{\epsilon^2 }{4 \mathcal{D}_\epsilon(\mu^0, \widetilde{\mu}^0)^2 |\Ss^{LP}||\Ss^{\lambda^{LP}}|} \cdot \frac{\alpha (2 -\alpha \beta)}{2(1+2\alpha \beta)^2} \right)},
$$
where $\varphi$ and $\Phi$ are defined in assumption \ref{et} and proposition \ref{et2}, and $diam(\Phi):=\sup \{|\Phi(x)-\Phi(y)|: x,y \in \mathcal{X}^{LP}\}$. Then, there exists a $n \in [0,N]$ such that $\pi_{\theta_n}$ is an $\epsilon-$shared equilibrium associated to $\mu^0$.
\end{theorem}

The proof of theorem \ref{gradet} also gives this corollary:
\begin{corollary}
\label{gradetc}
Under the setting of theorem \ref{gradet}, let $(\alpha_n)$ a sequence of learning rates such that $0 < \alpha_n < \frac{2}{\beta}$ and $\sum \alpha_n = +\infty$. Then for any $\epsilon$, there exists a $N_\epsilon$ such that for all $n \geq N_\epsilon$, $\pi_{\theta_n}$ is an $\epsilon$-shared equilibrium. In particular, every converging subsequence $(\theta_{m_n})$ converges to a shared equilibrium. If there exists a unique shared equilibrium, $(\theta_n)$ converges to that equilibrium.
\end{corollary}

\begin{remark}
Theorem \ref{gradet} still holds in the case of $\zeta \in [0,1]$ with finite time-horizon $T$, or in the episodic case where the episode ends upon the first hitting time of a certain goal state. In the first case, the normalization factor $(1-\zeta)^{-1}=\sum_{t=0}^\infty \zeta^t$ becomes $T_{\zeta}:=T+1$ if $\zeta=1$, and $T_{\zeta}:= \frac{1-\zeta^{T+1}}{1-\zeta}$ otherwise. In the second case it is upper bounded by some finite number $T_{\zeta}$, provided we assume that under all policies, the probability of reaching one of the goal states is strictly positive (this is assumption 2.1 in \cite{Tamar2012-ki}). In all cases we have $\beta = 2 \RR_{max} |\As^{LP}| T_{\zeta}^3$.
\end{remark}

A natural question is whether theorem \ref{gradet} holds in the presence of noisy observations of the true gradient, i.e. we do not observe $\nabla_{1} \widehat{V}$ directly, but a noisy, unbiased version of it. In this case $(\theta_n)$ becomes a sequence of random variables, where each coordinate is impacted by some noise at each iteration $n$. The projection step makes it non trivial to adapt our proof to this case. The technical difficulty lies in the fact that the coordinates $\theta_{n,a}^{s,\lambda}$ can come arbitrarily close to the boundary of the simplex. Given a $\theta$ in the simplex, and a small stochastic increment $\delta \theta$, the quantity $P_{\mathcal{X}^{LP}}[\theta + \delta \theta]$ is in general a non-trivial object which coordinates are not independent. However, for "small" $\delta \theta$, we can show that $P_{\mathcal{X}^{LP}}[\theta + \delta \theta] = \theta + \delta \theta - K^{-1} \sum_{k=1}^K \delta \theta_k$. The latter is very convenient to manipulate, especially because we can deal explicitly with the stochasticity of the coordinates. Unfortunately, the technical difficulty lies in the fact that this formula holds provided we are sufficiently far away from the boundary, which we haven't managed to control uniformly as iterations progress. If we knew that for every iteration $n$, $\theta_n$ would be at least at a small but uniform distance $\nu$ from the boundary, we would be done. One solution would be to consider a mixture between $\theta_n$, and the uniform distribution (with weight $\nu$), but this brings in other technical difficulties.

For this reason, we turn to the softmax parametrization, where $\pi_{\theta}(a|s,\lambda) := \frac{\exp(\theta^{s,\lambda}_{a})}{\sum_a \exp(\theta^{s,\lambda}_{a})}$, and $(\theta^{s,\lambda}_{a})_{a \in \As^{LP}}$ is an arbitrary real-valued vector. In this setting, projection is not needed as $\pi_{\theta}$ lies by construction in the simplex. Note that showing convergence rates for this parametrization also suffers from the same above-mentioned technical difficulty of getting close to the boundary, as pointed out in \cite{convpg}, section 5.1, in the single-agent case. For this reason, the latter work considers the so-called log-barrier regularizer which precisely prevents the policy from getting too close to the boundary, and simply adds a term $\nu \ln \pi$ to the objective, cf. (\ref{logpo}). We are able to extend their approach to the multiagent setting, making use of our key proposition \ref{et2} as we did in the proof of theorem \ref{gradet}. Regarding the assumption on the structure of the noise, we will consider a "multiplicative" approach where the noise impacting each coordinate is proportional to that coordinate. This approach has the advantage of being invariant with respect to the magnitude of the coordinates, i.e. multiplying a coordinate by any positive number doesn't change the noise. Although this approach has been considered in \cite{msgd}, the most common approach is to consider additive noise. 

To the best of our knowledge, in addition to the use of proposition \ref{et2}, the novelty in our results of theorems \ref{gradet2}, \ref{gradet3} lies in lemma \ref{lemsto}. The idea is essentially to use the central limit theorem and see that when the number of coordinates $K:=|\Ss^{LP}||\Ss^{\lambda^{LP}}||\As^{LP}|$ is large (which is the case in practice), we can make use of the Berry-Esseen theorem to get a convergence result with high probability (and not just in expectation). In the non-noisy case of theorem \ref{gradet}, the $\beta$-smoothness of $\widehat{V}$ allows us to get a guaranteed improvement in $\widehat{V}$, that we can control via the extended transitivity map $\Phi$, cf. proposition \ref{et2}:
$$
\Phi(\theta_{n+1})-\Phi(\theta_{n}) \geq \varphi( \widehat{V}(\theta_{n+1},\theta_n) - \widehat{V}(\theta_n,\theta_n)).
$$

For us to be able to use the latter, we need the $\widehat{V}$ increment $\widehat{V}(\theta_{n+1},\theta_n) - \widehat{V}(\theta_n,\theta_n)$ to be non-negative on some set of (hopefully high) probability, and not only in expectation. This is because on paths where this increment is negative, we have no way to relate it to $\Phi$ due to the definition of extended transitivity. At this point we can make a technical comment that we leave for future work: for the last inequality of proposition \ref{et2} to happen, we need $(x,Y)$ to take value in $S_0$ a.s., which means that we cannot use this inequality if we have $u(Y,x)-u(x,x) \leq 0$ on some set, even if $\EE[u(Y,x)]-u(x,x)>0$. If one wanted to use the latter, one would need to strengthen extended transitivity in assumption \ref{et}, by requiring the statement to be true not only for every $y \in S$, but for any $y$ random variable taking value in $S$, i.e. $y \in \Delta(S)$.

For now, we turn to lemma \ref{lemsto}, the core of our next results as it shows how to deal with the noise.

\begin{lemma}
\label{lemsto}
Let $f$ be a $\beta$-smooth function $\R^d \to \R$ \footnote{that is, $\nabla f$ is $\beta$-Lipschitz.}. Assume that $x \in \R^d$, $\nabla f(x) \neq 0$ and let:
$$
y = x + \alpha \nabla f(x) + \alpha \nabla f(x) \odot \phi,
$$

where $\odot$ is the elementwise product, $\phi$ is a vector of length $d$ of zero-mean independent random variables. If $\phi$ is upper bounded by $\phi_{max}$, and lower bounded by $\phi_{min}>-1$ almost surely, then for any $\eta \in [0,1]$:
$$
f(y)-f(x) \geq (1+\phi_{min})\eta ||\nabla f(x)||_2^2 \hspace{1mm} \mbox{ a.s. } \hspace{3mm} \mbox{ for } \hspace{3mm} \alpha = \frac{2 (1-\eta)}{\beta(1+\phi_{max})}.
$$

Alternatively, if $\phi$ has finite (absolute) moments up to order 6, with $\EE[|\phi_{i}|^k] = \sigma_k$ for $k \leq 6$ and all coordinates $i \in [1,d]$, then for any $\eta,\eta'\in [0,1]$, we have for $\alpha = \frac{2(1-\eta)}{\beta(1+\sigma_2)}$:
$$
\PP[f(y)-f(x) \geq \eta'||\nabla f(x)||_2^2] \geq 1 - \mathcal{N}\left(\frac{\eta'-\eta}{\sqrt{\kappa_{4}}} \cdot \widehat{\mathcal{E}}_d \sqrt{d}\right) - \frac{0.6 \kappa_{6}}{\kappa_{4}^{3/2}}\mathcal{E}_d,
$$

where $\mathcal{N}$ is the standard normal c.d.f. and:
\begin{align*}
&\kappa_{4} := (1 - \alpha \beta)^2\sigma_2 + \frac{\alpha^2 \beta^2}{4} (\sigma_{4}-\sigma^{2}_2), \hspace{13mm} \kappa_{6} \leq \sum_{k=0}^6 |c_k| \sigma_{k},\\
& \mathcal{E}_d := \frac{\max_{i\in[1,d]} [\nabla f(x)]_i^2  }{\sqrt{\sum_{i=1}^{d} [\nabla f(x)]_i^4}}, \hspace{15mm} \widehat{\mathcal{E}}_d:= \frac{ \sum_{i=1}^{d}[\nabla f(x)]_i^2}{\sqrt{d \cdot \sum_{i=1}^{d} [\nabla f(x)]_i^4}}.
\end{align*}

$\mathcal{E}_d$ and $\widehat{\mathcal{E}}_d$ are respectively the maximum and average dispersion of the vector $\nabla f(x)$, and both lie in the interval $[d^{-1/2},1]$ due to the $L_1$-$L_2$ norm inequality. When all entries of $\nabla f(x)$ are the same, $\mathcal{E}_d=d^{-1/2}$ and $\widehat{\mathcal{E}}_d=1$. When all entries but one are zero, $\mathcal{E}_d = 1$ and $\widehat{\mathcal{E}}_d=d^{-1/2}$. The constants $c_{k}$ are given as follows, with $y=\frac{\alpha \beta}{2}$:
\begin{align*}
&c_0= y^3 \sigma_2^3, \hspace{8mm} c_1= 3y^2(1-2y) \sigma_2^2, \hspace{8mm} c_2:= -3 \sigma_2^2 y^3+3\sigma_2y(1-2y)^2,\\
& c_3=(1-2y)^3-6\sigma_2y^2(1-2y), \hspace{8mm} c_4=3 \sigma_2y^3-3y(1-2y)^2,\\
& c_5=3y^2(1-2y), \hspace{8mm} c_6= - y^3.
\end{align*}
\end{lemma}

The following results will use the softmax parametrization with log-barrier regularization as previously discussed. We define the logarithm of a policy $\pi$ as follows, where $K:=|\Ss^{LP}||\Ss^{\lambda^{LP}}||\As^{LP}|$:
\begin{align}
\label{logpo}
\ln \pi := \frac{1}{K} \sum_{s,a,\lambda} \ln \pi(a|s,\lambda). 
\end{align}
In theorem \ref{gradet2}, we consider the case of uniformly bounded noise, where we are able to get a convergence result with probability one. The proof is almost the same as that of theorem \ref{gradet}, except that we use updated smoothness constants linked to the change of parametrization, where the constants are taken from \cite{convpg} (single-agent case). We also use lemma \ref{lemsto} to deal with the noise. In theorem \ref{gradet3}, we consider the case of bounded moments, and we get a result in probability. Note that we consider an algorithm where at each iteration $n$, if the noisy gradient yields an improvement $\theta_{n} \to \theta_{n+1}$, we keep the new parameter $\theta_{n+1}$, if not, we discard it and try again. The high-level idea is that we want to almost never discard $\theta_{n+1}$. By lemma \ref{lemsto}, if the coordinates of the gradient are sufficiently homogeneous, we get a convergence result with high probability, i.e. we almost never discard $\theta_{n+1}$. Indeed, if the coordinates of the gradient are the same, the probability to get an improvement large enough is at least:
$$
1 - \mathcal{N}\left(\frac{\eta'-\eta}{\sqrt{\kappa_{4}}}\sqrt{K}\right) - \frac{0.6 \kappa_{6}}{\kappa_{4}^{3/2} \sqrt{K}},
$$
which is almost one if the number of coordinates $K$ is large, since we are free to choose $\eta'<\eta$.

\begin{theorem}
\label{gradet2}
\textbf{(Uniformly bounded noise)}. Let $\bm{\Lambda^{LP}}$ be a supertype profile, $\mu^0$, $\widetilde{\mu}^0$ be arbitrary initial state distributions, $\epsilon>0$ and $\eta\in (0,1)$. Assume that $\Ss^{LP}$, $\As^{LP}$ and $\Ss^{\lambda^{LP}}$ are finite, that the rewards $\RR^{LP}$ are bounded by $\RR_{max}$ and that $(\phi_n)$ is a i.i.d. sequence of vectors of length $K:=|\Ss^{LP}||\Ss^{\lambda^{LP}}||\As^{LP}|$ of mean-zero random variables that satisfies the conditions of lemma \ref{lemsto} (uniformly bounded case). Consider the infinite time-horizon setting $T=+\infty$ with discount factor $\zeta<1$, as well as the softmax policy parametrization $\pi_{\theta}(a|s,\lambda) := \frac{\exp(\theta^{s,\lambda}_{a})}{\sum_a \exp(\theta^{s,\lambda}_{a})}$, where $(\theta^{s,\lambda}_{a})_{a \in \As^{LP}}$ is an arbitrary real-valued vector. Consider the following update rule with log-barrier regularization, whereby at each iteration $n$:
\begin{align*}
& \theta_{n+1} = \theta_n + \alpha \nabla_{1} L_{\widetilde{\mu}^0}(\pi_{\theta_n},\pi_{\theta_n}) \odot (\bm{1} + \phi_n),\\ 
&L_{\widetilde{\mu}^0}(\pi_1,\pi_2):=\widehat{V}_{\widetilde{\mu}^0}(\pi_1,\pi_2) + \nu (\ln \pi_1 - \ln \pi_2),
\end{align*}

where $\alpha = \frac{2 (1-\eta)}{\beta_\nu(1+\phi_{max})}$ is the learning rate, $\beta_\nu := \frac{8 \RR_{max}}{(1-\zeta)^3} + \frac{2 \nu \RR_{max}}{|\Ss^{LP}||\Ss^{\lambda^{LP}}|}$, and the regularizer $\nu \leq \frac{(1-\zeta) \epsilon}{2\mathcal{D}_\epsilon(\mu^0, \widetilde{\mu}^0)}$. Assume $\mathcal{D}_\epsilon(\mu^0, \widetilde{\mu}^0)<+\infty$ in (\ref{dmis}), and that the symmetric 2-player game with pure strategy set $\mathcal{X}^{LP}$ and payoff $L_{\widetilde{\mu}^0}$ is extended transitive associated to mappings $\varphi_\nu$ and $\Phi_\nu$ in assumption \ref{et} and proposition \ref{et2}. Define:
$$
N := \frac{diam(\Phi_\nu)}{\varphi_\nu \left(\frac{(1-\zeta)^2 \epsilon^2 \eta (1+\phi_{min})}{16 K^2 \mathcal{D}_\epsilon(\mu^0, \widetilde{\mu}^0)^2} \right)},
$$
where $diam(\Phi_\nu):=\sup \{|\Phi_\nu(x)-\Phi_\nu(y)|: x,y \in \mathcal{X}^{LP}\}$. Then, with probability one, there exists a $n \in [0,N]$ such that $\pi_{\theta_n}$ is an $\epsilon-$shared equilibrium associated to $\mu^0$. 
\end{theorem}

\begin{theorem}
\label{gradet3}
\textbf{(Noise with bounded moments)}. Let $\bm{\Lambda^{LP}}$ be a supertype profile, $\mu^0$, $\widetilde{\mu}^0$ be arbitrary initial state distributions, $\epsilon>0$ and $\eta' < \eta\in (0,1)$. Assume that $\Ss^{LP}$, $\As^{LP}$ and $\Ss^{\lambda^{LP}}$ are finite, that the rewards $\RR^{LP}$ are bounded by $\RR_{max}$ and that $(\phi_n)$ is a i.i.d. sequence of vectors of length $K:=|\Ss^{LP}||\Ss^{\lambda^{LP}}||\As^{LP}|$ of independent mean-zero random variables that satisfies the conditions of lemma \ref{lemsto} (bounded moments case). Consider the infinite time-horizon setting $T=+\infty$ with discount factor $\zeta<1$, as well as the softmax policy parametrization $\pi_{\theta}(a|s,\lambda) := \frac{\exp(\theta^{s,\lambda}_{a})}{\sum_a \exp(\theta^{s,\lambda}_{a})}$, where $(\theta^{s,\lambda}_{a})_{a \in \As^{LP}}$ is an arbitrary real-valued vector. Consider the following update rule with log-barrier regularization, whereby at each iteration $n$:
\begin{align*}
& \theta_{n+1} = \theta_n + \alpha \nabla_{1} L_{\widetilde{\mu}^0}(\pi_{\theta_n},\pi_{\theta_n}) \odot (\bm{1} + \phi_n) \hspace{5mm} \mbox{ on the event } \Omega_n,\\ 
& \theta_{n+1} = \theta_n \hspace{5mm} \mbox{ otherwise.}\\
&L_{\widetilde{\mu}^0}(\pi_1,\pi_2):=\widehat{V}_{\widetilde{\mu}^0}(\pi_1,\pi_2) + \nu (\ln \pi_1 - \ln \pi_2),\\
& \Omega_n := \{L_{\widetilde{\mu}^0}(\pi_{\theta_{n+1}},\pi_{\theta_n}) - L_{\widetilde{\mu}^0}(\pi_{\theta_{n}},\pi_{\theta_n})\geq \eta'||\nabla_{1} L_{\widetilde{\mu}^0}(\pi_{\theta_n},\pi_{\theta_n})||_2^2\},
\end{align*}

where $\alpha = \frac{2(1-\eta)}{\beta_\nu(1+\sigma_2)}$ is the learning rate, $\beta_\nu := \frac{8 \RR_{max}}{(1-\zeta)^3} + \frac{2 \nu \RR_{max}}{|\Ss^{LP}||\Ss^{\lambda^{LP}}|}$, and the regularizer $\nu \leq \frac{(1-\zeta) \epsilon}{2\mathcal{D}_\epsilon(\mu^0, \widetilde{\mu}^0)}$. Assume $\mathcal{D}_\epsilon(\mu^0, \widetilde{\mu}^0)<+\infty$ in (\ref{dmis}), and that the symmetric 2-player game with pure strategy set $\mathcal{X}^{LP}$ and payoff $L_{\widetilde{\mu}^0}$ is extended transitive associated to mappings $\varphi_\nu$ and $\Phi_\nu$ in assumption \ref{et} and proposition \ref{et2}. Define:
$$
N := \frac{diam(\Phi_\nu)}{\varphi_\nu \left(\frac{(1-\zeta)^2 \epsilon^2 \eta'}{16 K^2 \mathcal{D}_\epsilon(\mu^0, \widetilde{\mu}^0)^2} \right)},
$$
where $diam(\Phi_\nu):=\sup \{|\Phi_\nu(x)-\Phi_\nu(y)|: x,y \in \mathcal{X}^{LP}\}$. Define $\mathcal{E}_{K} := \sup_{n \in [0,N_{tot}]} \mathcal{E}_{n,K}$, and $\widehat{\mathcal{E}}_{K} := \inf_{n \in [0,N_{tot}]} \widehat{\mathcal{E}}_{n,K}$, where the latter coefficient as well as the moments $\kappa_{j}$ are defined in lemma \ref{lemsto} and are associated to the gradient $\nabla_{1} L_{\widetilde{\mu}^0}(\pi_{\theta_n},\pi_{\theta_n})$. Assume that $p_{\eta,\eta',K}>0$, where:
$$
p_{\eta,\eta',K} := 1 - \mathcal{N}\left(\frac{\eta'-\eta}{\sqrt{\kappa_{4}}} \cdot \widehat{\mathcal{E}}_{K} \sqrt{K}\right) - \frac{0.6 \kappa_{6}}{\kappa_{4}^{3/2}}\mathcal{E}_{K}.
$$
For a maximum iteration budget of $N_{tot}$, the probability $p_\epsilon$ that there exists a $n \in [0,N_{tot}]$ such that $\pi_{\theta_n}$ is an $\epsilon-$shared equilibrium associated to $\mu^0$ is at least Binomial with parameters $N_{tot}$ and $p_{\eta,\eta',K}$:
$$
p_\epsilon \geq \sum_{i=N}^{N_{tot}} \binom{N_{tot}}{i} p_{\eta,\eta',K}^{i} (1-p_{\eta,\eta',K})^{N_{tot}-i}.
$$
If the gradient has high entropy, this probability is very high if $K$ is large. In the extreme case where its coordinates are the same, we have $\mathcal{E}_{K}=K^{-1/2}$, $\widehat{\mathcal{E}}_{K}=1$ and therefore:
$$
p_{\eta,\eta',K} = 1 - \mathcal{N}\left(\frac{\eta'-\eta}{\sqrt{\kappa_{4}}} \sqrt{K}\right) - \frac{0.6 \kappa_{6}}{\kappa_{4}^{3/2} \sqrt{K}},
$$
and in particular:
$$
\lim_{K \to +\infty} p_{\eta,\eta',K} = 1.
$$
On the other end of the spectrum, if all gradient coordinates but one are zero, then $\mathcal{E}_{K}=1$, $\widehat{\mathcal{E}}_{K}=K^{-1/2}$ and therefore:
$$
p_{\eta,\eta',K} = 1 - \mathcal{N}\left(\frac{\eta'-\eta}{\sqrt{\kappa_{4}}}\right) - \frac{0.6 \kappa_{6}}{\kappa_{4}^{3/2}}.
$$
\end{theorem}

Note that in theorem \ref{gradet3}, we have considered a worst-case probability via the definition of $\widehat{\mathcal{E}}_{K}$, $\mathcal{E}_{K}$, but we can write down the same result where the probability at each iteration is time-dependent and based on the time-dependent quantities $\widehat{\mathcal{E}}_{n,K}$, $\mathcal{E}_{n,K}$.

Corollary \ref{gradetc} stills holds for theorem \ref{gradet2}. For theorem \ref{gradet3}, the following statement holds on a set which probability can be computed explicitly: "for any $\epsilon$, there exists a $N_\epsilon$ such that for all $n \geq N_\epsilon$, $\theta_n$ is an $\epsilon$-shared equilibrium". As long as $p_{\eta,\eta',K}>0$, we can let $N_{tot}$ go to infinity and get that every converging subsequence $(\theta_{m_n})$ converges to a shared equilibrium. If there exists a unique shared equilibrium, $(\theta_n)$ converges to that equilibrium due to all converging subsequences having the same limit point.

\subsection{Empirical analysis of game components}
\label{secgamec}

In this section we show how modern game theoretical tools, and precisely differentiable games and their potential-Hamiltonian decomposition introduced in \cite{dmg}, can help us analyze the nature of our financial game between LPs and LTs. Our contribution is definition \ref{defcomp}, where we introduce weights that quantify to which extent the game is potential or Hamiltonian. We notice that it is important, in order to analyze the interactions between market players, to remove the players' self-interactions from the potential term of the game Jacobian when comparing it to its Hamiltonian counterpart, so as to obtain a fair comparison between the two game components that only involves interactions across players $i \neq j$.

A differentiable game is a game where players' pure strategies are vectors of real numbers. For example, these vectors can be weights of a neural network, and hence games where each player controls its own neural network are examples of differentiable games. In this section, we write $\theta_i \in \R^d$ the weights of player $i$, and $\pi_{\theta_i}$ its neural net-driven policy. \cite{dmg} introduce two quantities of interest. One is the game gradient $\mathcal{G}$ quantifying the sensitivity of each player's expected utility $V_i$ w.r.t. its own strategy $\theta_i$:
\begin{align}
    \mathcal{G}(\bm{\theta}):=\left( \nabla_{\theta_i} V_i\right)_{i \in [1,n]}.
\end{align}

The second is the game Jacobian defined as the gradient of $\mathcal{G}$, namely the square (block) matrix $\mathcal{J}(\bm{\theta})$\footnote{we denote $\mathcal{J}(\bm{\theta})$ by $\mathcal{J}$ when there is no ambiguity.} of size $\sum_{i=1}^n |\theta_i|$ such that its block $(i,j)$ is the matrix:
\begin{align}
\label{gjac}
    [\mathcal{J}(\bm{\theta})]_{ij}:= \nabla_{\theta_j} \mathcal{G}_i= \nabla^2_{\theta_i,\theta_j} V_i.
\end{align}

What distinguishes a game from classical function optimization is that in general, there exists no function which admits $\mathcal{G}$ as its gradient, unless the game is potential. In the seminal paper on potential games \cite{potential}, potential games are defined as games where the change in the utility of a player $i$ due to a change in its own strategy is equal to that of some function called the potential. It is proved that if utilities $V_i$ are smooth and $\theta_i \in \R$, a game is potential if and only if:
\begin{align}
\label{pot}
\frac{\partial^2 V_i}{\partial \theta_i \partial \theta_j}(\bm{\theta}) = \frac{\partial^2 V_j}{\partial \theta_i \partial \theta_j}(\bm{\theta}) \hspace{5mm} \forall i,j.  
\end{align}

Potential games constitute an important class of games, and in short, correspond to games where there is an underlying "skill" that players can learn (the potential). Note that a game can be potential and zero-sum, for example the 2 player game where utilities of players playing respectively $x$ and $y$ are $f(x)-f(y)$ and $f(y)-f(x)$ for some function $f$. In that case, a potential function is $f$, which both players are trying to maximize, and hence $f$ can be thought of as the game skill that both players are trying to master.

In the case $d=1$ and with the definition of $\mathcal{J}$ in (\ref{gjac}), an equivalent formulation of (\ref{pot}) is that $\mathcal{J}$ is a symmetric matrix. An extension to the case $\theta_i \in \R^d$ with $d>1$ was considered in \cite{dmg}: defining $\mathcal{J}$ as in (\ref{gjac}), we get that the game is potential if and only if $\mathcal{J}$ is symmetric. Since any matrix $M$ can be decomposed uniquely into symmetric and antisymmetric parts, $M=\mathcal{S_*}+\mathcal{A_*}$ with $\mathcal{S_*}=\frac{1}{2}(M+M^T)$, $\mathcal{A_*}=\frac{1}{2}(M-M^T)$, we can do so for $\mathcal{J}$, which yields the decomposition of any differentiable game into potential and Hamiltonian components. Hamiltonian games are in spirit, the same as harmonic games in \cite{dgame}. The game is potential if and only if $\mathcal{A_*}=0$, and is Hamiltonian if and only if $\mathcal{S_*}=0$. This is discussed in \cite{dmg}.

\begin{definition}
\cite{dmg} Let $\mathcal{J}:=\mathcal{S_*}+\mathcal{A_*}$ the unique decomposition of $\mathcal{J}$ into symmetric and antisymmetric matrices. The game is Hamiltonian if $\mathcal{S_*}=0$, and potential if $\mathcal{A_*}=0$.
\end{definition}

\begin{definition}
\label{dra}
Let $\mathcal{J}:=\mathcal{D}+\mathcal{S}+\mathcal{A}$ the unique decomposition of $\mathcal{J}$ into a diagonal matrix, a symmetric matrix with zero diagonal, and an antisymmetric matrix. We say that the game is weak Hamiltonian if $\mathcal{S}=0$.
\end{definition}

By our previous discussion and definition \ref{dra}, we have $\mathcal{S_*}=\mathcal{S}+\mathcal{D}$ and $\mathcal{A_*}=\mathcal{A}$. Therefore, the game is potential if and only if $\mathcal{A_*}=0$, if and only if $\mathcal{A}=0$. Hamiltonian and weak Hamiltonian games differ in that the former require both $\mathcal{D}$ and $\mathcal{S}$ to be zero, whereas the latter only require $\mathcal{S}$ to be zero. Since we are interested in quantifying interactions among players, we exclude its diagonal $\mathcal{D}$ from the symmetric component of the game Jacobian, which quantifies the self-interaction of a player with itself $\nabla^2_{\theta_i} V_i$. Doing so yields that both matrices $\mathcal{S}$ and $\mathcal{A}$ in definition \ref{dra} only include interactions across players $i \neq j$ as both have zero diagonal elements. This makes the comparison between the potential and weak Hamiltonian components fair. Not removing the players' self-interactions from our weights in definition \ref{defcomp} would make the symmetric component dominate artificially its antisymmetric counterpart, which is also what we see in practice. 

In our case, $\theta_i$ are the weights of a neural network, and therefore it is impossible to store these matrices since they are too large. Instead, we can compute the "Jacobian-gradient-product" $\mathcal{J}^T\mathcal{G}$, simply obtained by taking the gradient of $\frac{1}{2} ||\mathcal{G}||^2$ in modern Automatic Adjoint Differentiation (AAD) frameworks such as tensorflow. From there, we can compute its symmetric and antisymmetric parts $\mathcal{S}\mathcal{G}$ and $\mathcal{A}^T\mathcal{G}$ (cf. \cite{dmg}, appendix A). The latter constitute the interaction of the matrices $\mathcal{A}$ and $\mathcal{S}$ with the gradient $\mathcal{G}$, and allow us to approximate their norms by the norms of $\mathcal{A}^T\mathcal{G}$ and  $\mathcal{S}\mathcal{G}$, which we will use as a practical replacement for the norms of $\mathcal{A}$ and $\mathcal{S}$. The quantity $\mathcal{J}^T\mathcal{G}$ is the main building block of consensus optimization \cite{co}, whereas \cite{dmg} argues that using the antisymmetric part only $\mathcal{A}^T\mathcal{G}$ enables convergence to stable fixed points (Symplectic Gradient Adjustment).

This framework gives us an empirical way to quantify the nature the interactions between LP and LT agents, that we introduce in definition \ref{defcomp}. The weights $\omega^A$, $\omega^S$ are computed empirically during training, and tell us to which extent the interactions between two players are potential, or (weak) Hamiltonian. 

\begin{definition}
\label{defcomp}
Let $\mathcal{P}$ be a subset of players in $[1,n]$ associated to the strategy vector $\bm{\theta}$, and denote $\widetilde{\mathcal{J}}$ the restriction of the game Jacobian $\mathcal{J}$ to $\mathcal{P}$. Let $\widetilde{\mathcal{J}}:=\widetilde{\mathcal{D}}+\widetilde{\mathcal{S}}+\widetilde{\mathcal{A}}$ be the (unique) decomposition of $\widetilde{\mathcal{J}}$ into diagonal, symmetric (with zero diagonal) and antisymmetric matrices. We define the Hamiltonian and potential weights of the game as:
\begin{align}
\begin{split}
    \omega^A(\bm{\theta}):= \frac{||\widetilde{\mathcal{A}}^T\mathcal{G}||}{||\widetilde{\mathcal{A}}^T\mathcal{G}||+||\widetilde{\mathcal{S}}\mathcal{G}||}, \hspace{4mm}
    \omega^S(\bm{\theta}):= 1-\omega^A(\bm{\theta}).
\end{split}
\end{align}
\end{definition}

Note that by \cite{ramponi,shenHessian}, the game Jacobian $\mathcal{J}$ is computed, in the case of RL, for $i \neq j$:
$$
[\mathcal{J}(\bm{\theta})]_{ij} = \frac{1}{B}\sum_{b=1}^{B} \sum_{t=0}^{T} \zeta^t \RR^{i}_{t,b} \sum_{t'=0}^{t} \nabla_{\theta_i} \ln \pi_{\theta_i} \left(a^{(i)}_{t',b}|s^{(i)}_{t',b}\right) \nabla_{\theta_j} \ln \pi_{\theta_j} \left(a^{(j)}_{t',b}|s^{(j)}_{t',b}\right)^T.
$$
where $B$ is the number of episodes sampled and:
$$
[\mathcal{J}(\bm{\theta})]_{ii} = \frac{1}{B}\sum_{b=1}^{B} \nabla_{\theta_i} g^i_b \nabla_{\theta_i} \ln \pi_{\theta_i} \left(a^{(i)}_{t',b}|s^{(i)}_{t',b}\right) + \nabla^2_{\theta_i} g^i_b,
$$
$$
g^i_b:=\frac{1}{B}\sum_{b=1}^{B} \sum_{t=0}^{T} \zeta^t \RR^{i}_{t,b} \sum_{t'=0}^{t} \nabla_{\theta_i} \ln \pi_{\theta_i} \left(a^{(i)}_{t',b}|s^{(i)}_{t',b}\right).
$$
In figure \ref{figzeyu1} we illustrate the Hamiltonian weight of definition \ref{defcomp} during training, in two situations. 

In the first case (left), we consider a game between a PnL-driven LP ($\omega=1$) and a PnL-driven LT ($\omega=1$), and no ECN. We see that the game is quite balanced with $\omega^A \approx 0.6$. Remember that in a potential game, there exists a quantity (the potential) that both players are trying to maximize, and in particular, by the example we discussed earlier, a game can be both zero-sum and potential. In this case, the potential component $\omega^S \approx 0.4$ comes from the fact that there is a common skill that both the LP and LT are trying to learn, namely how to earn PnL. The Hamiltonian component comes from the fact that learning this skill does not happen independently from the other player, i.e. there is competition, or coupling, between both players, where the PnL is a function of both players' strategies. We then introduce a second LP, identical to the first. We see the interesting observation that the Hamiltonian weight between the LP and the LT decreases to $0.2$, while that between the two LPs reaches $0.6$. The interpretation is that the competition that was initially occurring between the LP and the LT in the 1v1 case, switches to taking place between the two LPs in the second phase, thus benefiting the LT. In this case, the link between the two LPs becomes slightly more adversarial, which makes the link between the LP and the LT almost cooperative. In simple terms, competition between service providers benefits the customer.  

In the second situation (right), we compare the 1v1 game between a PnL-driven LP and a PnL-driven LT discussed above with a 1v1 game between a PnL-driven LP and a flow-driven LT ($\omega=0$), with even quantity targets $q^a=q^b=0.5$. We see that for the latter case, there is almost no potential component, $\omega^A \approx 1$. This is expected since the objective of the flow LT is completely unrelated to that of the LP, which is PnL. Hence, it is expected that there is no common skill that both players are trying to learn.

\begin{figure}[ht]
  \centering
  \begin{subfigure}[b]{0.49 \textwidth}
        \centering
        \includegraphics[width=\textwidth]{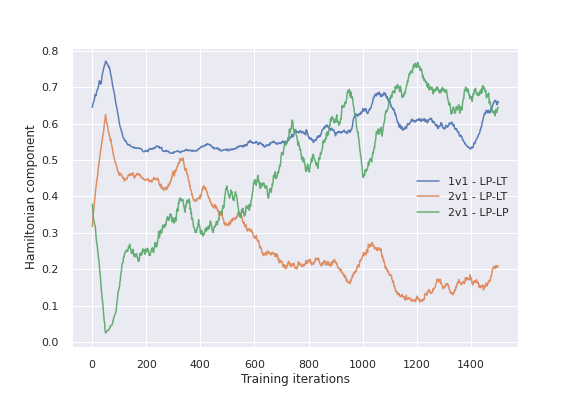}
    \end{subfigure}
    \hfill
    \begin{subfigure}[b]{0.49 \textwidth}
        \centering
        \includegraphics[width=\textwidth]{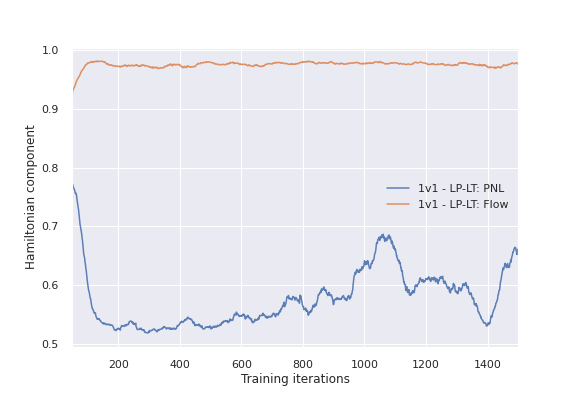}
    \end{subfigure}
  \caption{(Left) Hamiltonian component $\omega^A$ for a game between 1 LP and 1 LT (1v1), and between 2 LPs and 1 LT (2v1). The link LP-LT is fairly balanced in the 1v1 case. $\omega^S \approx 0.4$ since there is a skill that both players are trying to learn (earning PnL), $\omega^A \approx 0.6$ comes from the competition between players. This link vanishes in the 2v1 case where the competition occurs between the two LPs instead: competition between service providers benefits the customer. (Right) Hamiltonian component $\omega^A$ for a game between 1 LP and 1 PnL LT ($\omega=1$), and between 1 LP and 1 Flow LT ($\omega=0$). In the latter case, $\omega^S \approx 0$ since the objective of the LT is of different nature than that of the LP.}
  \label{figzeyu1}
\end{figure}

\section{Calibration of agent supertypes to equilibrium targets}
\label{seccalib}

In this section, we present a novel RL-based equilibrium calibration algorithm which we found performed well compared to a Bayesian optimization baseline (section \ref{secexpcalib1}).

Calibration refers to acting on the supertype profile $\bm{\Lambda}:=(\bm{\Lambda^{LP}}, \bm{\Lambda^{LT}})$ so as to match externally specified targets on the shared equilibrium. In a game that satisfies the conditions of theorem \ref{n11} (see also theorems \ref{gradet}, \ref{gradet2}, \ref{gradet3}), agents will reach a shared equilibrium associated to $\bm{\Lambda}$. For the simulation to accurately model specific real world observations, we would like the emergent behavior of agents in that equilibrium to satisfy certain constraints. For example, these constraints can be on LPs' market share or individual trade distribution. As mentioned in the recent work \cite{Avegliano2019-vx}, there is currently no consensus on how to calibrate parameters of agent-based models. Most methods studied so far build a surrogate of the simulator \cite{Avegliano2019-vx,surlamp}. The difficulty is that for every choice of $\bm{\Lambda}$, one should in principle train agents until equilibrium is reached and record the associated calibration loss, and repeat this process until the loss is small enough, which is prohibitively expensive. The baseline we consider in our experiments follows this philosophy by periodically trying new $\bm{\Lambda}$ obtained via Bayesian optimization (BO). One issue is that BO can potentially perform large moves in the supertype space, which could prevent the shared policy to correctly learn an equilibrium since it would not be given sufficient time to adapt. Another issue is that BO becomes slower over time due to the necessity to perform matrix inversion of a matrix which size grows over time.

Our solution is therefore to smoothly vary $\bm{\Lambda}$ during training: we introduce a RL calibrator agent whose goal is to optimally pick $\bm{\Lambda}$ and who learns jointly with RL agents learning a shared equilibrium, but under a slower timescale. The two-timescale stochastic approximation framework is widely used in RL \cite{2ts,aapkt} and is well-suited to our problem as it allows the RL calibrator's policy to be updated more slowly than the agents' shared policy, yet simultaneously, thus giving enough time to agents to approximately reach an equilibrium. This RL-based formulation allows us to further exploit smoothness properties of specific RL algorithms such as PPO \cite{ppo}, where a KL penalty controls the policy update. Since the calibrator's policy is stochastic, this is a distributional approach (in that at every training iteration, you have a distribution of supertype profiles rather than a fixed one) which will contribute to further smooth the objective function in (\ref{vcalib}), a technique which is also used in the context of evolution strategies \cite{es}. Note that our formulation of section \ref{secpomg} is general enough to accommodate the case where $\bm{\Lambda}$ is a distribution $f(\bm{\Lambda})$ over supertype profiles: indeed, define new supertypes $\widetilde{\Lambda_i}:=f_i$, where $f_1$ is the marginal distribution of $\Lambda_1$, and for $i \geq 2$, $f_i$ is the distribution of $\Lambda_i$ conditional on $(\Lambda_k)_{k \leq i-1}$ (induced by $f$). This means that instead of a fixed $\bm{\Lambda}$, one can choose a distribution $f(\bm{\Lambda})$ by simply defining supertypes appropriately, and in that case it is important to see that the shared policy at equilibrium will depend on the distribution $f$ rather than on a fixed supertype profile. 

\textbf{RL formulation of the calibrator.} The RL calibrator's \textit{state} is the current supertype profile across agents $\bm{\Lambda}$, and its \textit{action} is a vector of increments $\bm{\delta \Lambda}$ to apply to the supertypes, resulting in new supertypes $\bm{\Lambda}+\bm{\delta \Lambda}$, where we assume that $\Lambda_i$ takes value in some subset $\mathcal{S}^{\Lambda_i}$ of $\mathbb{R}^d$. We denote $\pi^{\Lambda}(\bm{\delta \Lambda}| \bm{\Lambda})$ the calibrator policy, i.e. the distribution of $\bm{\delta \Lambda}$ conditional on $\bm{\Lambda}$. The agent types will then be sampled from the new supertype $\bm{\Lambda}+\bm{\delta \Lambda}$, and a reward $r^{cal}$ will be observed. The latter is a random variable that depends on the state-action-type agent trajectories, namely the random variable $\bm{z_{t}}:=(\bm{z^{LP}_{t}},\bm{z^{LT}_{t}})$ defined in section \ref{secpomg}. The calibrator reward $r^{cal}$ is composed of $K$ externally specified targets $f^{(k)}_{*} \in \mathbb{R}$ for functions of the form $f^{(k)}_{cal}((\bm{z_t})_{t \geq 0})$, and is defined as a weighted sum of the reciprocal of loss functions $\ell_k$\footnote{$\ell_k^{-1}$ denotes the reciprocal of $\ell_k$.}:
\begin{align}
\label{rewcalib}
r^{cal} = \sum_{k=1}^K w_k \ell_k^{-1}(f^{(k)}_{*}-f^{(k)}_{cal}((\bm{z_{t}})_{t \geq 0})).
\end{align}
This approach is in line with the literature on "learning to learn" \cite{l2l,lto}, since the goal of the RL calibrator is to learn optimal directions $\bm{\delta \Lambda}$ to take in the supertype space, given a current location $\bm{\Lambda}$. The result is algorithm \ref{calsheq}, where at stage $n=1$, the supertype profile $\bm{\Lambda_1}$ is sampled across episodes $b$ as $\bm{\Lambda^b_1} = \bm{\Lambda_{0}}+\bm{\delta \Lambda^b}$, with $\bm{\delta \Lambda^b} \sim \pi_1^{\Lambda}(\cdot | \bm{\Lambda_{0}})$ and where we denote $\widetilde{\pi}_1^{\Lambda}:=\bm{\Lambda_{0}} + \pi_1^{\Lambda}(\cdot | \bm{\Lambda_{0}})$ the resulting distribution of $\bm{\Lambda_1}$. Then, we run multi-agent episodes $b$ according to (\ref{vi2}), each one of them with its supertype profile $\bm{\Lambda^b_1}$, and record the rewards $r^{cal}_b$, thus corresponding to the calibrator state $\bm{\Lambda_{0}}$, and actions $\bm{\delta \Lambda^b}$. The process is repeated, yielding for each episode $b$ at stage $n \geq 2$, $\bm{\Lambda^b_n} \sim \bm{\Lambda^b_{n-1}}+\pi_n^{\Lambda}(\cdot | \bm{\Lambda^b_{n-1}})$, resulting in a distribution $\widetilde{\pi}_n^{\Lambda}$ for $\bm{\Lambda_n}$, empirically observed through the sampled $\{\bm{\Lambda^b_n}\}_{b=1..B}$. As a result, at stage $n$, our aim is to find a calibrator policy $\pi^{\Lambda}_n$ that maximizes the objective in (\ref{vcalib}).
\begin{align}
\label{vcalib}
V^{calib}(\pi_n,\pi_n^{\Lambda}):=\EE_{\substack{\bm{\Lambda} \sim \widetilde{\pi}_{n-1}^{\Lambda},\hspace{0.5mm} \bm{\Lambda'} \sim \pi_n^{\Lambda}(\cdot | \bm{\Lambda})+\bm{\Lambda},\hspace{0.5mm}\lambda_i \sim p_{\Lambda'_i}, \hspace{0.5mm} a^{(i)}_t \sim \pi_n(\cdot|\cdot,\lambda_i)}} \left[ r^{cal}\right].
\end{align}

Our RL formulation would also allow to consider the sum of calibrator rewards obtained over multiple training iterations $n$, which can sometimes help stabilize learning as observed in \cite{lto}. 

\begin{algorithm}[ht]
\caption{\textbf{(CALSHEQ) Calibration of Shared Equilibria}}\label{calsheq}
\textbf{Input:} {learning rates $(\alpha_n^{cal})$, $(\alpha_n^{shared})$ satisfying assumption \ref{alr}, initial calibrator and shared policies $\pi_0^{\Lambda}$, $\pi_0$, initial supertype profile $\bm{\Lambda^b_{0}} = \bm{\Lambda_{0}}$ across episodes $b \in [1,B]$.}
\begin{algorithmic}[1]
\While{$\pi_n^{\Lambda}$, $\pi_n$ not converged}
    \For{each episode $b \in [1,B]$}
        \State{Sample supertype increment $\bm{\delta \Lambda^b} \sim \pi_n^{\Lambda}(\cdot | \bm{\Lambda^b_{n-1}})$ and set $\bm{ \Lambda^b_n}:=\bm{\Lambda^b_{n-1}}+\bm{\delta \Lambda^b}$}
        \State{Sample multi-agent episode with supertype profile $\bm{\Lambda^b_n}$ and shared policy $\pi_n$, with $\lambda_i \sim p_{\Lambda^b_{n,i}}$, $a^{(i)}_t \sim \pi_n(\cdot|\cdot,\lambda_i)$, $i \in [1,n^{LP}]$ cf. (\ref{vi2})}
    \EndFor
     \State{update $\pi_n$ with learning rate $\alpha_n^{shared}$ based on gradient (\ref{ut1}).}
    \State{update $\pi_n^{\Lambda}$ with learning rate $\alpha_n^{cal}$ based on gradient (\ref{ut2}) associated to (\ref{vcalib}).}
\EndWhile
\end{algorithmic}
\end{algorithm}

\textbf{Two-timescale update rule.} Let $\theta_n$ and $\theta^\Lambda_n$ be the (neural net) parameters of the shared policy and calibrator policy at stage $n$, so that $\pi_n$ and $\pi_n^{\Lambda}$ are shorthand for, respectively, $\pi_{\theta_n}$ and $\pi_{\theta^\Lambda_n}^{\Lambda}$. Let $\alpha_n^{shared}$ and $\alpha_n^{cal}$ be the learning rates at stage $n$, according to which $\theta_n$ and $\theta^\Lambda_n$ will be updated. The idea of two-timescale stochastic approximation is that from the point of view of the shared policy, the distribution of supertypes being chosen by the calibrator should be seen as "quasi-static", which, informally, will give enough time to the shared policy to approximately reach an equilibrium. This is reflected in the condition $\alpha_n^{cal} \stackrel{n \to + \infty}{=} o(\alpha_n^{shared})$ in assumption \ref{alr}, standard under the two-timescale framework \cite{borkar}. $\pi^{\Lambda}$ is then updated based on the objective (\ref{vcalib}) using a classical RL gradient update. This process ensures that $\pi^{\Lambda}$ is updated smoothly during training and learns optimal directions to take in the supertype space, benefiting from the multiple locations $\bm{\Lambda^b_n}$ experienced across episodes and over training iterations. Our framework shares some similarities with the work on learning to optimize in swarms of particles \cite{ltos}, since at each stage $n$, we have a distribution of supertype profiles empirically observed through the $B$ episodes, where each $\bm{\Lambda^b_n}$ can be seen as a particle.

The idea behind assumption \ref{alr2} is to have the distribution of $\bm{\Lambda_n}$ only depend on the past through $\theta^{\Lambda}_n$. We need it to express our parameter updates in the standard two-timescale framework (\ref{ut1})-(\ref{ut2}). It can always be achieved by taking $\pi_{\theta}^{\Lambda}$ of the form $\pi_{\theta}^{\Lambda}(\cdot | \bm{\Lambda}):= \widehat{\pi}_{\theta}^{\Lambda}(\cdot) - \bm{\Lambda}$ for some $\widehat{\pi}_{\theta}^{\Lambda}$. More generally, the recursion $\bm{\Lambda_k} \sim \bm{\Lambda_{k-1}}+\pi_{\theta^{\Lambda}_k}^{\Lambda}(\cdot | \bm{\Lambda_{k-1}})$ generates a Markov chain on $\R^d$, and the assumption is that $\widetilde{\pi}_{n-1}^{\Lambda}$ is a stationary distribution of the transition kernel $P_n(x,A):=x+\pi_{\theta^{\Lambda}_n}^{\Lambda}(A-x |x)$, where $A$ is a Borel subset of $\R^d$. Informally, it means that $\widetilde{\pi}_{n}^{\Lambda}$ is mostly driven by $\pi_{n}^{\Lambda}$ as learning progresses. By assumption \ref{alr2}, we can use the following notation for the gradient $\nabla^{shared}_{\theta,B}$ in (\ref{grad}):
\begin{align}
\label{ut0}
\begin{split}
& x_{n} := \nabla^{shared}_{\theta_n,B} + \nu  \nabla \ln \pi_{\theta_n}, \hspace{6mm}
\widehat{x}(\theta_n, \theta^\Lambda_n) := \EE_{\pi_{\theta_n},\widehat{\pi}_{\theta^{\Lambda}_n}^{\Lambda}} [x_{n}], \\
& X_{n+1} := \alpha_n^{shared} (x_n - \widehat{x}(\theta_n, \theta^\Lambda_n)),
\end{split}
\end{align}
where we have allowed log-barrier regularization with parameter $\nu\geq 0$ as in theorem \ref{gradet2}, so that the latter is applicable. We refer to (\ref{logpo}) for the definition of the logarithm of a policy, which is simply the average of the log-policy over all states, types and actions. Note that $x_n$ is an average of $B$ terms, all of which have the same expectation $\widehat{x}(\theta_n, \theta^\Lambda_n)$ (hence $\widehat{x}$ doesn't depend on $B$). The update of $\theta_n$ reads:
\begin{align}
\label{ut1}
\begin{split}
& \theta_{n+1} = \theta_n +  \alpha_n^{shared} \widehat{x}(\theta_n, \theta^\Lambda_n) +  X_{n+1}.
\end{split}
\end{align}

$X$ is a martingale difference sequence that represents the noise in our observation of the true gradient, arising due to finite batch size $B<+\infty$. Precisely, $\EE[X_{n+1}|\mathcal{F}_n]=0$, where $\mathcal{F}_n$ is the sigma-algebra generated by the discrete-time processes $\theta$ and $\theta^\Lambda$ up to time $n$. Similarly, the update of $\theta^\Lambda_n$ reads:
\begin{align}
\label{ut2}
\begin{split}
& \theta^\Lambda_{n+1} = \theta^\Lambda_n +  \alpha_n^{cal} \widehat{y}(\theta_n, \theta^\Lambda_n) +  Y_{n+1},\\
&y_n:=\frac{1}{B}\sum_{b=1}^{B} \nabla_{\theta^\Lambda} \ln \pi_{n}^{\Lambda}(\bm{\delta \Lambda^b} | \bm{\Lambda^b_n}) \ r^{cal}_{b},\\
&\widehat{y}(\theta_n, \theta^\Lambda_n) := \EE_{\pi_{\theta_n},\widehat{\pi}_{\theta^{\Lambda}_n}^{\Lambda}} [y_{n}], \hspace{4mm} Y_{n+1} := \alpha_n^{cal} (y_n - \widehat{y}(\theta_n, \theta^\Lambda_n)),
\end{split}
\end{align}
where $y_n$ is the gradient of $V^{calib}$ with respect to the calibrator policy parameter $\theta^\Lambda_n$ in (\ref{vcalib}). Precisely:
\begin{align}
\label{pgr}
\widehat{y}(\theta_n, \theta^\Lambda_n) = \nabla_{2} V^{calib}(\pi_{\theta_n}, \pi^\Lambda_{\theta^{\Lambda}_n}),
\end{align}

where we remind our notation that for a function of multiple arguments, $\nabla_{k}$ is the gradient with respect to the $k^{th}$ argument. 

\begin{assumption}
\label{alr}
\textbf{(Learning rate)}. The learning rate $(\alpha_n^{cal})$ has infinite $\ell^1$ norm, and $\lim_{n \to +\infty} \frac{\alpha_n^{cal}}{\alpha_n^{shared}}=0$. Further, $\alpha_n^{shared}< \frac{2}{\beta_{\nu}(1+\phi_{max})}$, where the latter are given in theorem \ref{gradet2}.
\end{assumption}

\begin{assumption}
\label{alr2}
\textbf{(Markovian property of the update rule)}. Let $\bm{\Lambda_0}$ and $\theta^{\Lambda}_0$ be fixed, and the sequence $(\theta^{\Lambda}_n)_{n>0}$ generated as in (\ref{ut2}). There exists a family of distributions $\widehat{\pi}_{\theta}^{\Lambda}$ parametrized by $\theta$, such that for any $n \geq 0$, $\widetilde{\pi}_{n}^{\Lambda} = \widehat{\pi}_{\theta^{\Lambda}_n}^{\Lambda}$, where we recall that $\widetilde{\pi}_{n}^{\Lambda}$ is the distribution of $\bm{\Lambda_n}$ generated through the recursion $\bm{\Lambda_k} \sim \bm{\Lambda_{k-1}}+\pi_{\theta^{\Lambda}_k}^{\Lambda}(\cdot | \bm{\Lambda_{k-1}})$.
\end{assumption}

\begin{assumption}
\label{alr6} \textbf{(Extended transitivity of the calibration game)}. The 2-player symmetric game with payoff $u(\theta^\Lambda_1,\theta^\Lambda_2):=V^{calib}(f_{eq}(\theta^\Lambda_2), \theta^\Lambda_1)$ is extended transitive. By theorem \ref{n11}, it admits at least one symmetric equilibrium. We assume that either this equilibrium is unique, or there is a countable number of equilibria with the additional assumption that the extended transitivity mapping $\Phi_{cal}$ is continuous.
\end{assumption}

\begin{assumption}
\label{alr3} \textbf{(Regularity of shared equilibria)}. Let $\widehat{\pi}_{\theta^\Lambda}^{\Lambda}$ be defined in assumption \ref{alr2}. For every $\theta^\Lambda$, the game $\widehat{V}$ with supertype sampled from $\widehat{\pi}_{\theta^\Lambda}^{\Lambda}$ admits at most one shared equilibrium $f_{eq}(\theta^{\Lambda})$, and $f_{eq}$ is Lipschitz.
\end{assumption}

\begin{assumption}
\label{alr4} \textbf{(Uniformly bounded noise)}. The noise processes $X$ and $Y$ in (\ref{ut0}), (\ref{ut2}) are multiplicative and uniformly bounded, namely:
\begin{align*}
X_{n+1} = \alpha_n^{shared} \widehat{x}(\theta_n, \theta^\Lambda_n) \odot \phi^{shared}_n, \hspace{4mm} Y_{n+1} = \alpha_n^{cal} \widehat{y}(\theta_n, \theta^\Lambda_n) \odot \phi^{cal}_n,
\end{align*}
where $\phi^{shared}_n$, $\phi^{cal}_n$ are vectors of i.i.d. mean-zero random variables upper bounded by $\phi_{max}$, and lower bounded by $\phi_{min}>-1$ almost surely.
\end{assumption}

\begin{assumption}
\label{alr5} \textbf{(Regularity of the calibrator policy)}. The mapping $\nabla \ln \pi^\Lambda_{\theta^{\Lambda}}$ is Lipschitz, and the calibrator reward $r^{cal}$ in (\ref{vcalib}) is uniformly bounded.
\end{assumption}

\begin{theorem}
\label{alrt}
Let assumptions \ref{alr}, \ref{alr2}, \ref{alr6}, \ref{alr3}, \ref{alr4}, \ref{alr5} hold true, as well as the conditions of theorem \ref{gradet2}. Then, the iterates (\ref{ut1})-(\ref{ut2}) generated by algorithm \ref{calsheq} converge almost surely to the point $(f_{eq}(\theta^\Lambda_*),\theta^\Lambda_*)$, where $\theta^\Lambda_*$ is a critical point of the mapping $\theta^{\Lambda} \to V^{calib}(\pi_{f_{eq}(\theta^\Lambda_*)}, \pi^\Lambda_{\theta^{\Lambda}})$ in (\ref{vcalib}), i.e. $ \nabla_{2} V^{calib}(\pi_{f_{eq}(\theta^\Lambda_*)}, \pi^\Lambda_{\theta^{\Lambda}_*})=0$.
\end{theorem}

\textbf{Discussion of the assumptions used in theorem \ref{alrt}.} We already discussed how assumption \ref{alr2} allows us to express the update rules of $\theta$ and $\theta^\Lambda$ in the standard two-timescale setting. The idea of two-timescale analysis in \cite{borkar} is to conclude on the convergence of the iterates (\ref{ut1})-(\ref{ut2}) by analyzing the ordinary differential equations (ODE)$\frac{d\theta}{dt} = \widehat{x}(\theta, \theta^\Lambda)$ (for fixed $\theta^\Lambda$), and $\frac{d\theta^\Lambda}{dt} = \widehat{y}(f_{eq}(\theta^\Lambda), \theta^\Lambda)$, where $f_{eq}(\theta^\Lambda)$ is the unique global asymptotically stable equilibrium of the former ODE ($f_{eq}$ is also assumed to be Lipschitz). The "time" variable is related to the accumulation of the learning rates, namely the first moment at which either the cumulative sum of $\alpha_n^{cal}$ or $\alpha_n^{shared}$ hits some threshold. If $\alpha_n^{cal}$ is small compared to $\alpha_n^{shared}$, time flows more slowly for the calibrator, which can be seen as static compared to the shared policy. Concretely, for a fixed time $t$, much fewer iterations $n$ are needed for $\sum_{k=0}^n \alpha_k^{shared}$ than for $\sum_{k=0}^n \alpha_k^{cal}$ to reach $t$. 

The result in \cite{borkar} assumes that both ODEs have a unique global asymptotically stable equilibrium, and that both mapping $\widehat{x}$, $\widehat{y}$ are Lipschitz. The Lipschitz requirement is a classical one when analyzing ODEs. The existence of global asymptotically stable equilibria is more severe. In short, it takes for granted that the ODE converge "somewhere stable", already a non-trivial result as seen in section \ref{secsharedeq}. In fact, Borkar's work mentions that the ODE for $\theta^\Lambda$ is allowed to have a global asymptotically stable attractor, in which case $\theta^\Lambda_n$ would converge to one of the elements of this set (see also \cite{leslie}, proposition 4 and theorem 5; \cite{benaim}, corollaries 5.4 and 6.6). In our case, we already proved in theorem \ref{gradet2} and corollary \ref{gradetc} that the update for $\theta_n$ converges to a (shared) equilibrium, provided the learning rate is small enough and the noise is multiplicative and bounded. We will therefore make the assumption \ref{alr3} that this equilibrium is unique, and that we are in the setting of theorem \ref{gradet2} (assumption \ref{alr4}). Note that an "equilibria" for the ODE governing $\theta$ is a point where $\widehat{x}(\theta, \theta^\Lambda) = 0$, i.e. the gradient with respect to $\theta$ is zero by definition of $\widehat{x}$. In principle, this is different from shared equilibria discussed in section \ref{secsharedeq}. However the proof of our theorems \ref{gradet}-\ref{gradet2} shows that all points where the gradient is zero are shared equilibria, due to the specific nature of $\widehat{V}$ which is $\beta$-smooth (i.e. $\beta$-Lipschitz gradient). This is obtained via a "gradient domination" argument as in the single-agent case \cite{convpg}. 

Regarding the learning rates and our assumption \ref{alr}, the core assumption in \cite{borkar} is $\lim_{n \to +\infty} \frac{\alpha_n^{cal}}{\alpha_n^{shared}}=0$, which relates to one timescale being slower than another as previously discussed. It is also assumed that they satisfy the usual Robbins-Monro conditions, namely infinite $\ell_1$ norm and finite $\ell_2$ norm \footnote{we remind that $||\alpha||_{\ell_p}:= \left(\sum_{n=0}^\infty |\alpha_n|^p\right)^{\frac{1}{p}}$.}. The latter is only assumed to show that the noise terms $X$ and $Y$ can be asymptotically neglected. In our case, because we assume that the noise is uniformly bounded and multiplicative, we get convergence without assuming $\ell_2$ integrability, and theorem \ref{gradet2} holds as long as we have infinite $\ell_1$ norm and $\alpha_n^{shared}$ uniformly bounded by a small enough constant, which in particular allows $\alpha_n^{shared}$ to be constant. The requirement $\alpha_n^{cal} \stackrel{n \to + \infty}{=} o(\alpha_n^{shared})$ implies that $\alpha_n^{cal} \to 0$.

The remaining element to discuss is global asymptotic stability for the ODE governing $\theta^\Lambda$, $\frac{d\theta^\Lambda}{dt} = \widehat{y}(f_{eq}(\theta^\Lambda), \theta^\Lambda)$, and its discrete-time counterpart:
\begin{align}
\label{conv2t}
\widetilde{\theta}^\Lambda_{n+1} = \widetilde{\theta}^\Lambda_n +  \alpha_n^{cal} \widehat{y}(f_{eq}(\widetilde{\theta}^\Lambda_n), \widetilde{\theta}^\Lambda_n).
\end{align}

If the latter were a classical gradient update, we would be done, since under some regularity conditions, the latter would converge to a critical point of the update rule, i.e. a point $\theta^\Lambda_*$ such that  $\widehat{y}(f_{eq}(\theta^\Lambda_*), \theta^\Lambda_*)=0$. However this is not quite the case, since there are two variables $f_{eq}(\widetilde{\theta}^\Lambda_n)$ and $\widetilde{\theta}^\Lambda_n$, but we are only taking the partial gradient with respect to the second one, cf. (\ref{pgr}). In plain words, given a current supertype $\widetilde{\theta}^\Lambda_n$, and a current shared policy $f_{eq}(\widetilde{\theta}^\Lambda_n)$, taking a gradient step (\ref{conv2t}) ensures - under a smoothness condition - that we can find a new supertype $\widetilde{\theta}^\Lambda_{n+1}$ that increases the calibrator reward $V^{calib}$ provided that the shared policy remains $f_{eq}(\widetilde{\theta}^\Lambda_n)$. But the shared policy also jumps to $f_{eq}(\widetilde{\theta}^\Lambda_{n+1})$. Therefore, informally, we need to ensure that the gain we obtain when we move $\widetilde{\theta}^\Lambda_{n} \to \widetilde{\theta}^\Lambda_{n+1}$ is not offset by the loss we suffer when moving $f_{eq}(\widetilde{\theta}^\Lambda_{n}) \to f_{eq}(\widetilde{\theta}^\Lambda_{n+1})$. This is exactly the reasoning behind extended transitivity assumption \ref{et}. So, to remediate this issue, we will assume that the game with payoff $u(\theta^\Lambda_1,\theta^\Lambda_2):=V^{calib}(f_{eq}(\theta^\Lambda_2), \theta^\Lambda_1)$ is extended transitive (assumption \ref{alr6}). This can be seen as a more "interpretable" requirement than global asymptotic stability of the ODE, but both are related to convergence of the ODE. The line of argument to prove convergence will then be that of theorem \ref{gradet2}. By theorem \ref{n11}, this "game" has a symmetric equilibrium. By corollary \ref{gradetc}, if this equilibrium is unique, it is globally attracting. By \cite{leslie}, proposition 4 and theorem 5, $\theta^\Lambda_n$ converges to that equilibrium. In the case where there are a countable number of equilibria, \cite{leslie}, proposition 4 and theorem 5 give convergence of $\theta^\Lambda_n$ to one of these equilibria provided that there exists a strict Lyapunov function. The countable set of fixed points was also made in \cite{Tamar2012-ki,Tamar2013VarianceAA} who also employ a two-timescale approach. A Lyapunov function is a continuous function $v$ such that $v(g_t(z))$ is strictly increasing in $t$ whenever $z$ is such that $\widehat{y}(f_{eq}(z),z) \neq 0$, and $g_t(z)$ is the solution of the ODE $\frac{d\theta^\Lambda}{dt} = \widehat{y}(f_{eq}(\theta^\Lambda), \theta^\Lambda)$ started at $z$. But by extended transitivity assumption and the proof of theorem \ref{gradet2}, the extended transitivity mapping $\Phi_\nu$ is such a Lyapunov function (note that this is why we required continuity of $\Phi_\nu$). Indeed, we get:
$$
\Phi_\nu(\theta^\Lambda_{n+1})-\Phi_\nu(\theta^\Lambda_{n}) \geq \varphi_\nu \left(\left(\frac{1}{\alpha_n^{cal}} - \frac{\beta_\nu}{2}\right)||\theta^\Lambda_{n+1}-\theta^\Lambda_{n}||_2^2 \right),
$$
which is positive whenever $\widehat{y}(f_{eq}(\theta^\Lambda_n),\theta^\Lambda_n) \neq 0$. This requires $\beta$-smoothness of $V^{calib}$ as in the proof of theorem \ref{gradet2}. Assumption \ref{alr5} ensures that the mapping $\theta^{\Lambda} \to V^{calib}(\pi_{\theta}, \pi^\Lambda_{\theta^{\Lambda}})$ is $\beta$-smooth.

\section{Experiments}
\label{secexp}

Both shared LP, shared LT and calibrator policies were trained jointly using Proximal Policy Optimization \cite{ppo}, an extension of TRPO \cite{pmlr-v37-schulman15}, with a KL penalty to control the smoothness of policy updates \cite{ppo}. We used configuration parameters in line with \cite{ppo}, that is a clip parameter of 0.3, an adaptive KL penalty with a KL target of $0.01$ (so as to smoothly vary the supertype profile during calibration) and a learning rate of $10^{-4}$. We found that entropy regularization was not specifically helpful in our case. Episodes were taken of length $T=100$ time steps, using $B=60$ parallel runs in between policy updates. As a result, each policy update was performed with a batch size of $n_{agents} \cdot T \cdot B$ for the shared policy (where $n_{agents}$ is the number of agents per policy), and $B$ for the calibrator's policy, together with 30 iterations of stochastic gradient descent. We used for each policy a fully connected neural net with 2 hidden layers, 256 nodes per layer, and $\tanh$ activation. Since our action space is continuous, the outputs of the neural net are the mean and stDev of a normal distribution, which is then used to sample actions probabilistically (the covariance matrix across actions is chosen to be diagonal). We take the discount factor $\zeta=1$.

For all experiments, we check empirically that we reach convergence for $\pi^{LT}$ and $\pi^{LP}$. We investigate the fit to the calibration targets as well as emergent behavior learnt by our agents. Experiments were conducted in the RLlib multi-agent framework \cite{rllib}, ran on AWS using a EC2 C5 24xlarge instance with 96 CPUs, resulting in a training time of approximately half a day per experiment. As discussed previously, we will adopt the terminology \textit{flow LTs} and \textit{PnL LTs} for LTs having a PnL weight of $\omega=0$ and $\omega=1$, respectively.

\subsection{Calibration}
\label{secexpcalib1}

\textbf{Experimental setup.} We consider calibration in the case of stationary LTs: LTs are assumed at every point in time $t$ to either want to buy or sell with equal probability a fixed quantity $q_{LT}$. This corresponds to the case $\omega=0$, $q^a=q^b=0.5$ in section \ref{secrlagents}, here we do not train $\pi^{LT}$ to achieve such objective but directly sample from this simple distribution. We split $500$ LTs into 10 supertypes, supertype $i \in [1,10]$ being associated to $q_{LT} = i$.

LPs are assumed to have PnL weight $\omega=1$. We consider 2 distinct supertypes for 5 to 10 LPs in total depending on the experiment, where the first LP is assigned supertype 1 and the $n^{LP}-1$ others are assigned supertype 2, cf. table \ref{tab1}. Supertypes 1 and 2 are respectively vectors of size 12 and 11, resulting in 23 parameters to calibrate in total. For each supertype we have i) 10 probabilities to be connected to LT supertypes, ii) the LP risk aversion, which is 1 parameter for supertype 1, and the mean/variance of a (clipped) normal random variable for supertype 2. The corresponding ranges for the the calibrator's state $\bm{\Lambda}$ and action $\bm{\delta \Lambda}$ are reported in table \ref{tab3}. In contrast, experimental results in \cite{surlamp} only calibrate 8 or 12 parameters, although not in a RL context. In a given episode, a LP may be connected to some LTs in a given supertype, and disconnected to others: these connectivities are sampled randomly at the beginning of the episode as Bernoulli random variables with the same probability.

The calibration targets we consider are related to i) LP market share, ii) the distribution of per-timestep individual trade quantities that a given LP receives from various LTs. The constraint is on 9 percentiles of the distribution, for each supertype. 

\begin{table}[!ht]
\caption{RL calibrator state and action spaces.}
\label{tab3}
\begin{center}
\begin{scriptsize}
\begin{tabular}{ ccc}

\makecell{Supertype parameter flavor $j$} & \makecell{state $\Lambda_i(j)$ range} & \makecell{action $\delta \Lambda_i(j)$ range}  \\ 
\hline
LT supertype connectivity probability&$[0,1]$&$[-1,1]$ \\  
LP risk aversion Gaussian mean &$[0,5]$&$[-5,5]$ \\ 
LP risk aversion Gaussian stDev &$[0,2]$&$[-2,2]$ \\ 
\hline
\end{tabular}\end{scriptsize}\end{center}
\end{table}

We give in table \ref{tab2} a breakdown of the calibrator reward sub-objectives corresponding to equation (\ref{rewcalib}). For each experiment, all sub-objectives are required to be achieved simultaneously. Precisely, the reward formulation associated to table \ref{tab2} are given below, where we denote $m_{super_1}=m_{super_1}((\bm{z_t})_{t \geq 0})$ the average market share of supertype 1 observed throughout an episode, $m_{total}$ the sum of all LPs market share \footnote{since there is an ECN, this sum is less than or equal to one.}, $\widehat{v}_{super_j}(p)$ the observed $(10p)^{th}\%$ percentile of supertype $j$'s per timestep individual trade quantity distribution for $p \in [1,9]$ (empirical percentile).

In experiment 1: 
\begin{align*}
& r^{cal} = (1+r^{(1)}+0.2r^{(2)})^{-1},\\
& r^{(1)}=\frac{1}{2}(\max(0.15 - m_{super_1},0) + \max(0.8 - m_{total},0)),\\
& r^{(2)}=\frac{1}{9}\sum_{p=1}^9|v_{super_1}(p) - \widehat{v}_{super_1}(p)|, \hspace{3mm}  v_{super_1} = [8,8,8,9,9,9,10,10,10].
\end{align*}
In experiment 2/3: 
\begin{align*}
& r^{cal} = (1+r^{(1)}+0.2r^{(2)}+0.2r^{(3)})^{-1},\\
& r^{(1)}=\frac{1}{2}(\max(0.15 - m_{super_1},0) + \max(0.8 - m_{total},0)),\\
& r^{(j+1)}= \frac{1}{9}\sum_{p=1}^9|v_{super_j}(p) - \widehat{v}_{super_j}(p)|, \hspace{3mm} j \in \{1,2\},\\
& v_{super_1} = [8,8,8,9,9,9,10,10,10], \hspace{3mm}  v_{super_2} = [2,3,3,4,5,5,6,6,7].
\end{align*}
In experiment 4: 
\begin{align*}
& r^{cal} = (1+r^{(1)}+r^{(2)})^{-1}, \hspace{3mm} r^{(1)}=|0.25 - m_{super_1}|, \hspace{3mm} r^{(2)}=\max(0.8 - m_{total},0).
\end{align*}

\begin{table}[!ht]
\caption{Summary of experiment configuration.}
\label{tab1}
\begin{center}
\begin{scriptsize}
\begin{tabular}{cccccc}
Experiment \# & \makecell{\# LP \\ Agents} & \makecell{Budget \# Training  \\Steps ($10^6$)} & \makecell{\# distinct\\ LP Supertypes} & \makecell{\# LP Supertype parameters \\ to be calibrated} &\makecell{Total \# \\Calibration Targets} \\ 
\hline
$1$&$5$&$40$&$2$&$20$&11 \\  
$2$&$5$&$40$&$2$&$20$&20  \\  
$3$&$10$&$20$&$2$&$20$&20  \\  
$4$&$10$&$20$&$2$&$23$&2 \\  
\hline
\end{tabular}\end{scriptsize}\end{center}
\end{table}

\begin{table}[!ht]
\caption{Calibration target breakdown}
\label{tab2}
\begin{center}
\begin{scriptsize}
\begin{tabular}{cccc}
Experiment \# &\makecell{\# Calibration\\ Targets} & \makecell{Calibration Target Type}\\ 
\hline
1&9&\makecell{trade quantity distribution supertype 1 \\ 
percentiles $10\%-90\%$ target $8$, $8$, $8$, $9$, $9$, $9$, $10$, $10$, $10$} \\   
\hline
&2&\makecell{market share supertype 1 $\geq 15\%$ + total market share $\geq 80\%$} \\   
\hline
 2&18&\makecell{trade quantity distribution Supertypes 1+2 \\
 supertype 1 - percentiles $10\%-90\%$ target $8$, $8$, $8$, $9$, $9$, $9$, $10$, $10$, $10$\\
 supertype 2 - percentiles $10\%-90\%$ target $2$, $3$, $3$, $4$, $5$, $5$, $6$, $6$, $7$} \\
 \hline
&2&\makecell{market share supertype 1 $\geq 15\%$ + total market share $\geq 80\%$} \\   
\hline
 3&18&\makecell{trade quantity distribution Supertypes 1+2\\
 supertype 1 - percentiles $10\%-90\%$ target $8$, $8$, $8$, $9$, $9$, $9$, $10$, $10$, $10$\\
 supertype 2 - percentiles $10\%-90\%$ target $2$, $3$, $3$, $4$, $5$, $5$, $6$, $6$, $7$} \\
\hline
&2&\makecell{market share supertype 1 $\geq 15\%$ + total market share $\geq 80\%$} \\   
\hline
4&1&\makecell{market share supertype 1 $=25\%$} \\
\hline
&1&\makecell{total market share $\geq 80\%$} \\ 
\hline
\end{tabular}\end{scriptsize}\end{center}
\end{table}

\textbf{Baseline.} The baseline we consider is Bayesian optimization (BO), a method that has been used for hyperparameter optimization. The latter can be considered as similar to this calibration task, and BO will allow us to periodically record the calibration loss related to a certain choice of supertype $\bm{\Lambda}$, and suggest an optimal point to try next, via building a Gaussian Process-based surrogate of the simulator. 

Every $M$ policy training iterations, we record the calibrator's reward, and use Bayesian optimization to suggest the next best $\bm{\Lambda}$ to try. We empirically noticed that if $M$ was taken too low ($M \sim 10$), the shared policy couldn't adapt as the supertype profile changes were too frequent (and potentially too drastic), thus leading to degenerate behaviors (e.g. LPs not trading at all). We tested values of $M=10$, $M=50$, $M=100$, $M=200$, and opted for $M=100$ as we found it was a good trade-off between doing sufficiently frequent supertype profile updates and at the same time giving enough time to the shared policy to adapt. We chose an acquisition function of upper confidence bound (UCB) type \cite{ucb}. Given the nature of our problem where agents on the shared policy need to be given sufficient time to adapt to a new supertype profile $\bm{\Lambda}$, we opted for a relatively low UCB exploration parameter of $\kappa=0.5$, which we empirically found yielded a good trade-off between exploration and exploitation (taking high exploration coefficient can yield drastic changes in the supertype profile space, which can prevent agents to learn correctly an equilibrium). In figure \ref{f0} we perform an ablation study focused on experiment 1: we look at the impact of the choice of $M$ in the EI (expected improvement) and UCB (exploration parameter of $\kappa=1.5$) cases and find that different choices of $M$ and of the acquisition function yield similar performance. We also look at the case "CALSHEQ\_no\_state" where the calibrator policy directly samples supertype values (rather than increments) without any state information (i.e. the calibrator policy's action is conditioned on a constant), and find that it translates into a significant decrease in performance. We further note that decreasing $M$ has a cost, especially when $\bm{\Lambda}$ is high dimensional, since the BO step will become more and more expensive over time due to the necessity to invert a covariance matrix that gets larger over time. For example, in the case of experiment 1, we observed with $M=1$ that the training hadn't reached the 20M timestep budget after 2 days (for a calibrator reward in line with other values of $M$). The covariance function of the Gaussian process was set to a Matern kernel with $\nu=2.5$.

\textbf{Performance metrics.} We evaluate our findings according to the following three criteria 1) calibrator reward in (\ref{vcalib}), quantifying the accuracy of the equilibrium fit to the target(s), where one equals perfect fit, 2) convergence of LP agents' rewards to an equilibrium and 3) smoothness of the supertype profile $\bm{\Lambda}$ as a function of training iterations, ensuring that equilibria is given sufficient time to be reached, cf. discussion in section \ref{seccalib}.

\textbf{Results.} In figure \ref{f1} we display calibrator and agents' reward evolution during training. It is seen that \textit{CALSHEQ} outperforms BO in that i) the calibrator's rewards converge more smoothly and achieve on average better results in less time, ii) in experiment 4, supertype 1's reward in the BO case converges to a negative value, which should not happen as LPs always have the possibility to earn zero income by making their prices not competitive. The reason for it is that as mentioned in section \ref{seccalib}, BO can potentially perform large moves in the supertype space when searching for a solution, and consequently agents may not be given sufficient time to adapt to new supertype profiles $\bm{\Lambda}$. This fact is further seen in figures \ref{f2}-\ref{f5} where we show a sample of supertype parameters during training. It is seen that \textit{CALSHEQ} smoothly varies these parameters, giving enough time to agents on the shared policy to adapt, and preventing rewards to diverge as previously discussed. 

The RL calibrator's total reward in (\ref{vcalib}) is computed as weighted sum of various sub-objectives. In figures \ref{f7} and \ref{f8}-\ref{f10}, we zoom on the individual components that constitute the overall reward, together with the associated target values. It is seen that \textit{CALSHEQ} converges more smoothly and more accurately than BO to the target values.

\begin{figure}[hb]
  \centering
  \centerline{\includegraphics[width=\columnwidth]{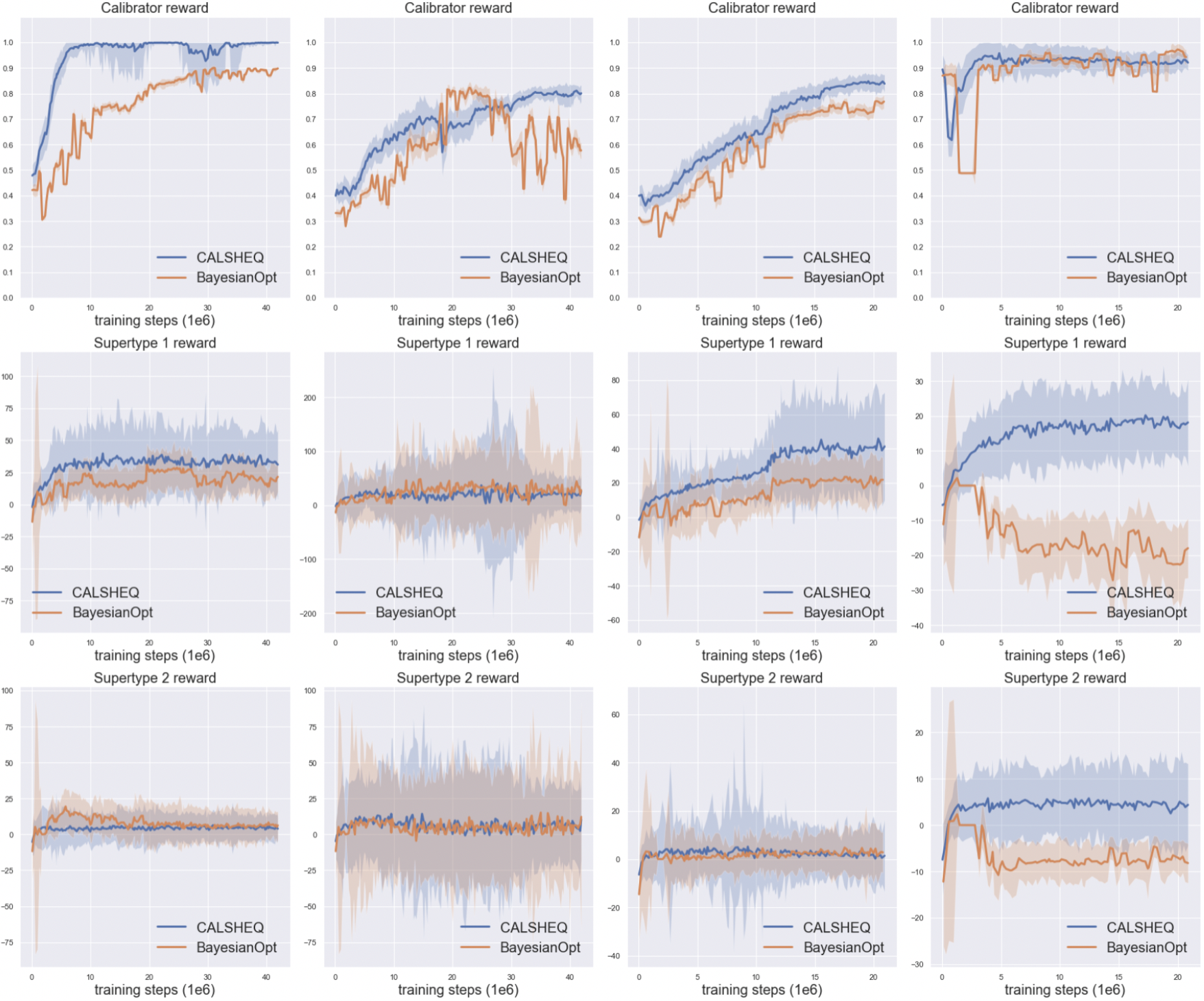}}
  \caption{Cumulative rewards during training for the calibration experiments, averaged over episodes $B$ - Calibrator \textbf{(Top row)}, LP Supertypes 1/2 \textbf{(Mid/Bottom rows)} - experiments 1-2-3-4 (across columns). \textit{CALSHEQ} (ours) and baseline (Bayesian optimization). Shaded area represents $\pm 1$ stDev.}
  \label{f1}
\end{figure} 
\clearpage

\begin{figure}[ht]
  \centering
  \centerline{\includegraphics[width=\columnwidth]{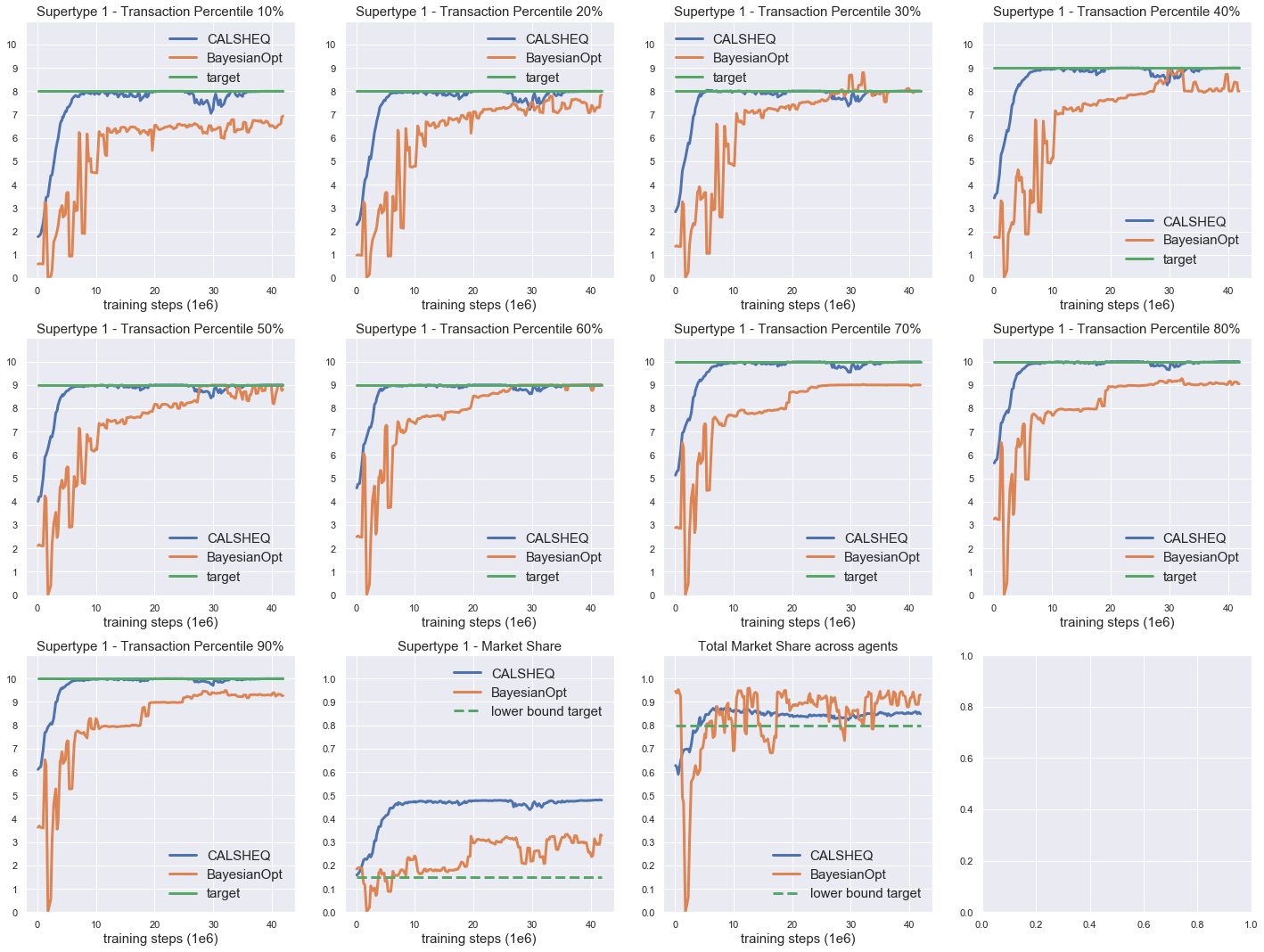}}
  \caption{Calibration experiment 1 - Calibration target fit for trade quantity distribution percentile and market share during training, averaged over episodes $B$. Dashed line target indicates that the constraint was set to be greater than target (not equal to it). \textit{CALSHEQ} (ours) and baseline (Bayesian optimization).}
  \label{f7}
\end{figure}

\subsection{Emergent agent behavior}
\label{secemergent}

In this section, we train shared policies $\pi^{LT}$ and $\pi^{LP}$, and investigate several emergent behaviors learnt by our agents. 

We consider a base agent configuration for which we vary the attributes of the first LP and study how its policy changes as a consequence. We consider 2 LP supertypes with risk aversion $\gamma = 0.5$, and respectively 1 and 2 agents; 2 LT supertypes consisting of respectively 12 flow LTs ($\omega=0$) and 2 PnL LTs ($\omega=1$) all connected to the ECN. The network of agents and ECN is set to be fully connected.

We train both shared policies $\pi^{LT}$ and $\pi^{LP}$ by randomizing the following attributes of the first LP: connectivity to flow and PnL LTs, PnL weight, risk aversion, and proceed to analyzing its learnt behavior as a function of these attributes. Precisely, we define the supertype of the first LP to sample uniformly the latter quantities over $[0,1]$ for connectivities and PnL weight, and $[0,2]$ for risk aversion. We display in figures \ref{figgn1}-\ref{figgn2} the rewards during training. Once the two policies have been trained, we compute metrics of interest from the point of view of the first LP by running a fixed number of episodes (100) per chosen set of agent types and aggregating the corresponding data. For example, if we look at the LP's inventory as a function of its risk aversion, we run 100 episodes per chosen value of risk aversion, and compute the average inventory over those episodes. 

At a given time $t$, let $\mathcal{O}^{ask}_t$, $\mathcal{O}^{bid}_t$ be the sets of trades performed by a given LP with the various LTs on both sides. For a given side $\alpha \in \{\text{bid},\text{ask}\}$, we define $\mathcal{F}_t(\alpha)$ the \textit{flow} of a LP as the sum of absolute trade quantities, and $\epsilon_{t,\alpha}$ its related price:
\begin{align}
\label{flowrprice}
    \mathcal{F}_t(\alpha):= \sum_{q : q \in \mathcal{O}^{\alpha}_t} |q|, \hspace{4mm} \epsilon_{t,\alpha} := \left\{
	\begin{array}{ll}
		\frac{1}{2} \epsilon_{t,\text{spread}} + \epsilon_{t,\text{skew}}  & \text{if } \alpha = \text{ask} \\
		\frac{1}{2} \epsilon_{t,\text{spread}} - \epsilon_{t,\text{skew}}  & \text{if } \alpha = \text{bid}
	\end{array}.
\right.
\end{align}

Note that with this convention, the lower $\epsilon_{t,\alpha}$, the more attractive the price from a LT standpoint. We further define the \textit{flow response curve} $\epsilon \to \mathcal{F}(\epsilon)$ the average flow obtained for a given price level $\epsilon$:
\begin{align}
\label{flowrprice2}
    \mathcal{F}(\epsilon):= \frac{\sum_{\alpha \in \{\text{bid},\text{ask}\}} \sum_{t=1}^T \mathcal{F}_t(\alpha) 1_{\{\epsilon_{t,\alpha} = \epsilon\}}}{\sum_{\alpha \in \{\text{bid},\text{ask}\}} \sum_{t=1}^T  1_{\{\epsilon_{t,\alpha} = \epsilon\}}}.
\end{align}

\textbf{Skewing intensity}. Recall that $2\epsilon_{t,skew} = \epsilon_{t,\text{ask}}-\epsilon_{t,\text{bid}}$ exactly captures the LP pricing asymmetry between bid and ask sides, and that \textit{skewing} refers to setting prices asymmetrically on both sides so as to reduce, or internalize, one's inventory. The more positive the inventory, the more we expect the LP to want to sell, i.e. the more negative $\epsilon_{t,skew}$. We define the \textit{skewing intensity} as the slope of the linear regression to the cloud of points $\epsilon_{t,\text{skew}} = f(q_t)$, where $q_t$ is the LP's net inventory at time $t$. In figure \ref{figg11} we look at this regression line for various connectivities to flow LTs. We see that the more connected, the more the LP learns to skew, materialized by the slope getting more negative, cf. figure \ref{figg1}. This is because flow LTs trade independently of PnL, hence provide to the LP a stream of liquidity available at all times: the LP learns that the more connected to these LTs, the more it can reduce its inventory by setting its prices asymmetrically via $\epsilon_{t,\text{skew}}$. Similarly, we see on figure \ref{figg1} that the more risk averse, the more intense the skewing. This is because the more risk averse, the more eager to reduce its inventory. We display in figure \ref{figira} the absolute inventory as a function of risk aversion, and indeed observe the decreasing pattern. The more the PnL weight, the less intense the skewing: this is because skewing intensely costs PnL. The more connected to PnL LTs, the more intense the skewing: this is because PnL LTs play a role similar to risk aversion, since from the point of view of the LP, both penalize mid-price variations via inventory PnL in section \ref{secrlagents}. The corresponding plots of  $\epsilon_{t,\text{skew}}$ vs. inventory are presented in figures \ref{figg12}, \ref{figg13}, \ref{figg14}, and the skew intensity distribution in figure \ref{figg9}.

\textbf{Hedge fraction}. We plot the LP hedge fraction $\epsilon_{\text{hedge}}$ in figure \ref{figg3}. The more risk averse, the more eager the LP is to liquidate its inventory, hence the more hedging. Similarly, hedging is an increasing function of connectivity to PnL LTs, since we discussed that they play a similar role as risk aversion. The more the PnL weight, the less the hedging: this is because hedging costs PnL. Finally, the more connected to flow LTs, the less the hedging, since the easier it is to reduce one's inventory, as discussed when analyzing skewing intensity.

\textbf{Pricing}. We plot in figure \ref{figg7rs} the average price $\epsilon_{t,\alpha}$ in (\ref{flowrprice}). Remember that the lower the latter, the more competitive the pricing. As expected the higher the PnL weight, the higher the impact of the pricing on the LP, hence the less competitive the pricing. The more the connectivity to PnL LTs, the more the LP's PnL suffers, hence the less competitive the pricing so as to compensate the related loss. The more connected to flow LTs, the more competitive the pricing: this is because the LP gets more revenue, so can afford to be more competitive on its prices.

\textbf{Inventory holding time}. We define the \textit{inventory holding time} $\tau$ of a given inventory $q_t$ the first time where $q_{t+\tau}$ has opposite sign than $q_t$, i.e. the first time where the inventory goes back to zero. We plot in figure \ref{figg4} the average inventory holding time, and in figure \ref{figg5} a more granular view of the holding time as a function of inventory. The higher the PnL weight, the higher the holding time, since reducing one's inventory is done via skewing or hedging, which both cost PnL. The more risk averse, the more the hedging and skewing, hence the lower the holding time. The more connected to flow LTs, the lower the holding time since the more intense the skewing. The more connected to PnL LTs, the less competitive the pricing, hence the higher the holding time.

\textbf{Flow}. We plot in figure \ref{figg7} the flow $\mathcal{F}(\epsilon)$ received by the LP from each LT agent class, where we normalize each curve by the number of agents of that class connected to the LP. This way, such curves represent the typical flow shape received by the LP from a representative flow or PnL LT. As the PnL weight increases, it is interesting to see that the flow from PnL LTs decreases more abruptly than that from flow LTs: this is because the price $\epsilon_{t,\alpha}$ gets significantly less competitive, to which the PnL LTs respond more intensely as their objective is PnL related. As flow LT connectivity increases, we saw in figure \ref{figg11} that the LP tailors its pricing to this class of LT as they are lucrative from a PnL point of view and help reduce inventory, hence it becomes gradually less focused on PnL LTs, and consequently gets less related flow. 

\textbf{Market share and PnL}. We present market share and PnL in figures \ref{figg6}, \ref{figg8}. As PnL weight increases, the LP's objective becomes more weighted towards market share and hence market share decreases and PnL increases. Connectivity to PnL LTs penalizes PnL as expected, and interestingly it also penalizes market share since we saw that the price $\epsilon_{t,\alpha}$ gets less competitive on average. PnL is also a a decreasing function of risk aversion.

\begin{figure}[ht]
    \centering
  \centerline{\includegraphics[scale=0.4]{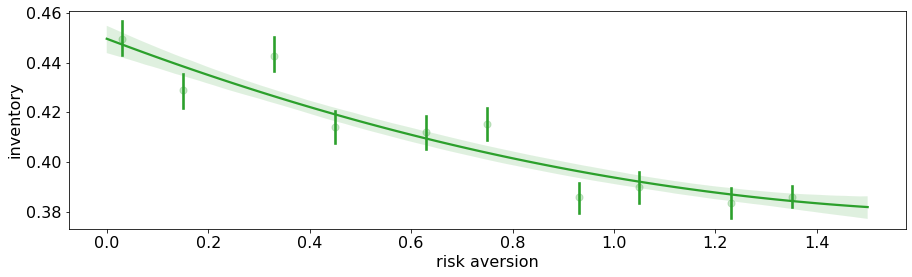}}
  \caption{Average absolute LP inventory $|q_t|$ as a function of risk aversion.}
    \label{figira}
\end{figure}

\begin{figure}[ht]
    \centering
    \begin{subfigure}[b]{\fsizeee \textwidth}
        \centering
        \includegraphics[width=\textwidth]{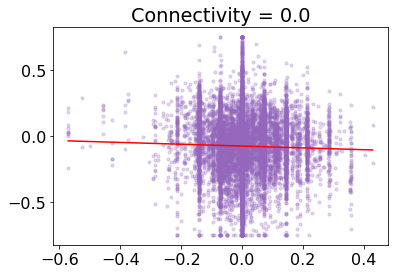}
   \end{subfigure}
    \hfill
    \begin{subfigure}[b]{\fsizeee \textwidth}
        \centering
        \includegraphics[width=\textwidth]{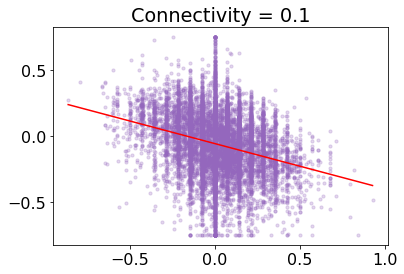}
    \end{subfigure}
    \begin{subfigure}[b]{\fsizeee \textwidth}
        \centering
        \includegraphics[width=\textwidth]{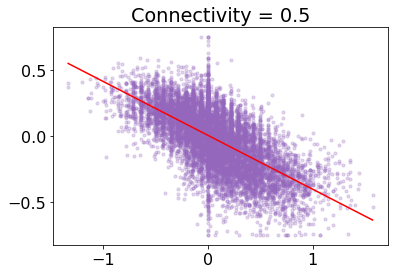}
    \end{subfigure}
    \begin{subfigure}[b]{\fsizeee \textwidth}
        \centering
        \includegraphics[width=\textwidth]{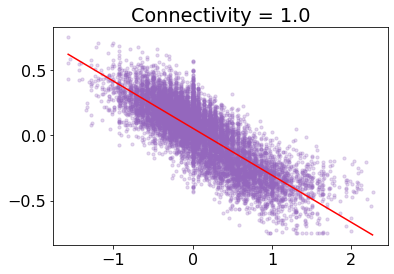}
    \end{subfigure}
    \caption{Skew $\epsilon_{t,\text{skew}}$ (y-axis) vs. inventory (x-axis) for various flow LT connectivity values. The more connected to flow LTs, the more intensely the LP skews, i.e. price asymmetrically as a function of its inventory. This is quantified by the slope of the regression line getting more negative (skewing intensity), cf. also figure \ref{figg1}.}
    \label{figg11}
\end{figure}

\clearpage

\begin{figure}[ht]
    \centering
    \begin{subfigure}[b]{\fsizee \textwidth}
        \centering
        \includegraphics[width=\textwidth]{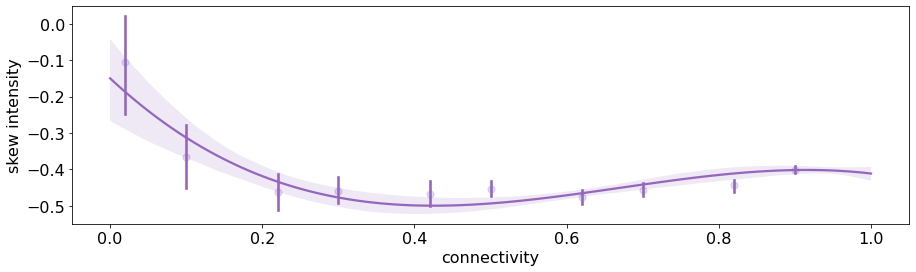}
   \end{subfigure}
    \hfill
    \begin{subfigure}[b]{\fsizee \textwidth}
        \centering
        \includegraphics[width=\textwidth]{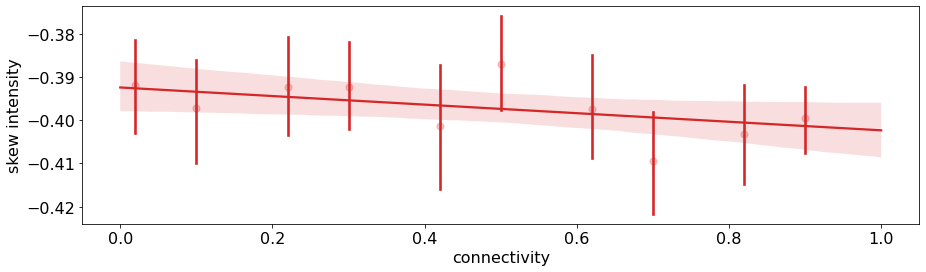}
    \end{subfigure}
    \begin{subfigure}[b]{\fsizee \textwidth}
        \centering
        \includegraphics[width=\textwidth]{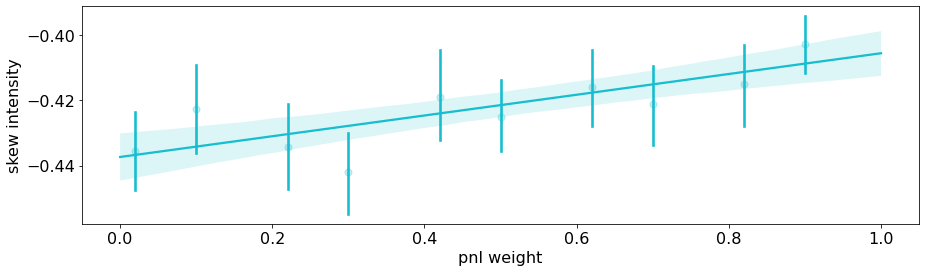}
    \end{subfigure}
    \begin{subfigure}[b]{\fsizee \textwidth}
        \centering
        \includegraphics[width=\textwidth]{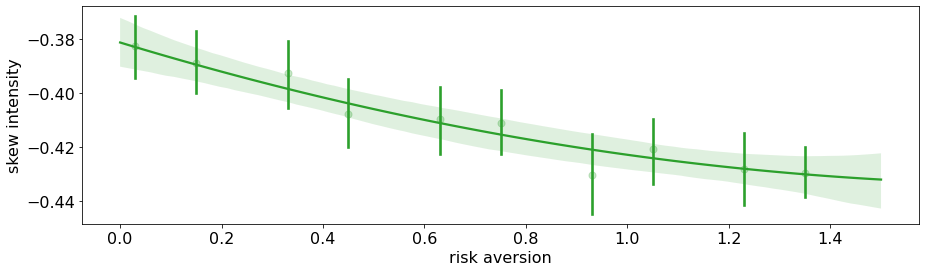}
    \end{subfigure}
    \caption{Skew intensity (y-axis) as a function of flow LT connectivity, PnL LT connectivity, PnL weight, risk aversion (x-axis, from top to bottom). The more intense the skewing, the more negative the skewing intensity (the lower).}
    \label{figg1}
\end{figure}

\begin{figure}[ht]
    \centering
    \begin{subfigure}[b]{\fsizee \textwidth}
        \centering
        \includegraphics[width=\textwidth]{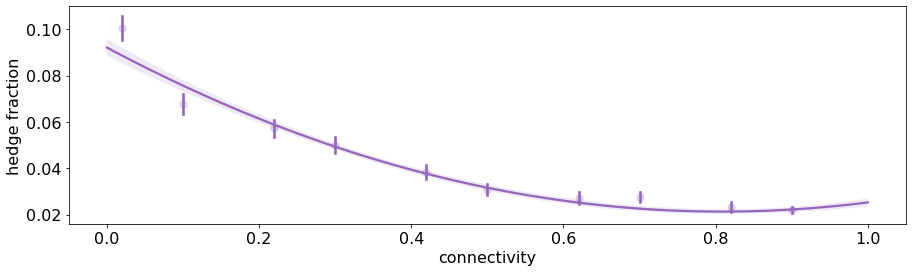}
    \end{subfigure}
    \hfill
    \begin{subfigure}[b]{\fsizee \textwidth}
        \centering
        \includegraphics[width=\textwidth]{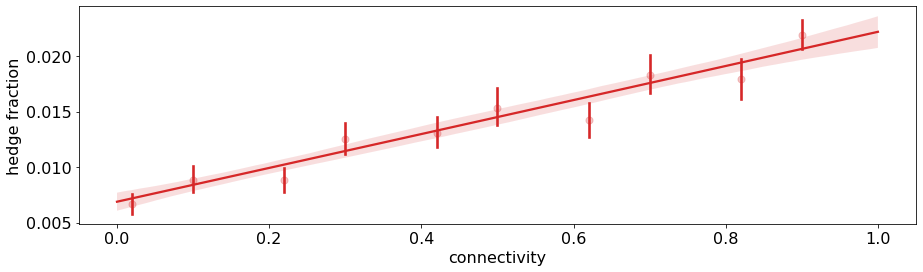}
    \end{subfigure}
    \begin{subfigure}[b]{\fsizee \textwidth}
        \centering
        \includegraphics[width=\textwidth]{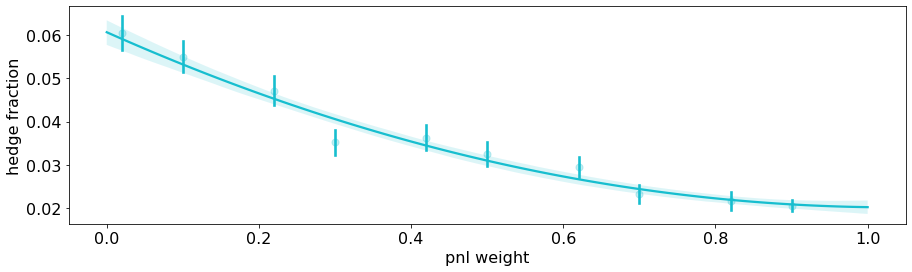}
    \end{subfigure}
    \begin{subfigure}[b]{\fsizee \textwidth}
        \centering
        \includegraphics[width=\textwidth]{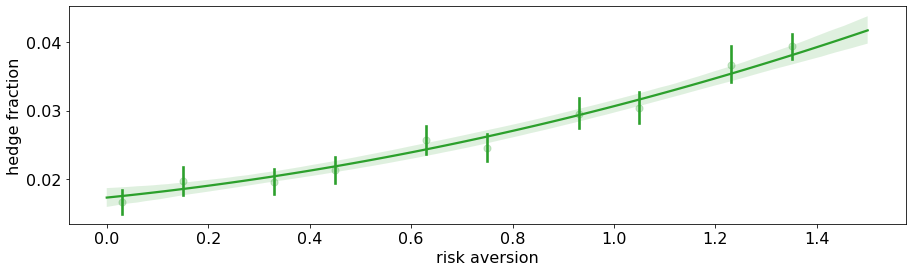}
    \end{subfigure}
    \caption{Hedge fraction (y-axis) as a function of flow LT connectivity, PnL LT connectivity, PnL weight, risk aversion (x-axis, from top to bottom).}
    \label{figg3}
\end{figure}

\begin{figure}[ht]
    \centering
    \begin{subfigure}[b]{\fsizee \textwidth}
        \centering
        \includegraphics[width=\textwidth]{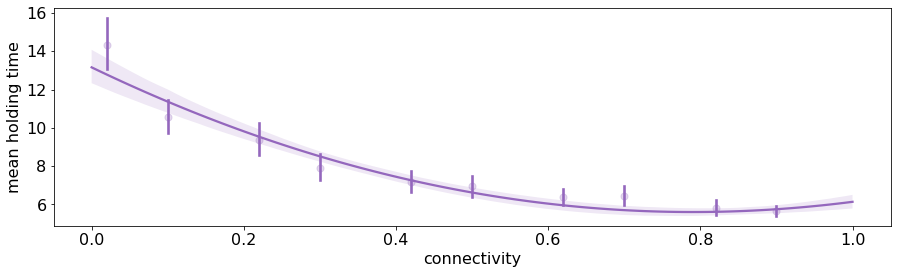}
    \end{subfigure}
    \hfill
    \begin{subfigure}[b]{\fsizee \textwidth}
        \centering
        \includegraphics[width=\textwidth]{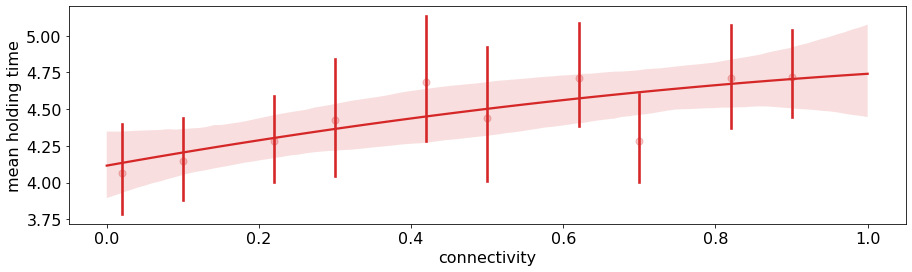}
    \end{subfigure}
    \begin{subfigure}[b]{\fsizee \textwidth}
        \centering
        \includegraphics[width=\textwidth]{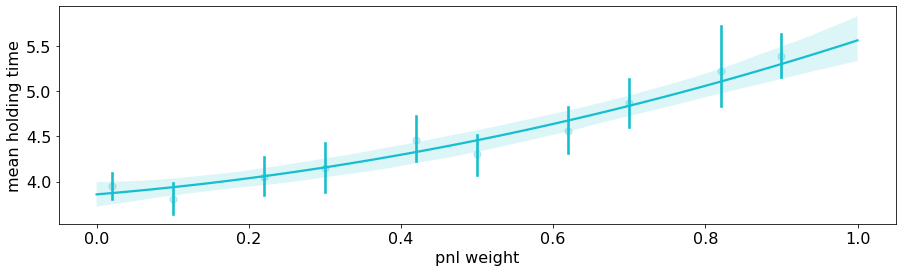}
    \end{subfigure}
    \begin{subfigure}[b]{\fsizee \textwidth}
        \centering
        \includegraphics[width=\textwidth]{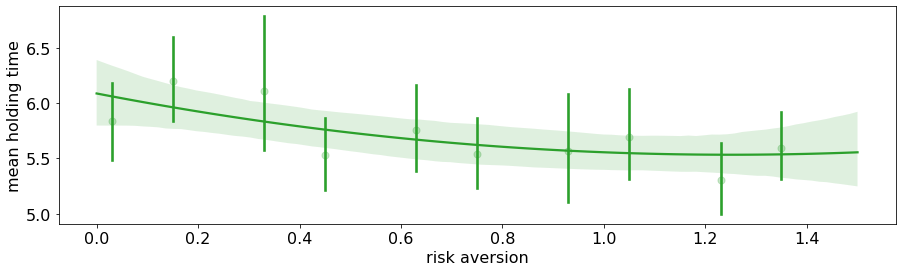}
    \end{subfigure}
    \caption{Mean inventory holding time (y-axis) as a function of flow LT connectivity, PnL LT connectivity, PnL weight, risk aversion (x-axis, from top to bottom).}
    \label{figg4}
\end{figure}

\clearpage

\begin{figure}[ht]
    \centering
    \begin{subfigure}[b]{\fsizee \textwidth}
        \centering
        \includegraphics[width=\textwidth]{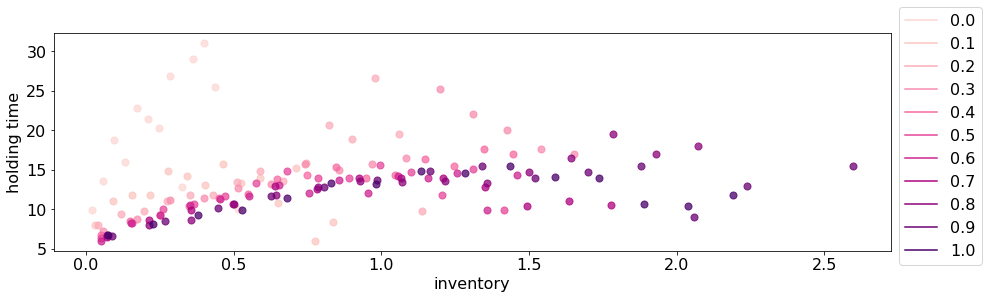}
    \end{subfigure}
    \hfill
    \begin{subfigure}[b]{\fsizee \textwidth}
        \centering
        \includegraphics[width=\textwidth]{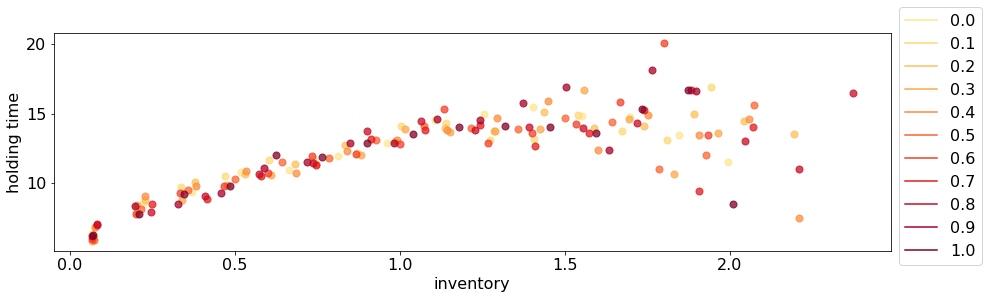}
    \end{subfigure}
    \begin{subfigure}[b]{\fsizee \textwidth}
        \centering
        \includegraphics[width=\textwidth]{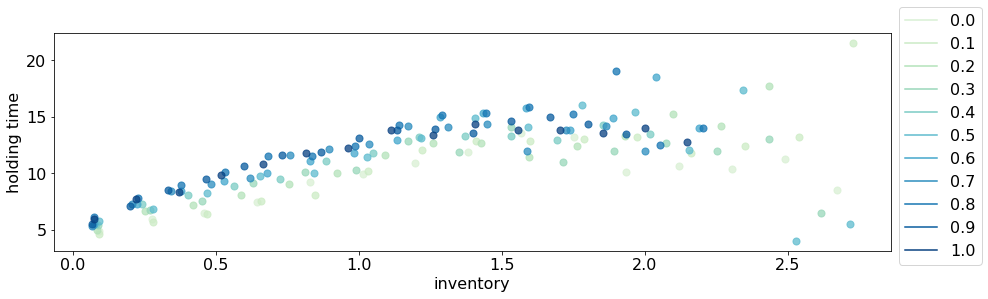}
    \end{subfigure}
    \begin{subfigure}[b]{\fsizee \textwidth}
        \centering
        \includegraphics[width=\textwidth]{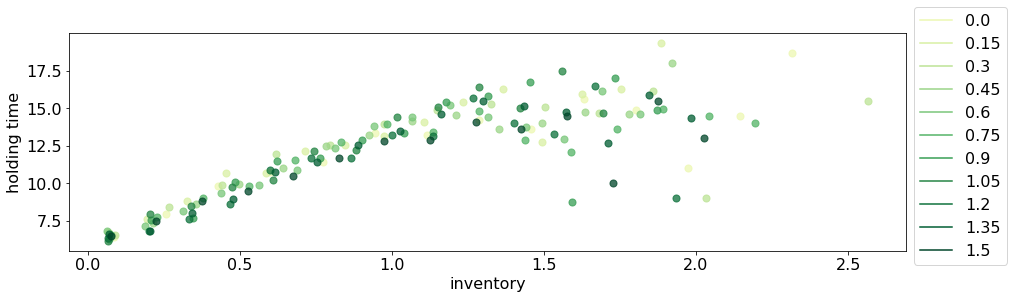}
    \end{subfigure}
    \caption{Inventory holding time (y-axis) as a function of inventory (x-axis), for various values of flow LT connectivity, PnL LT connectivity, PnL weight, risk aversion (from top to bottom).}
    \label{figg5}
\end{figure}

\begin{figure}[ht]
    \centering
    \begin{subfigure}[b]{\fsizee \textwidth}
        \centering
        \includegraphics[width=\textwidth]{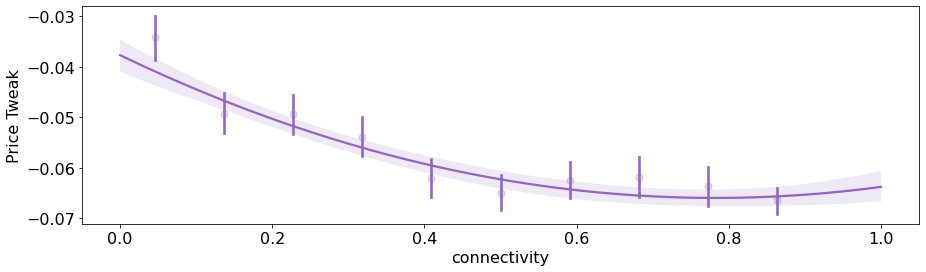}
    \end{subfigure}
    \hfill
    \begin{subfigure}[b]{\fsizee \textwidth}
        \centering
        \includegraphics[width=\textwidth]{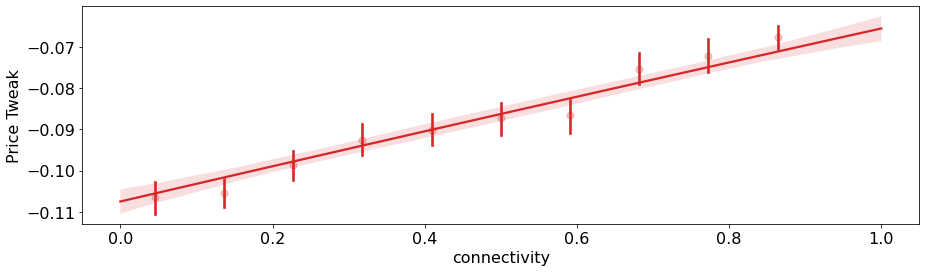}
    \end{subfigure}
    \begin{subfigure}[b]{\fsizee \textwidth}
        \centering
        \includegraphics[width=\textwidth]{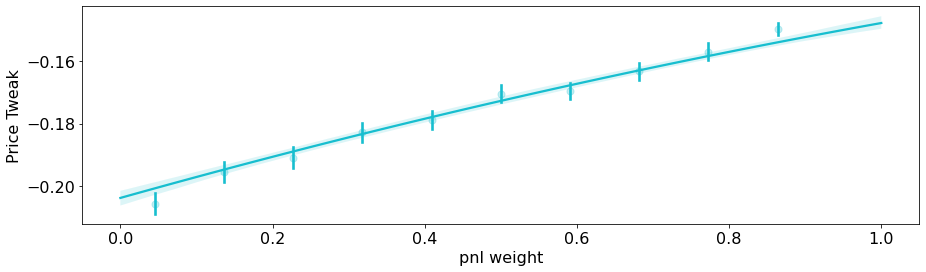}
    \end{subfigure}
    \begin{subfigure}[b]{\fsizee \textwidth}
        \centering
        \includegraphics[width=\textwidth]{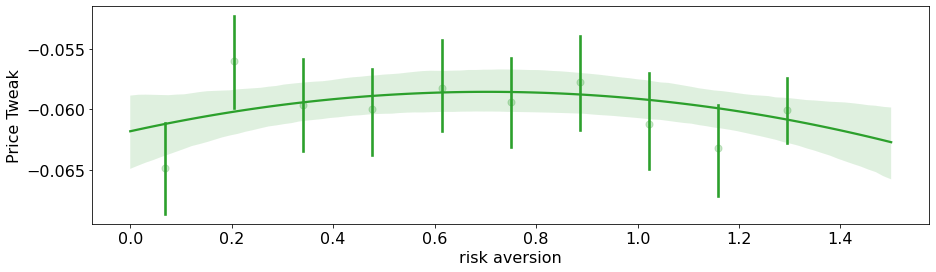}
    \end{subfigure}
    \caption{Average price $\epsilon_{t,\alpha}$ in (\ref{flowrprice}) (y-axis) as a function of flow LT connectivity, PnL LT connectivity, PnL weight, risk aversion (x-axis, from top to bottom).}
    \label{figg7rs}
\end{figure}

\begin{figure}[ht]
    \centering
    \begin{subfigure}[b]{\fsizee \textwidth}
        \centering
        \includegraphics[width=\textwidth]{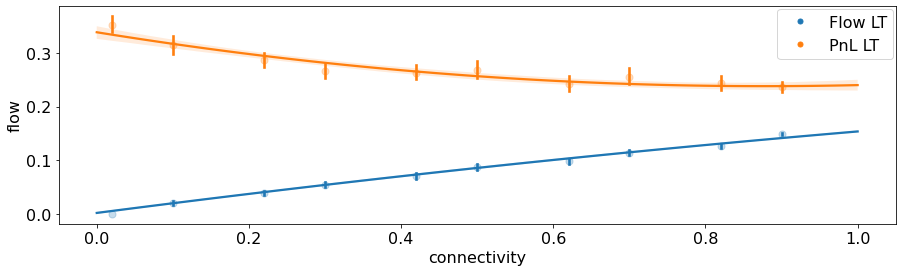}
    \end{subfigure}
    \hfill
    \begin{subfigure}[b]{\fsizee \textwidth}
        \centering
        \includegraphics[width=\textwidth]{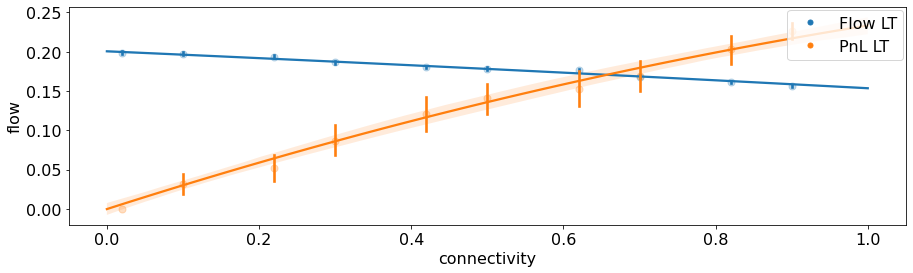}
    \end{subfigure}
    \begin{subfigure}[b]{\fsizee \textwidth}
        \centering
        \includegraphics[width=\textwidth]{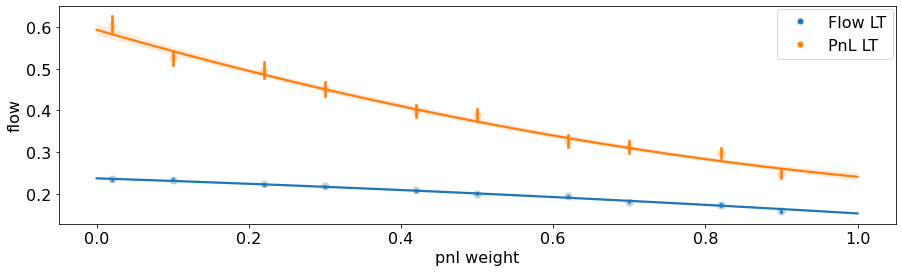}
    \end{subfigure}
    \begin{subfigure}[b]{\fsizee \textwidth}
        \centering
        \includegraphics[width=\textwidth]{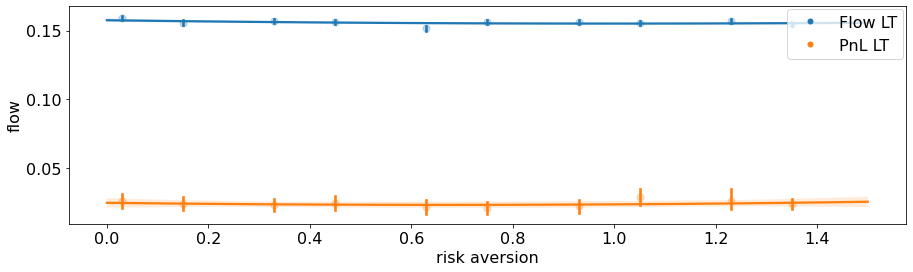}
    \end{subfigure}
    \caption{Flow received by the LP from PnL LTs (orange) and flow LTs (blue) (y-axis) as a function of flow LT connectivity, PnL LT connectivity, PnL weight, risk aversion (x-axis, from top to bottom). Flow curves are normalized by the number of LT agents in each agent class (flow or PnL) that the LP is connected to, so that all curves are comparable and can be interpreted as the average flow received from a typical agent of that class.}
    \label{figg7}
\end{figure}
\clearpage

\begin{figure}[ht]
    \centering
    \begin{subfigure}[b]{\fsizee \textwidth}
        \centering
        \includegraphics[width=\textwidth]{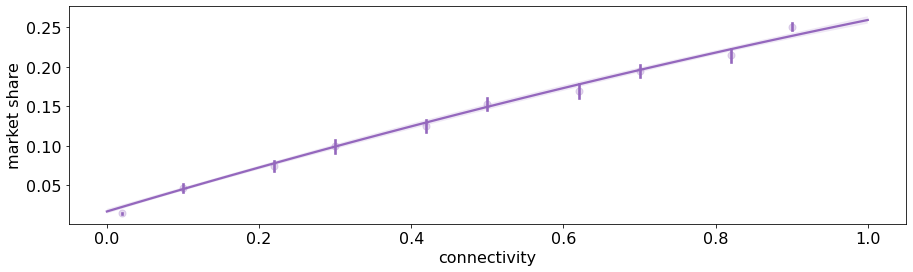}
    \end{subfigure}
    \hfill
    \begin{subfigure}[b]{\fsizee \textwidth}
        \centering
        \includegraphics[width=\textwidth]{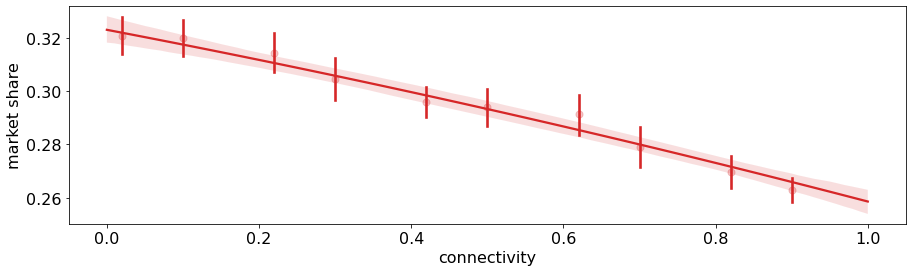}
    \end{subfigure}
    \begin{subfigure}[b]{\fsizee \textwidth}
        \centering
        \includegraphics[width=\textwidth]{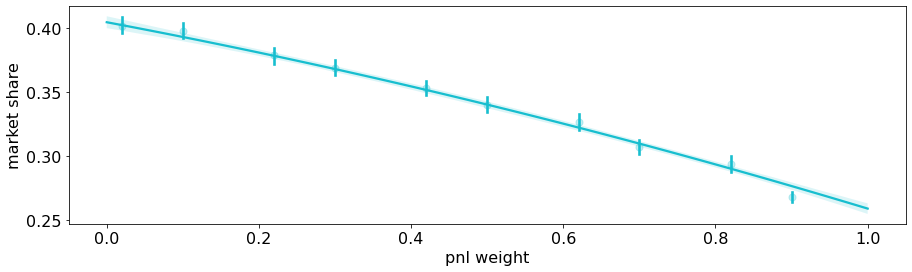}
    \end{subfigure}
    \begin{subfigure}[b]{\fsizee \textwidth}
        \centering
        \includegraphics[width=\textwidth]{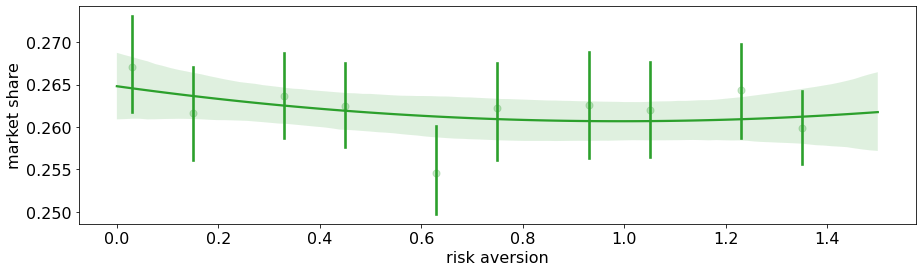}
    \end{subfigure}
    \caption{Market Share (y-axis) as a function of flow LT connectivity, PnL LT connectivity, PnL weight, risk aversion (x-axis, from top to bottom).}
    \label{figg6}
\end{figure}

\begin{figure}[ht]
    \centering
    \begin{subfigure}[b]{\fsizee \textwidth}
        \centering
        \includegraphics[width=\textwidth]{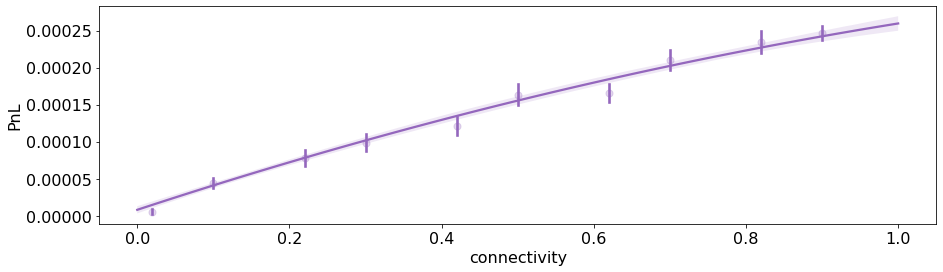}
    \end{subfigure}
    \hfill
    \begin{subfigure}[b]{\fsizee \textwidth}
        \centering
        \includegraphics[width=\textwidth]{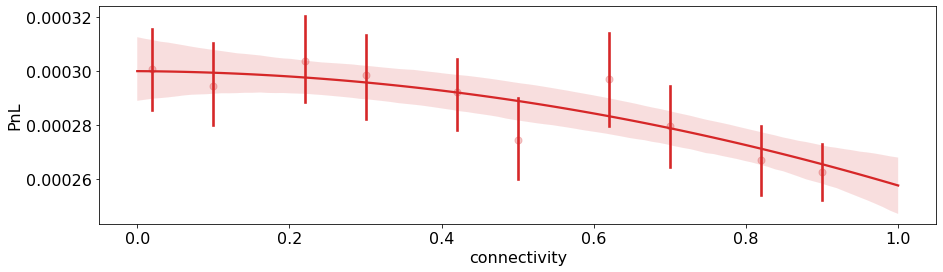}
    \end{subfigure}
    \begin{subfigure}[b]{\fsizee \textwidth}
        \centering
        \includegraphics[width=\textwidth]{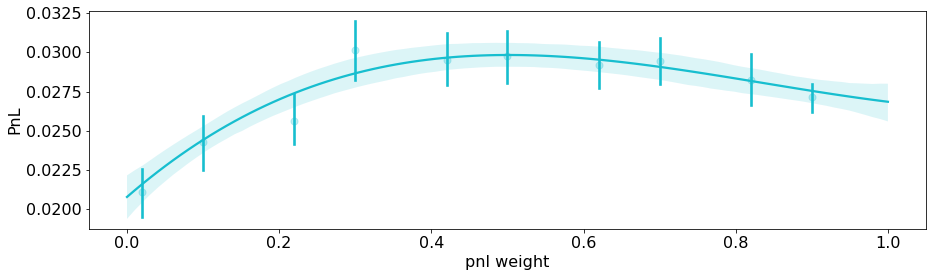}
    \end{subfigure}
    \begin{subfigure}[b]{\fsizee \textwidth}
        \centering
        \includegraphics[width=\textwidth]{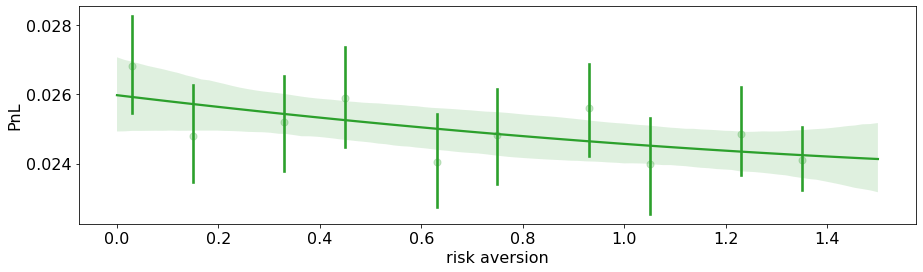}
    \end{subfigure}
    \caption{PnL (y-axis) as a function of flow LT connectivity, PnL LT connectivity, PnL weight, risk aversion (x-axis, from top to bottom).}
    \label{figg8}
\end{figure}

\begin{figure}[ht]
    \centering
    \begin{subfigure}[b]{0.7 \textwidth}
        \centering
        \includegraphics[width=\textwidth]{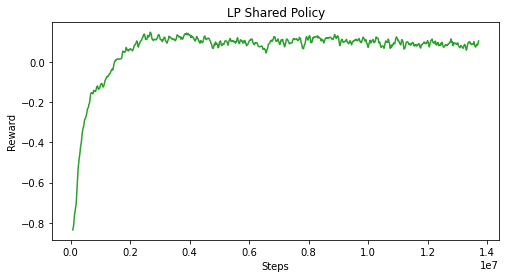}
    \end{subfigure}
    \hfill
    \begin{subfigure}[b]{0.7 \textwidth}
        \centering
        \includegraphics[width=\textwidth]{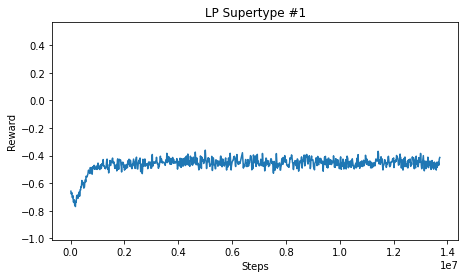}
    \end{subfigure}
    \begin{subfigure}[b]{0.7 \textwidth}
        \centering
        \includegraphics[width=\textwidth]{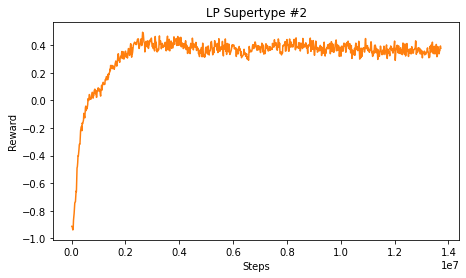}
    \end{subfigure}
    \caption{LP cumulative reward during training: shared policy, and supertypes 1 and 2.}
    \label{figgn1}
\end{figure}

\begin{figure}[ht]
    \centering
    \begin{subfigure}[b]{0.7 \textwidth}
        \centering
        \includegraphics[width=\textwidth]{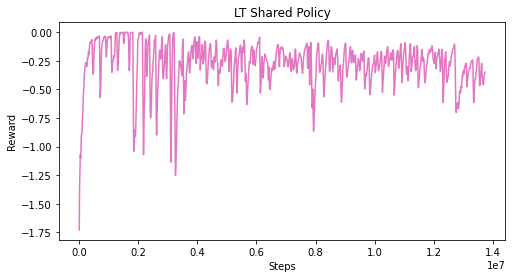}
    \end{subfigure}
    \hfill
    \begin{subfigure}[b]{0.7 \textwidth}
        \centering
        \includegraphics[width=\textwidth]{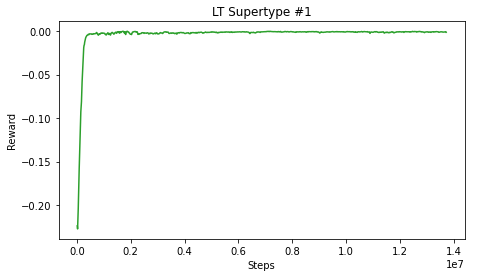}
    \end{subfigure}
    \begin{subfigure}[b]{0.7\textwidth}
        \centering
        \includegraphics[width=\textwidth]{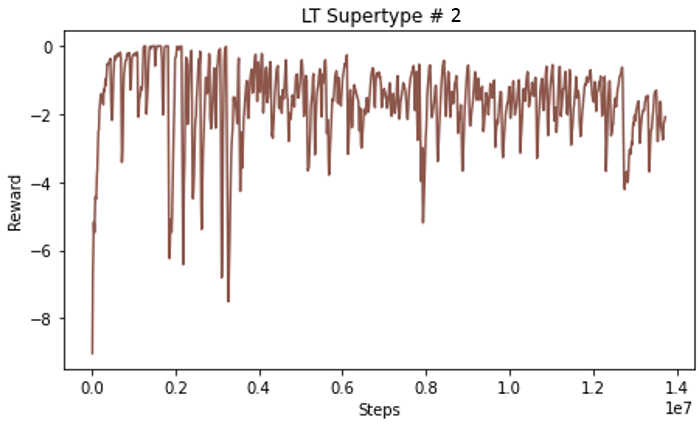}
    \end{subfigure}
    \caption{LT cumulative reward during training: shared policy, and supertypes 1 and 2.}
    \label{figgn2}
\end{figure}

\clearpage
\section{Conclusion}
In this work, we formalized the interactions between liquidity providers and liquidity takers in a dealer market as a multi-stage, multi-type stochastic game. We showed how a suitable design of parametrized families of reward functions coupled with shared policy learning makes our reinforcement-learning-driven agents learn emergent behaviors relative to a wide spectrum of objectives encompassing profit-and-loss, optimal execution and market share, by playing against each other. In particular, we find that liquidity providers naturally learn to balance hedging and skewing as a function of their objectives, where skewing refers to setting their pricing asymmetrically on the bid and ask sides as a function of their inventory. We introduced a novel RL-based calibration algorithm which we found performed well at imposing constraints on the game equilibrium. On the theoretical side, we showed convergence rates for our multi-agent reinforcement learning policy gradient algorithm under a transitivity assumption, closely related to generalized ordinal potential games.

\section*{Acknowledgements}
We would like to thank Arnaud Floesser, Cathy Lin, Chi Nzelu and their FX ATS (Automated Trading Strategies) team from J.P. Morgan Corporate and Investment Bank, for their support and inputs that provided the impetus and motivation for this work.


\section*{Disclaimer}
This paper was prepared for information purposes by the Artificial Intelligence Research group of JPMorgan Chase \& Co and its affiliates (“JP Morgan”), and is not a product of the Research Department of JP Morgan. JP Morgan makes no representation and warranty whatsoever and disclaims all liability, for the completeness, accuracy or reliability of the information contained herein. This document is not intended as investment research or investment advice, or a recommendation, offer or solicitation for the purchase or sale of any security, financial instrument, financial product or service, or to be used in any way for evaluating the merits of participating in any transaction, and shall not constitute a solicitation under any jurisdiction or to any person, if such solicitation under such jurisdiction or to such person would be unlawful. 

\bibliography{finalmafi}
\bibliographystyle{apalike}

\setcounter{section}{0}
\appendix
\renewcommand{\thesection}{\Alph{section}}

\clearpage

\section{Appendix}

\textbf{Proof of theorem \ref{dam}.} The strategy consists in applying convergence results in \cite{EK} (theorem 4.2, ch. 7). We need to show that for every $r>0$, $x>0$:
\begin{align}
\label{ek1}
\lim_{\epsilon \to 0} \sup_{||V|| \leq r} \epsilon^{-1} \PP[||V^\epsilon_{n+1}-V^\epsilon_n|| \geq x | V^\epsilon_n=V] = 0
\end{align}
\begin{align}
\label{ek2}
\lim_{\epsilon \to 0} \sup_{||V|| \leq r}  ||\epsilon^{-1}\EE[V^\epsilon_{n+1}-V^\epsilon_n| V^\epsilon_n=V] - b(V)||= 0
\end{align}
\begin{align}
\label{ek3}
\lim_{\epsilon \to 0} \sup_{||V|| \leq r}  ||\epsilon^{-1}\EE[(V^\epsilon_{n+1}-V^\epsilon_n)(V^\epsilon_{n+1}-V^\epsilon_n)^T| V^\epsilon_n=V] - Q(V)||= 0
\end{align}

where $V^\epsilon_n:=(V^\epsilon_{i,n})_{i \in [1,2m]}$, $\odot$ is the elementwise product, and $b(V):=-\mu^- \odot V + \mu^+$. (\ref{ek2}), (\ref{ek3}) both follow by direct computation. Indeed, we have $V^\epsilon_{n+1}-V^\epsilon_n = -\delta^-_{n,\epsilon} \odot V^\epsilon_n + \delta^+_{n,\epsilon}$, and therefore $\epsilon^{-1}\EE[V^\epsilon_{n+1}-V^\epsilon_n| V^\epsilon_n=V] = b(V)$. This proves (\ref{ek2}). For (\ref{ek3}), we proceed again by direct computation, denoting $\widetilde{\delta}^\pm_{i,n}:=\delta^\pm_{i,n}-\mu^\pm_{i}$:
\begin{align*}
    \begin{split}
    \epsilon^{-1}\EE[(V^\epsilon_{i,n+1}-V^\epsilon_{i,n})(V^\epsilon_{j,n+1}-&V^\epsilon_{j,n})| V^\epsilon_n=V] = \\
    & \EE[\widetilde{\delta}^+_{i,n} \widetilde{\delta}^+_{j,n}] - V_i  \EE[\widetilde{\delta}^-_{i,n} \widetilde{\delta}^+_{j,n}] - V_j  \EE[\widetilde{\delta}^-_{j,n} \widetilde{\delta}^+_{i,n}] + V_i V_j \EE[\widetilde{\delta}^-_{i,n} \widetilde{\delta}^-_{j,n}]
    \end{split}
\end{align*}
Expanding the latter yields $\epsilon^{-1}\EE[(V^\epsilon_{i,n+1}-V^\epsilon_{i,n})(V^\epsilon_{j,n+1}-V^\epsilon_{j,n})| V^\epsilon_n=V]=[Q(V)]_{ij}$, and hence (\ref{ek3}). We complete the proof by showing (\ref{ek1}). Conditionally on $V^\epsilon_n=V$, we have:
\begin{align*}
||V^\epsilon_{n+1}-V^\epsilon_n|| &= ||\epsilon(\mu^+ -\mu^- \odot V) + \sqrt{\epsilon}(\widetilde{\delta}^+_{n} -\widetilde{\delta}^-_{n} \odot V)||\\
& \leq \epsilon c_V + \sqrt{\epsilon}(||\widetilde{\delta}^+_{n}||+2||V||)
\end{align*}
where $c_V:=||\mu^+|| + ||\mu^-|| ||V||$ is a constant only dependent on $V$, and remembering that $\delta^-$ takes value in $[0,1]$ almost surely. We now apply Chebyshev's inequality, since by assumption $\delta^{+}_{i,n}$ are bounded in $L^{2+\eta}$:
\begin{align*}
    & \PP[||V^\epsilon_{n+1}-V^\epsilon_n|| \geq x | V^\epsilon_n=V] \\
    \leq & \PP[\epsilon c_V + \sqrt{\epsilon}(||\widetilde{\delta}^+_{n}||+2||V||) \geq x]\\
    \leq & \frac{\epsilon^{1+\frac{\eta}{2}}}{(x-\epsilon c_V )^{2+\eta}}\EE[(||\widetilde{\delta}^+_{n}||+2||V||)^{2+\eta}]\\
    \leq & \frac{\epsilon^{1+\frac{\eta}{2}}}{(x-\epsilon c_V )^{2+\eta}}(m_\eta+2^{1+\eta}(2||V||)^{2+\eta})
\end{align*}
where the last inequality comes from Holder's inequality and $m_\eta:=\EE[||\widetilde{\delta}^+_{n}||^{2+\eta}]$. We thus have for $||V|| \leq r$, $c_r:= ||\mu^+|| + ||\mu^-|| r$ and $\epsilon \in (0,\frac{x}{2 c_r})$:
$$
\epsilon^{-1} \PP[||V^\epsilon_{n+1}-V^\epsilon_n|| \geq x | V^\epsilon_n=V] \leq \epsilon^{\frac{\eta}{2}} \left(\frac{2}{x}\right)^{2+\eta} (m_\eta + 8^{2\eta} r^{2+\eta})
$$
which shows (\ref{ek1}). \qed

\vspace{3mm}

\textbf{Proof of proposition \ref{corvar}}. Denoting $u_n:=\EE[V_{i,n}]$, we have from the recursion equation $u_{n+1}=(1-\mu_i^-)u_{n}+\mu_i^+$, from which $u_{n}=\mu_i^\infty+(V_{i,0}-\mu_i^\infty)(1-\mu_i^-)^n$. The limit of $u_n$ follows. Denoting $u_t:=\EE[V^*_{i,t}]$, taking expectations we have the ODE $u_t'=\mu_i^-(\mu_i^\infty-u_t)$, from which the limit of $u_t$ follows, similar to the Ornstein-Uhlenbeck case. The covariance follows the same way. Indeed, let $z_n:=\EE[V_{i,n}V_{j,n}]-\mu_i^\infty\mu_j^\infty$, then we have, assuming without loss of generality that $V_{i,0}=\mu^\infty_{i}$:
$$
z_{n+1} = z_n \left( 1-\mu_i^- -\mu_j^- + \mu_i^-\mu_j^-+\rho_{ij}^-\sigma_i^-\sigma_j^-\right)+\sigma_{ij}^{\infty }
$$
which is of the form $z_{n+1} = (1-a) z_n + b$. The limit of $z_n$ follows directly as $\frac{b}{a}$. For the continuous-time limit $z_t:=\EE[V_{i,t}^*V_{j,t}^*]-\mu_i^\infty\mu_j^\infty$, simply observe that the rescaling by $\epsilon$ replaces $n$ by $\lfloor \epsilon^{-1}t \rfloor$, $\mu_i^\pm$ by $\epsilon \mu_i^\pm$ and $\sigma_i^\pm$ by $\sqrt{\epsilon} \sigma_i^\pm$. Hence the discounting term $\left( 1-\mu_i^- -\mu_j^- + \mu_i^-\mu_j^-+\rho_{ij}^-\sigma_i^-\sigma_j^-\right)^n$ becomes $\left( 1-\epsilon \gamma + \circ(\epsilon)\right)^{\lfloor \epsilon^{-1}t \rfloor}$, with $\gamma:=\mu_i^- +\mu_j^- -\rho_{ij}^-\sigma_i^-\sigma_j^-$. The Riemann sum $\sum_{k=0}^{\lfloor \epsilon^{-1}t \rfloor}(1-\gamma \epsilon+ \circ(\epsilon))^{k}$ then converges to $\int_0^t e^{-\gamma x}dx$, which goes to $\gamma^{-1}$ as $t \to \infty$ and yields the desired result.

\vspace{3mm}

\textbf{Proof of proposition \ref{damr}}. We have:
$$
\mathcal{V}(x) := Q(x)-Q(0) = \sigma^{-2} x^2
 +  2\sigma^- \left( \mu^\infty \sigma^- - \rho\sigma^+ \right) x
$$
It is clear that if $\sigma^-=0$, $\mathcal{V}(x)=0$ $\forall x$. The roots of $\mathcal{V}$ are 0 and $2 \left(\rho \frac{\sigma^+}{\sigma^-} - \mu^\infty\right)$. Since $x=V^*-\mu^\infty$, the roots are $V^*=\mu^\infty$ and  $V^*= 2 \rho \frac{\sigma^+}{\sigma^-} - \mu^\infty$. We denote $\gamma_* \leq \gamma^*$ these 2 roots. The leading order coefficient of $\mathcal{V}$ is positive, therefore $\mathcal{V}$ is negative in the region $V^* \in (\gamma_*,\gamma^*)$, and nonnegative otherwise.

$V^*$ is never self-inhibiting if and only if $\gamma_*=\gamma^*$, i.e. $\mu^\infty=2 \rho \frac{\sigma^+}{\sigma^-} - \mu^\infty$, i.e. $\sigma^+ \mu^-\rho = \sigma^-\mu^+$.

We note that $\gamma_*$ and $\gamma^*$ depend on the model parameters only through the ratios $\frac{\sigma^+}{\sigma^-} \rho$ and $\frac{\mu^+}{\mu^-}$. Finally, we have $\partial_\rho \mathcal{V}(x) = -2\sigma^-\sigma^+ x$, which is positive if and only if $x$ is negative.

\vspace{3mm}

\begin{lemma}
\label{f2sis}
Every extended transitive game with payoff $f$ has at least one $(f,\epsilon)$-self-play sequence for every $\epsilon>0$, and every such sequence is finite.
\end{lemma}
\begin{proof}
First note that such a game has at least one $(f,\epsilon)$-self-play sequence for every $\epsilon>0$ since every $(x,x)$ is a $(f,\epsilon)$-self-play sequence of size 0 (cf. definition 2). Then, let $(x_n,y_n)$ be a $(f,\epsilon)$-self-play sequence. By definition of the self-play sequence we have $f(x_{2n+1},x_{2n})>f(x_{2n},x_{2n})+\epsilon$. By extended transitivity (cf. assumption 1) this implies $\Phi(x_{2n+1})>\Phi(x_{2n}) + \delta_\epsilon$. But $x_{2n+1}=x_{2n+2}$ by definition of the self-play sequence, hence $\Phi(x_{2n+2})>\Phi(x_{2n}) + \delta_\epsilon$. By induction $\Phi(x_{2n})>\Phi(x_{0})+n\delta_\epsilon$ for $n \geq 1$. If the sequence is not finite and since $\delta_\epsilon>0$, one can take the limit as $n \to \infty$ and get a contradiction, since $\Phi$ is bounded by extended transitivity assumption.
\end{proof}

\begin{theorem}
\label{n1}
An extended transitive game with payoff $f$ has a symmetric pure strategy $\epsilon-$Nash equilibrium for every $\epsilon>0$, which further can be reached within a finite number of steps following a $(f,\epsilon)$-self-play sequence.
\end{theorem}
\begin{proof}
Let $\epsilon>0$. Take a $(f,\epsilon)$-self-play sequence. By lemma \ref{f2sis}, such a sequence exists and is finite, hence one may take a $(f,\epsilon)$-self-play sequence of maximal size, say $2N_\epsilon$. Assume that its end point $(x,x)$ is not an $\epsilon-$Nash equilibrium. Then $\exists y$: $f(y,x)>f(x,x)+\epsilon$, which means that one can extend the $(f,\epsilon)$-self-play sequence to size $2N_\epsilon+2$ with entries $(y,x)$ and $(y,y)$, which violates the fact that such a sequence was taken of maximal size.
\end{proof}

\begin{theorem}
\label{n2}
An extended transitive game with continuous payoff $f$ and compact strategy set has a symmetric pure strategy Nash equilibrium.
\end{theorem}
\begin{proof}
By theorem \ref{n1}, take a sequence of $\epsilon_n$-Nash equilibria with $\epsilon_n \to 0$ and corresponding $(f,\epsilon_n)$-self-play sequence endpoints $(x_n,x_n)$. By compactness assumption, this sequence has a converging subsequence $(x_{m_n},x_{m_n})$, whose limit point $(x_*,x_*)$ belongs to the strategy set. We have by definition of $\epsilon_{m_n}$-Nash equilibrium that $f(x_{m_n},x_{m_n}) \geq \sup_y f(y,x_{m_n})-\epsilon_{m_n}$. Taking the limit as $n \to \infty$ and using continuity of $f$, we get $
f(x_{*},x_{*}) \geq \sup_y f(y,x_{*})$, which shows that $(x_*,x_*)$ is a symmetric pure strategy Nash equilibrium. 
\end{proof}

\textbf{Proof of theorem \ref{n11}.} The first part of the theorem follows from theorem \ref{n1}. Then, we have by assumption that $\Ss^{LP}$, $\As^{LP}$ and $\Ss^{\lambda^{LP}}$ are finite. Denote $m:=|\Ss^{LP}| \cdot |\As^{LP}| \cdot |\Ss^{\lambda^{LP}}|$. In that case we have:
$$
\mathcal{X}^{LP}=\{ (x^{s,\lambda}_{a} )\in [0,1]^{m}: \forall s \in [1,|\Ss^{LP}|], \lambda \in [1,|\Ss^{\lambda^{LP}}|], \hspace{2mm} \sum_{a=1}^{|\As^{LP}|} x^{s,\lambda}_a =1\}
$$
$\mathcal{X}^{LP}$ is a closed and bounded subset of $[0,1]^{m}$, hence by Heine–Borel theorem it is compact. Note that closedness comes from the fact that summation to 1 is preserved by passing to the limit. By assumption, the rewards are bounded, so by lemma 1, $V_{\Lambda_i}$ is continuous for all $i$, which yields continuity of $\widehat{V}$, hence we can apply theorem \ref{n2} to conclude.

\textbf{Proof of lemma \ref{lvf2}.}
Since by assumption $\Ss^{LP}$, $\As^{LP}$ and $\Ss^{\lambda^{LP}}$ are finite, we will use interchangeably sum and integral over these spaces. Let us denote the total variation metric for probability measures $\pi_1,\pi_2$ on $\mathcal{X}^{LP}$:
$$
\rho_{TV}(\pi_1,\pi_2):=\frac{1}{2} \max_{s,\lambda}\sum_{a \in \mathcal{A}^{LP}} |\pi_1(a|s,\lambda)-\pi_2(a|s,\lambda)|$$ 
and let us equip the product space $\mathcal{X}^{LP} \times \mathcal{X}^{LP}$ with the metric:
$$
\rho_{TV}((\pi_1,\pi_2),(\pi_3,\pi_4)):=\rho_{TV}(\pi_1,\pi_3)+\rho_{TV}(\pi_2,\pi_4).
$$
We work for simplicity with the infinite horizon case $T=+\infty$ and $\zeta<1$, but the finite horizon-case is dealt with the same way. In the calculations below, we omit the policy $\pi^{LT}$ of the LT agents since it is fixed and hence can be subsumed in the transition dynamics $\mathcal{T}$, similarly we omit the LP superscript, since it is clear we are dealing with agents of such class and there is no ambiguity. Remember that $z^{(i)}_t:=(s^{(i)}_t,a^{(i)}_t,\lambda_i)$. Let:
\begin{align*}
V_{\Lambda_i}(\pi_1,\pi_2,\bm{s},\bm{\lambda}):=\EE_{\substack{a^{(i)}_t \sim \pi_1(\cdot|\cdot,\lambda_i), \hspace{1mm} a^{(j)}_t \sim \pi_2(\cdot|\cdot,\lambda_j)}} \left[ \sum_{t=0}^\infty \zeta^t \RR(z^{(i)}_t,\bm{z^{(-i)}_t})|\bm{s_0}=\bm{s}\right], \hspace{2mm} j \neq i
\end{align*}
so that:
\begin{align*}
\widetilde{V}_{\Lambda_i}(\pi_1,\pi_2) = \int_{\bm{s}} \int_{\bm{\lambda}} V_{\Lambda_i}(\pi_1,\pi_2,\bm{s},\bm{\lambda}) \cdot \Pi_{j=1}^n [\mu^0_{\lambda_j}(d s_j) p_{\Lambda_j}(d \lambda_j)] 
\end{align*}

Then we have:
\begin{align*}
&V_{\Lambda_i}(\pi_1,\pi_2,\bm{s},\bm{\lambda})=  \int_{\bm{a}} 
\RR(z^{(i)},\bm{z^{(-i)}}) \pi_1(d a_i|s_i,\lambda_i) \Pi_{j\neq i} \pi_2(d a_j|s_j,\lambda_j)\\
&+ \zeta \int_{\bm{a}} \int_{\bm{s'}}\mathcal{T}(\bm{z},d\bm{s'}) V_{\Lambda_i}(\pi_1,\pi_2,\bm{s'},\bm{\lambda}) \pi_1(d a_i|s_i,\lambda_i) \Pi_{j\neq i} \pi_2(d a_j|s_j,\lambda_j)
\end{align*}

The goal is to compute $|V_{\Lambda_i}(\pi_1,\pi_2,\bm{s},\bm{\lambda})-V_{\Lambda_i}(\pi_3,\pi_4,\bm{s},\bm{\lambda})|$
and show that the latter is small provided that $\rho_{TV}((\pi_1,\pi_2),(\pi_3,\pi_4))$ is small. Let us use the notations:
\begin{align*}
c_1(\pi_1,\pi_2):=\int_{\bm{a}} \RR(z^{(i)},\bm{z^{(-i)}}) \pi_1(d a_i|s_i,\lambda_i) \Pi_{j\neq i} \pi_2(d a_j|s_j,\lambda_j)
\end{align*}
$$
c_2(\pi_1,\pi_2):=\int_{\bm{a}} \int_{\bm{s'}}\mathcal{T}(\bm{z},d\bm{s'}) V_{\Lambda_i}(\pi_1,\pi_2,\bm{s'},\bm{\lambda}) \pi_1(d a_i|s_i,\lambda_i) \Pi_{j\neq i} \pi_2(d a_j|s_j,\lambda_j)
$$
so that:
\begin{align*}
&V_{\Lambda_i}(\pi_1,\pi_2,\bm{s},\bm{\lambda})=  c_1(\pi_1,\pi_2) + \zeta c_2(\pi_1,\pi_2)
\end{align*}

By assumption $\RR(\cdot,\cdot)$ is bounded, say by $\RR_{max}$, hence we have:
\begin{align*}
& |c_1(\pi_1,\pi_2)-c_1(\pi_3,\pi_4)|\\
& \leq \RR_{max} \int_{\bm{a}} |\pi_1(d a_i|s_i,\lambda_i) \Pi_{j\neq i} \pi_2(d a_j|s_j,\lambda_j)-\pi_3(d a_i|s_i,\lambda_i) \Pi_{j\neq i} \pi_4(d a_j|s_j,\lambda_j)|\\
& \leq \RR_{max} \int_{a_i} |\pi_1(d a_i|s_i,\lambda_i) -\pi_3(d a_i|s_i,\lambda_i)|\\
& + \RR_{max} \int_{\bm{a_{-i}}} |\Pi_{j\neq i} \pi_2(d a_j|s_j,\lambda_j)-\Pi_{j\neq i} \pi_4(d a_j|s_j,\lambda_j)|\\
& \leq 2\RR_{max} \rho_{TV}(\pi_1,\pi_3) + 2\RR_{max}(n-1) \rho_{TV}(\pi_2,\pi_4) \leq 2n\RR_{max} \rho_{TV}((\pi_1,\pi_2),(\pi_3,\pi_4))
\end{align*}

For the term $|c_2(\pi_1,\pi_2)-c_2(\pi_3,\pi_4)|$, we can split:
\begin{align*}
V_{\Lambda_i}(\pi_1,\pi_2,\bm{s'},\bm{\lambda}) &\pi_1(d a_i|s_i,\lambda_i) \Pi_{j\neq i} \pi_2(d a_j|s_j,\lambda_j)\\
&-V_{\Lambda_i}(\pi_3,\pi_4,\bm{s'},\bm{\lambda}) \pi_3(d a_i|s_i,\lambda_i) \Pi_{j\neq i} \pi_4(d a_j|s_j,\lambda_j)\\
=V_{\Lambda_i}(\pi_1,\pi_2,\bm{s'},\bm{\lambda})&[ \pi_1(d a_i|s_i,\lambda_i) \Pi_{j\neq i} \pi_2(d a_j|s_j,\lambda_j)
-\pi_3(d a_i|s_i,\lambda_i) \Pi_{j\neq i} \pi_4(d a_j|s_j,\lambda_j)]\\
+\pi_3(d a_i|s_i,\lambda_i) \Pi_{j\neq i} &\pi_4(d a_j|s_j,\lambda_j) [V_{\Lambda_i}(\pi_1,\pi_2,\bm{s'},\bm{\lambda}) -V_{\Lambda_i}(\pi_3,\pi_4,\bm{s'},\bm{\lambda}) ]
\end{align*}

Since $V_{\Lambda_i}$ is bounded by $\RR_{max}(1-\zeta)^{-1}$, and noting that we have, as for $c_1$, that:
\begin{align*}
&\int_{\bm{a}} | \pi_1(d a_i|s_i,\lambda_i) \Pi_{j\neq i} \pi_2(d a_j|s_j,\lambda_j)
-\pi_3(d a_i|s_i,\lambda_i) \Pi_{j\neq i} \pi_4(d a_j|s_j,\lambda_j)| \\ &\leq  2n\rho_{TV}((\pi_1,\pi_2),(\pi_3,\pi_4))
\end{align*}
we then have:
\begin{align*}
|c_2(\pi_1,\pi_2)-c_2(\pi_3,\pi_4)| &\leq 2n\RR_{max}(1-\zeta)^{-1} \rho_{TV}((\pi_1,\pi_2),(\pi_3,\pi_4))\\
& + \max_{\bm{s},\bm{\lambda}} |V_{\Lambda_i}(\pi_1,\pi_2,\bm{s},\bm{\lambda})-V_{\Lambda_i}(\pi_3,\pi_4,\bm{s},\bm{\lambda})|
\end{align*}
We then have, collecting all terms together:
\begin{align*}
|V_{\Lambda_i}(\pi_1,\pi_2,\bm{s},\bm{\lambda})-V_{\Lambda_i}(\pi_3,\pi_4,\bm{s},\bm{\lambda})|&\leq 2n\RR_{max}(1+\zeta(1-\zeta)^{-1})\rho_{TV}((\pi_1,\pi_2),(\pi_3,\pi_4))\\
& + \zeta \max_{\bm{s},\bm{\lambda}} |V_{\Lambda_i}(\pi_1,\pi_2,\bm{s},\bm{\lambda})-V_{\Lambda_i}(\pi_3,\pi_4,\bm{s},\bm{\lambda})|
\end{align*}
Taking the maximum over $\bm{s},\bm{\lambda}$ on the left hand-side and rearranging terms finally yields:
\begin{align*}
&|V_{\Lambda_i}(\pi_1,\pi_2)-V_{\Lambda_i}(\pi_3,\pi_4)| \leq \max_{\bm{s},\bm{\lambda}} |V_{\Lambda_i}(\pi_1,\pi_2,\bm{s},\bm{\lambda})-V_{\Lambda_i}(\pi_3,\pi_4,\bm{s},\bm{\lambda})| \\ &\leq 2n(1-\zeta)^{-1} \RR_{max}(1+\zeta(1-\zeta)^{-1})\rho_{TV}((\pi_1,\pi_2),(\pi_3,\pi_4))\\
& = 2n(1-\zeta)^{-2} \RR_{max}\rho_{TV}((\pi_1,\pi_2),(\pi_3,\pi_4))
\end{align*}
which yields the desired continuity result. \qed

\textbf{Proof of proposition \ref{et2}.} Direction "$\Leftarrow$". Assume that there exists $\Phi$ and $\varphi$ as described, and assume $u(y,x)-u(x,x)>\epsilon$. Then $(x,y) \in S_0$ and therefore:
$$
\Phi(y) - \Phi(x) \geq \varphi(u(y,x)-u(x,x)) > \varphi(\epsilon) > 0,
$$
where we have used the fact that $\varphi$ is strictly increasing on $[0,+\infty)$ and positive on $(0,+\infty)$. Therefore, $u$ is extended transitive with $\delta_\epsilon:=\varphi(\epsilon)$.

Direction "$\Rightarrow$". Let: 
$$
S_\epsilon := \{(x,y) \in S: u(y,x)-u(x,x) > \epsilon\},\hspace{3mm} A_\epsilon:= \{\delta: \Phi(y)-\Phi(x) > \delta, \hspace{3mm} \forall (x,y) \in S_\epsilon\}.
$$

By extended transitivity assumption, $A_\epsilon$ contains at least an element $\delta_\epsilon>0$ for all $\epsilon>0$. Take $(\epsilon_n)_{n \in \mathbb{N}}$ to be the non-negative rational numbers, with $\epsilon_0:=0$ and $\delta_0:=0$. The sequence $(\delta_{\epsilon_n})_{n \geq 0}$ can be taken as non-decreasing, since for $n<m$ we have $S_{\epsilon_m} \subseteq S_{\epsilon_n}$ and therefore $A_{\epsilon_n} \subseteq A_{\epsilon_m}$. For example one can take $\delta_{\epsilon_n}= \sup A_{\epsilon_n}$, which is finite since $\Phi$ is bounded by assumption. We want to remove the duplicates to get a strictly increasing subsequence. For this, whenever $\delta_{\epsilon_{n_1}}=...=\delta_{\epsilon_{n_K}}$ for $n_1<...<n_K$, we retain only the last point $\delta_{\epsilon_{n_K}}$. This yields a strictly increasing subsequence that we denote by $(\delta_{\epsilon_{q_n}})_{n \geq 0}$. We use this approach so that for any duplicate $\delta_{\epsilon_{n_j}}$ that we remove, the point $(\epsilon_{n_j},\delta_{\epsilon_{n_j}})$ is above the line joining its two closest neighbors $(\epsilon_{q_{low}},\delta_{\epsilon_{q_{low}}})$, $(\epsilon_{q_{high}},\delta_{\epsilon_{q_{high}}})$ in the newly constructed subsequence. To conclude, we take $\widetilde{\varphi}$ to be the piecewise linear function joining the points $(\epsilon_{q_n}, \delta_{\epsilon_{q_n}})_{n \geq 0}$, and $\varphi(z):=\widetilde{\varphi}(\frac{1}{2}z)$. $\varphi$ is continuous and strictly increasing on $[0,+\infty)$, and $\varphi(0)=0$. Due to the approach we used to remove duplicates, we have $\delta_{\epsilon} \geq \widetilde{\varphi}(\epsilon)=\varphi(2\epsilon)$ for any positive rational number $\epsilon$, and actually for any positive real number $\epsilon$ using density of $\mathbb{Q}$ in $\mathbb{R}$. Therefore for any $(x,y) \in S_\epsilon$, $\Phi(y)-\Phi(x) \geq \delta_\epsilon \geq \varphi(2\epsilon)$. Finally, let $(x,y) \in S_0$, i.e. $u(y,x)-u(x,x)>0$. Let $\epsilon:=\frac{1}{2}(u(y,x)-u(x,x))>0$. Then $u(y,x)-u(x,x)> \epsilon$, and therefore:
$$
\Phi(y)-\Phi(x) \geq \varphi(2\epsilon) = \varphi(u(y,x)-u(x,x)). \hspace{5mm}\qed
$$

\textbf{Proof of theorem \ref{gradet}.} Remember that under the direct policy parametrization, each tuple $(s,a,\lambda)$ has its own parameter, $\pi_{\theta}(a|s,\lambda) := \theta_{s,\lambda, a}$, where $(\theta_{s,\lambda, a})_{a \in \As^{LP}}$ is in the probability simplex for every $s \in  \Ss^{LP}$, $\lambda \in \Ss^{\lambda^{LP}}$. For this reason, we will abuse notations and denote functions of $\pi_{\theta}$ in terms of $\theta$ only, for example $\widehat{V}(\theta,\theta)$ instead of $\widehat{V}(\pi_{\theta},\pi_{\theta})$. Throughout, $x$ and $y$ will denote real-valued vectors which value will change depending on the context. To make notations lighter, we will assume that by default, the initial state distribution is $\widetilde{\mu}^0$. When we will want to consider the distribution $\mu^0$, we will use explicit notations.

Recall that a function $f: \R^d \to \R$ is $\beta$-smooth if its gradient is Lipschitz, i.e. for all $x,y$:
$$
||\nabla f(x)-\nabla f(y)||_2 \leq \beta ||x-y||_2.
$$

In this case we have for all $x,y$ (\cite{convpg}, eq. (24)):
$$
|f(y)-f(x)-\nabla f(x) \cdot (y-x)| \leq \frac{\beta}{2} ||x-y||_2^2. 
$$

This yields in particular:
$$
f(y)-f(x) \geq - \frac{\beta}{2} ||x-y||_2^2 + \nabla f(x) \cdot (y-x).
$$

Consider in particular the case where $y=P_{\mathcal{X}^{LP}}[x + \alpha \nabla f(x)]$, that is $y$ is a PGA update to $x$. We want to show that:
\begin{align*}
    \begin{split}
        & \nabla f(x) \cdot (y-x) \geq \frac{1}{\alpha} (y-x) \cdot (y-x) \\
        \Leftrightarrow \hspace{4mm} & (y - x - \alpha \nabla f(x)) \cdot (y-x) \leq 0 \\
        \Leftrightarrow \hspace{4mm} & (P_{\mathcal{X}^{LP}}[z] - z) \cdot (P_{\mathcal{X}^{LP}}[z]-x) \leq 0, \hspace{4mm} z:=x + \alpha \nabla f(x).
    \end{split}
\end{align*}

But the latter is always true $\forall x, z$ by \cite{bubeck}, lemma 3.1. This yields the following inequality, always true when $f$ is $\beta$-smooth and $y=P_{\mathcal{X}^{LP}}[x + \alpha \nabla f(x)]$:
\begin{align}
\label{betass}
f(y)-f(x) \geq \left(\frac{1}{\alpha}-\frac{\beta}{2}\right) ||x-y||_2^2.
\end{align}

For fixed $x$, the function $y \to \widetilde{V}_{\Lambda^{LP}_i}(y,x)$ represents how the expected utility of an agent of supertype $i$ changes when all other agents remain fixed. This is equivalent to a single-agent problem, therefore we will be able to use the related results obtained in \cite{convpg}. We have by the latter (lemma D.3) that the single-agent value function $\widetilde{V}_{\Lambda^{LP}_i}$ is $\beta$-smooth for all agents $i$, where:
$$
\beta := \frac{2 \zeta |\As^{LP}| \RR_{max}}{(1-\zeta)^3}.
$$

But this implies that $\widehat{V}$ is also $\beta$-smooth, since:
\begin{equation*}
    \begin{split}
    ||\nabla_{1} \widehat{V}(y,x) - \nabla_{1} \widehat{V}(x,x)||_2 &= ||\frac{1}{n_{LP}} \sum_{i=1}^{n_{LP}} \nabla_{1} \widetilde{V}_{\Lambda^{LP}_i}(y,x) - \nabla_{1} \widetilde{V}_{\Lambda^{LP}_i}(x,x)||_2 \\
    & \leq \frac{1}{n_{LP}}  \sum_{i=1}^{n_{LP}} || \nabla_{1} \widetilde{V}_{\Lambda^{LP}_i}(y,x) - \nabla_{1} \widetilde{V}_{\Lambda^{LP}_i}(x,x)||_2\\
    & \leq \frac{1}{n_{LP}}  \sum_{i=1}^{n_{LP}} \beta ||x-y||_2 = \beta ||x-y||_2.
    \end{split}
\end{equation*}

Therefore, we have for one PGA iteration, using (\ref{betass}) with $\alpha<\frac{2}{\beta}$, so that $\frac{1}{\alpha}-\frac{\beta}{2} > 0$:
$$
 \widehat{V}(\theta_{n+1},\theta_n) - \widehat{V}(\theta_n,\theta_n) \geq \left(\frac{1}{\alpha}-\frac{\beta}{2}\right) ||\theta_{n+1}-\theta_n||_2^2.
$$

By extended transitivity assumption, let $\varphi$ the function in proposition \ref{et2}. Using the fact that $\varphi$ is strictly increasing on $[0,+\infty)$:
\begin{align}
    \label{p1}
    \Phi(\theta_{n+1})-\Phi(\theta_{n}) \geq \varphi( \widehat{V}(\theta_{n+1},\theta_n) - \widehat{V}(\theta_n,\theta_n)) \geq \varphi\left(\left(\frac{1}{\alpha}-\frac{\beta}{2}\right) ||\theta_{n+1}-\theta_n||_2^2\right).
\end{align}

Note that $diam(\Phi)<+\infty$ since $\Phi$ is bounded by extended transitivity assumption. Denote $n_* \in \argmin_{n \in [0,N-1]}||\theta_{n+1}-\theta_n||_2$. Taking the sum on both sides gives us:
$$
diam(\Phi) \geq \Phi(\theta_{N})-\Phi(\theta_{0}) \geq N \varphi\left(\left(\frac{1}{\alpha}-\frac{\beta}{2}\right) ||\theta_{n_*+1}-\theta_{n_*}||_2^2 \right).
$$

We get, using the fact that $\varphi$ is invertible and strictly increasing, that $||\theta_{n_*+1}-\theta_{n_*}||_2$ is small as $N \to \infty$:
$$
||\theta_{n_*+1}-\theta_{n_*}||_2 \leq \sqrt{\frac{ \alpha}{1-\frac{1}{2} \alpha \beta} \varphi^{-1}\left(\frac{diam(\Phi)}{N}\right)}.
$$

The latter tells us that for iteration $n_*$, the PGA update is small, therefore we can also conclude that the true gradient is small by (\cite{convpg}, Proposition B.1):
\begin{equation*}
\begin{split}
\max_{\substack{||\delta||_2 \leq 1\\\theta_{n^*} + \delta \in \mathcal{X}^{LP}}}\delta^T \nabla_{1} \widehat{V}(\theta_{n_*+1},\theta_{n_*}) 
\leq \frac{1+ \alpha \beta}{\alpha} ||\theta_{n_*+1}-\theta_{n_*}||_2 
\leq  \frac{1+ \alpha \beta}{\sqrt{\alpha \left(1-\frac{1}{2} \alpha \beta \right)}} \sqrt{ \varphi^{-1}\left(\frac{diam(\Phi)}{N}\right)}.
\end{split}
\end{equation*}

Since $\widehat{V}$ is $\beta$-smooth, this also implies:
\begin{equation*}
    \begin{split}
        \max_{\substack{||\delta||_2 \leq 1\\\theta_{n^*} + \delta \in \mathcal{X}^{LP}}}\delta^T \nabla_{1} \widehat{V}(\theta_{n_*},\theta_{n_*}) &\leq \max_{\substack{||\delta||_2 \leq 1\\\theta_{n^*} + \delta \in \mathcal{X}^{LP}}}\delta^T \nabla_{1} \widehat{V}(\theta_{n_*+1},\theta_{n_*}) + \beta ||\theta_{n_*+1}-\theta_{n_*}||_2 \\
        &\leq \frac{1+ 2\alpha \beta}{\sqrt{\alpha \left(1-\frac{1}{2} \alpha \beta \right)}} \sqrt{ \varphi^{-1}\left(\frac{diam(\Phi)}{N}\right)}.
    \end{split}
\end{equation*}

We conclude by a gradient domination argument. Remember that until now, we only considered initial state distribution $\widetilde{\mu}^0$. \cite{convpg}, Lemma 4.1, gives us the following result in the single agent case (see also \cite{markovpg}). It states that the best possible gain by a given agent using initial state distribution $\mu^0$, when others remain fixed, is bounded by its value function gradient under $\widetilde{\mu}^0$, namely for every $\theta$:
$$
\widetilde{V}_{\Lambda^{LP}_i,\mu^0}(\theta,\theta_{n_*}) - \widetilde{V}_{\Lambda^{LP}_i,\mu^0}(\theta_{n_*},\theta_{n_*}) \leq \mathcal{D}_{n_*,i}(\mu^0, \widetilde{\mu}^0) \max_{z \in \mathcal{X}^{LP}} (z-\theta_{n_*})^T \nabla_{1} \widetilde{V}_{\Lambda^{LP}_i, \widetilde{\mu}^0}(\theta_{n_*},\theta_{n_*}).
$$
where, using the notations in (\ref{dmis}):
\begin{align*}
\mathcal{D}_{n_*,i}(\mu^0, \widetilde{\mu}^0) := \sup_{\pi' \in BR_{\mu^0,i}(\pi_{\theta_{n_*}})} \max_{s \in \Ss^{LP}} \frac{d^{\pi'}_{\mu^0,i}(s)}{d^{\pi_{\theta_{n_*}}}_{\widetilde{\mu}^0,i}(s)}.
\end{align*}

However as we will see below, we can always choose $N$ large enough so that $\pi_{\theta_{n_*}}$ is an $\epsilon$-shared equilibrium, hence:
$$
\mathcal{D}_{n_*,i}(\mu^0, \widetilde{\mu}^0) \leq \mathcal{D}_\epsilon(\mu^0, \widetilde{\mu}^0).
$$

Taking the average over agents $i$, we get for every $\theta$:
\begin{align*}
\begin{split}
\widehat{V}_{\mu^0}(\theta,\theta_{n_*}) - \widehat{V}_{\mu^0}(\theta_{n_*},\theta_{n_*}) &\leq \mathcal{D}_\epsilon(\mu^0, \widetilde{\mu}^0) \max_{z \in \mathcal{X}^{LP}} (z-\theta_{n_*})^T \nabla_{1} \widehat{V}_{\widetilde{\mu}^0}(\theta_{n_*},\theta_{n_*})\\
& \leq 2 \mathcal{D}_\epsilon(\mu^0, \widetilde{\mu}^0) \sqrt{|\Ss^{LP}||\Ss^{\lambda^{LP}}|} \max_{\substack{||\delta||_2 \leq 1\\\theta_{n^*} + \delta \in \mathcal{X}^{LP}}} \delta^T \nabla_{1} \widehat{V}_{\widetilde{\mu}^0}(\theta_{n_*},\theta_{n_*})\\
& \leq 2 \mathcal{D}_\epsilon(\mu^0, \widetilde{\mu}^0) \sqrt{|\Ss^{LP}||\Ss^{\lambda^{LP}}|} \frac{1+ 2\alpha \beta}{\sqrt{\alpha \left(1-\frac{1}{2} \alpha \beta \right)}} \sqrt{ \varphi^{-1}\left(\frac{diam(\Phi)}{N}\right)}.
\end{split}
\end{align*}

We have used the following basic fact: for $x \in \mathcal{X}^{LP}$, the entries of $x$ are in the probability simplex for each $s, \lambda$, and therefore:
$$
||x||_2 = \sqrt{\sum_{a,s,\lambda} x_{a,s,\lambda}^2} \leq \sqrt{\sum_{a,s,\lambda} x_{a,s,\lambda}} = \sqrt{|\Ss^{LP}||\Ss^{\lambda^{LP}}|}.
$$
Hence for $x,y \in \mathcal{X}^{LP}$, $||x-y||_2 \leq 2 \sqrt{|\Ss^{LP}||\Ss^{\lambda^{LP}}|}$. 

Finally, we conclude by observing that to get an $\epsilon-$shared equilibrium, we need:
$$
 2 \mathcal{D}_\epsilon(\mu^0, \widetilde{\mu}^0) \sqrt{|\Ss^{LP}||\Ss^{\lambda^{LP}}|} \frac{1+ 2\alpha \beta}{\sqrt{\alpha \left(1-\frac{1}{2} \alpha \beta \right)}} \sqrt{ \varphi^{-1}\left(\frac{diam(\Phi)}{N}\right)} \leq \epsilon,
$$

that is:
$$
N \geq \frac{diam(\Phi)}{\varphi \left(\frac{\epsilon^2  }{4 \mathcal{D}_\epsilon(\mu^0, \widetilde{\mu}^0)^2|\Ss^{LP}||\Ss^{\lambda^{LP}}|} \cdot \frac{\alpha (2 -\alpha \beta)}{2(1+2\alpha \beta)^2} \right)}.
$$

This yields that for any $\theta$, we have:
$$
\widehat{V}_{\mu^0}(\theta,\theta_{n_*}) - \widehat{V}_{\mu^0}(\theta_{n_*},\theta_{n_*}) \leq \epsilon,
$$

which by definition means that $\pi_{\theta_{n_*}}$ is an $\epsilon-$shared equilibrium. 

Finally, denote the "projected gradient" $G_n:=\alpha^{-1}_n(\theta_{n+1} - \theta_n)$, which would be equal to the true gradient were there no projection, and where we have allowed the learning rate to be time-dependent (but still upper bounded by $\frac{2}{\beta}$). Equation (\ref{p1}) and finiteness of $diam(\Phi)$ give us that the series $\sum \varphi\left(\left(1-\frac{\beta \alpha_n}{2}\right) \alpha_n ||G_n||_2^2 \right)$ is (absolutely) convergent. This yields in particular that $\varphi\left(\left(1-\frac{\beta \alpha_n}{2}\right) \alpha_n ||G_n||_2^2 \right) \to 0$. i.e. $\left(1-\frac{\beta \alpha_n}{2}\right) \alpha_n ||G_n||_2^2 \to \varphi^{-1}(0)=0$ by continuity and invertibility of $\varphi$. If we assume that $\sum \alpha_n = +\infty$, we get that $\lim_{n \to +\infty} ||G_n||_2 = 0$. Using the "gradient domination" argument above, this means that for any $\epsilon$, there exists a $N_\epsilon$ such that for all $n \geq N_\epsilon$, $\theta_n$ is an $\epsilon$-shared equilibrium. In particular, every converging subsequence $(\theta_{m_n})$ converges to a shared equilibrium, using a compactness argument as in the proof of theorem \ref{n11}, as $\theta \in [0,1]^{m}$ and $\widehat{V}$ is continuous. If there exists a unique shared equilibrium, $\theta_n$ converges to that equilibrium, since all converging subsequences have the same limit point.
\qed

\textbf{Proof of lemma \ref{lemsto}}. Denote $\nabla := \nabla f(x)$, and the noisy gradient $\widetilde{\nabla} := \nabla + \nabla \odot \phi$, so that $y=x + \alpha \widetilde{\nabla}$. By $\beta$-smoothness, as in the proof of theorem \ref{gradet}, we have:
\begin{align}
\begin{split}
\label{eql1}
&f(y,x)-f(x,x) \geq - \frac{\beta}{2} ||y-x||_2^2 + \nabla \cdot (y-x)\\
\Leftrightarrow & \frac{1}{\alpha} \left(f(y,x)-f(x,x) \right) \geq \left(1-\frac{\alpha \beta}{2}\right) ||\widetilde{\nabla}||_2^2 - (\nabla \odot \phi) \cdot \widetilde{\nabla}.
\end{split}
\end{align}
 
Contrary to the non-stochastic case, the increment $f(y,x)-f(x,x)$ is not necessarily non-negative due to the stochastic term $\phi$. We would like somehow to get back to this case. We have:
\begin{equation}
\begin{split}
& \left(1-\frac{\alpha \beta}{2}\right) ||\widetilde{\nabla}||_2^2 - (\nabla \odot \phi) \cdot \widetilde{\nabla} \\
& = \left(1-\frac{\alpha \beta}{2}\right) (||\nabla||_2^2 + ||\nabla \odot \phi||_2^2 + 2 (\nabla \odot \phi) \cdot \nabla) - (\nabla \odot \phi) \cdot \widetilde{\nabla}  \\
& = \left(1-\frac{\alpha \beta}{2}\right) ||\nabla||_2^2 -\frac{\alpha \beta}{2} ||\nabla \odot \phi||_2^2  +   (1 - \alpha \beta) (\nabla \odot \phi) \cdot \nabla \\
& =  \sum_{i=1}^{d} \nabla^2_{i} \left[1-\frac{\alpha \beta}{2} -\frac{\alpha \beta}{2} \phi_{i}^{2}  +   (1 - \alpha \beta) \phi_{i} \right] =: \sum_{i=1}^{d} Y_i.
\end{split}
\end{equation}

If we assume that $-1< \phi_{min} \leq \phi_{i} \leq \phi_{max}$ a.s., then we can get the very simple result that $Y_i$ is non-negative provided:
$$
\alpha < \frac{2}{\beta} \frac{1+\phi_{i}}{1+2\phi_{i}+\phi_{i}^2} = \frac{2}{\beta(1+\phi_{i})},
$$
which is always true if $\alpha < \frac{2}{\beta(1+\phi_{max})}$. It is immediate to show that if $\alpha=\frac{2 (1-\eta)}{\beta(1+\phi_{max})}$ for $\eta \in [0,1]$, $1-\frac{\alpha \beta}{2} -\frac{\alpha \beta}{2} \phi_{i}^{2}  +   (1 - \alpha \beta) \phi_{i} \geq (1+\phi_{min}) \eta$. Now let us look at the unbounded case. Since $\phi$ is a vector of independent noise, we are computing the sum of independent random variables $Y_i$, so we can use the Berry-Esseen theorem. We have for $\alpha = \frac{2(1-\eta)}{\beta(1+\sigma_2)}$ and $\eta \in [0,1]$:
$$
\EE[Y_i] = \nabla^2_{i} \left[1-\frac{\alpha \beta}{2} -\frac{\alpha \beta}{2} \sigma_{2} \right]= \eta \nabla^2_{i},
$$
$$
Var[Y_i] = \nabla^4_{i} \left[ (1 - \alpha \beta)^2\sigma_{2} +\frac{\alpha^2 \beta^2}{4} (\sigma_{4}-\sigma^{2}_2) \right] =:  \nabla^4_{i} \kappa_{4}
$$
$$
\EE[|Y_i-\EE[Y_i]|^3] = \nabla^6_{i} \kappa_{6}.
$$

A tedious but straightforward calculation gives an upper bound for $\kappa_{6}$. If $\mathcal{N}$ is the standard normal c.d.f., the Berry-Esseen theorem states that \cite{tyurin}:
$$
\PP[\sum_{i=1}^{d} Y_i > x] \geq 1 - \mathcal{N}\left(\frac{x-\mu_Y}{\sigma_Y}\right) - \frac{0.6}{\sigma_Y} \max_{i\in[1,d]} \frac{\rho_i}{Var[Y_i]}
$$
$$
\mu_Y := \sum_{i=1}^d \EE[Y_i], \hspace{5mm} \sigma^2_Y:= \sum_{i=1}^d Var[Y_i], \hspace{5mm} \rho_i := \EE[|Y_i - \EE[Y_i]|^3].
$$

In our case, we have $\mu_Y = \eta ||\nabla||^2_2$, $\sigma^2_Y = \kappa_{4} \sum_{i=1}^{d} \nabla^4_{i}$ and $\frac{\rho_i}{Var[Y_i]} = \frac{\kappa_{6}}{\kappa_{4}} \nabla^2_{i}$. Therefore:
$$
\max_{i\in[1,d]} \frac{\rho_i}{\sigma_Y Var[Y_i]} = \frac{\kappa_{6}}{\kappa_{4}^{3/2}} \frac{\max_{i\in[1,d]}\nabla^2_{i}}{\sqrt{\sum_{i=1}^{d} \nabla^4_{i}}} = \frac{\kappa_{6}}{\kappa_{4}^{3/2}} \mathcal{E}_d.
$$

This yields, taking $x= \eta' ||\nabla||_2^2$ for $\eta' \in [0,1]$:
$$
\PP[\sum_{i=1}^{d} Y_i > \eta' ||\nabla||_2^2] \geq 1 - \mathcal{N}\left(\frac{\eta'-\eta}{\sqrt{\kappa_{4}}} \cdot \widehat{\mathcal{E}}_d \sqrt{d} \right) - \frac{0.6 \kappa_{6}}{\kappa_{4}^{3/2}} \mathcal{E}_d.
$$
\qed


\textbf{Proof of theorems \ref{gradet2}, \ref{gradet3}.} We use the same notational conventions as in the proof of theorem \ref{gradet}. Let $K:=|\Ss^{LP}||\Ss^{\lambda^{LP}}||\As^{LP}|$ the total number of coordinates. Remember that for a policy $\pi$, we define:
$$
\ln \pi := \frac{1}{K} \sum_{s,a,\lambda} \ln \pi(a|s,\lambda).
$$

Denote $L(\pi_{\theta},\pi_{\theta_n}):=\widehat{V}(\pi_{\theta},\pi_{\theta_n}) + \nu \ln \frac{\pi_{\theta}}{\pi_{\theta_n}}$ the log-barrier regularized objective. Consider the following update rule:
$$
\theta_{n+1} = \theta_n + \alpha \widetilde{\nabla}_n,
$$

where we denote the true gradient $\nabla_n := \nabla_{1} L(\pi_{\theta_n},\pi_{\theta_n})= \nabla_{1} \widehat{V}(\pi_{\theta_n},\pi_{\theta_n}) + \nu \nabla \ln \pi_{\theta_n}$, and the noisy gradient $\widetilde{\nabla}_n := \nabla_n + \nabla_n \odot \phi_n$. 

We showed in the proof of theorem \ref{gradet} that $\widehat{V}$ is $\beta$-smooth provided each one of the single-agent $\widetilde{V}_{\Lambda^{LP}_i}$ is. By \cite{convpg} lemma D.4, $\theta \to \widetilde{V}_{\Lambda^{LP}_i}(\pi_{\theta},\pi_{\theta_n}) + \nu \ln \pi_{\theta}$ is $\beta_\nu$-smooth with:
$$
\beta_\nu := \frac{8 \RR_{max}}{(1-\zeta)^3} + \frac{2 \nu \RR_{max}}{|\Ss^{LP}||\Ss^{\lambda^{LP}}|}.
$$

Therefore $L$ is $\beta_\nu$-smooth. By lemma \ref{lemsto}, we have for $\eta,\eta'\in [0,1]$ and $\alpha = \frac{2(1-\eta)}{\beta_\nu(1+\sigma_2)}$:
$$
\PP[L(\pi_{\theta_{n+1}},\pi_{\theta_n}) - L(\pi_{\theta_{n}},\pi_{\theta_n})\geq \eta'||\nabla_n||_2^2] \geq 1 - \mathcal{N}\left(\frac{\eta'-\eta}{\sqrt{\kappa_{4}}} \cdot \widehat{\mathcal{E}}_{n,K} \sqrt{K}\right) - \frac{0.6 \kappa_{6}}{\kappa_{4}^{3/2}}\mathcal{E}_{n,K},
$$
where $\mathcal{E}_{n,K}$, $\widehat{\mathcal{E}}_{n,K}$ are respectively the maximum and average dispersion associated to $\nabla_n$. Call $\mathcal{E}_{K} := \sup_{n \in [0,N_{tot}]} \mathcal{E}_{n,K}$, $\widehat{\mathcal{E}}_{K} := \inf_{n \in [0,N_{tot}]} \widehat{\mathcal{E}}_{n,K}$, so that the above probability is at least:
$$
p_{\eta,\eta',K} := 1 - \mathcal{N}\left(\frac{\eta'-\eta}{\sqrt{\kappa_{4}}} \cdot \widehat{\mathcal{E}}_{K} \sqrt{K}\right) - \frac{0.6 \kappa_{6}}{\kappa_{4}^{3/2}}\mathcal{E}_{K}.
$$

Let us work on the event $\Omega_n := \{L(\pi_{\theta_{n+1}},\pi_{\theta_n}) - L(\pi_{\theta_{n}},\pi_{\theta_n})\geq \eta'||\nabla_n||_2^2\}$. If the entropy of $\nabla_n$ is high, namely its coordinates are roughly of the same magnitude, $\widehat{\mathcal{E}}_{n,K} \approx 1$ and $\mathcal{E}_{n,K} \approx K^{-1/2}$, which means that the probability of the latter event is very close to one if $K$ is large and $\eta'<\eta$, which we are free to choose. The event $\Omega_n$ has probability at least $p_{\eta,\eta',K}$. Call $\widetilde{\Omega}_N$ the event where we achieve at least $N$ successful transitions among the total budget $N_{tot}$. $\widetilde{\Omega}_N$ has probability at least:
$$
\sum_{i=N}^{N_{tot}} \binom{N_{tot}}{i} p_{\eta,\eta',K}^{i} (1-p_{\eta,\eta',K})^{N_{tot}-i}.
$$
By assumption, the game $L$ is extended transitive. Call $\varphi_\nu$ the function in proposition \ref{et2}. We can proceed as in the proof of theorem \ref{gradet} and get on $\widetilde{\Omega}_N$, where $n_*$ is the integer in $[0,N_{tot}]$ that achieves the minimum positive update $||\theta_{n+1}-\theta_n||_2$:
$$
||\nabla_{n_*}||_2 \leq \sqrt{\frac{1}{\eta'}\varphi_\nu^{-1}\left(\frac{diam(\Phi_\nu)}{N}\right)}.
$$

By \cite{convpg} theorem 5.2:
\begin{align*}
&\mbox{if } \hspace{5mm} \sqrt{\frac{1}{\eta'}\varphi_\nu^{-1}\left(\frac{diam(\Phi_\nu)}{N}\right)} \leq \frac{\nu}{2K}\\
&\mbox{then } \hspace{5mm} \widetilde{V}_{\Lambda^{LP}_i,\mu^0}(\theta,\theta_{n_*}) - \widetilde{V}_{\Lambda^{LP}_i,\mu^0}(\theta_{n_*},\theta_{n_*}) \leq \frac{2\nu}{1-\zeta} \mathcal{D}_{n_*,i}(\mu^0, \widetilde{\mu}^0) \leq \frac{2\nu}{1-\zeta} \mathcal{D}_\epsilon(\mu^0, \widetilde{\mu}^0),
\end{align*}

where the last inequality uses the same reasoning as in the proof of theorem \ref{gradet} that we show below that we can always choose $N$ large enough so that $\pi_{\theta_{n_*}}$ is an $\epsilon$-shared equilibrium.

Taking the average over agents $i$, we get for every $\theta$:
\begin{align*}
\begin{split}
\widehat{V}_{\mu^0}(\theta,\theta_{n_*}) - \widehat{V}_{\mu^0}(\theta_{n_*},\theta_{n_*}) \leq \frac{2\nu}{1-\zeta} \mathcal{D}_\epsilon(\mu^0, \widetilde{\mu}^0).
\end{split}
\end{align*}

We need to take $\epsilon$ such that:
$$
\frac{2\nu}{1-\zeta} \mathcal{D}_\epsilon(\mu^0, \widetilde{\mu}^0) \leq \epsilon\\
\Leftrightarrow \nu \leq \frac{(1-\zeta) \epsilon}{2\mathcal{D}_\epsilon(\mu^0, \widetilde{\mu}^0)}.
$$
 Now pick $N$ such that:
$$
N \geq \frac{diam(\Phi)}{\varphi \left(\frac{\nu^2 \eta'}{4K^2} \right)} \geq \frac{diam(\Phi)}{\varphi \left(\frac{(1-\zeta)^2 \epsilon^2 \eta'}{16 K^2 \mathcal{D}_\epsilon(\mu^0, \widetilde{\mu}^0)^2} \right)}.
$$

In the extreme case where the coordinates of the gradient are the same, we have:
$$
p_{\eta,\eta',K} = 1 - \mathcal{N}\left(\frac{\eta'-\eta}{\sqrt{\kappa_{4}}} \sqrt{K}\right) - \frac{0.6 \kappa_{6}}{\sqrt{K}\kappa_{4}^{3/2}} \stackrel{K \to +\infty}{\to} 1.
$$

The case where $\phi$ is uniformly bounded is a direct consequence of lemma \ref{lemsto}: by the latter, we know $Y_i \geq (1+\phi_{min})\eta \nabla_{n,i}^2$ almost surely for $\alpha = \frac{2 (1-\eta)}{\beta_\nu(1+\phi_{max})}$. We can simply replace $\eta'$ by $(1+\phi_{min})\eta$ in the above expression for $N$.

\textbf{Proof of theorem \ref{alrt}.} The proof has been discussed in section \ref{seccalib}. The goal is to apply \cite{borkar}, theorem 1.1 (see also \cite{leslie}, theorem 5). The proof of theorem \ref{gradet2} gives us that $\nabla_1 \widehat{V}$, and hence $\widehat{x}$, is Lipschitz. Assumption \ref{alr5} gives that $\widehat{y}$ is Lipschitz. Theorem \ref{gradet2} together with assumptions \ref{alr3} and \ref{alr}, yield that the ODE $\frac{d\theta}{dt} = \widehat{x}(\theta, \theta^\Lambda)$ has a unique global asymptotically stable equilibrium $f_{eq}(\theta^\Lambda)$, which is further Lipschitz in $\theta^\Lambda$. By extended transitivity in assumption \ref{alr6}, together with assumption \ref{alr5} which ensures that the mapping $\theta^{\Lambda} \to V^{calib}(\pi_{\theta}, \pi^\Lambda_{\theta^{\Lambda}})$ is $\beta$-smooth, we can show as in the proof of theorem \ref{gradet2} that $\Phi_\nu$ is a Lyapunov function, since (up to noise terms that are dealt with as in the proof of theorem \ref{gradet2}):
$$
\Phi_\nu(\theta^\Lambda_{n+1})-\Phi_\nu(\theta^\Lambda_{n}) \geq \varphi_\nu \left(\left(\frac{1}{\alpha_n^{cal}} - \frac{\beta_\nu}{2}\right)||\theta^\Lambda_{n+1}-\theta^\Lambda_{n}||_2^2 \right).
$$
We conclude using \cite{leslie}, theorem 5 and proposition 4.

\clearpage

\begin{figure}[ht]
    \centering
    \begin{subfigure}[b]{\fsizeee \textwidth}
        \centering
        \includegraphics[width=\textwidth]{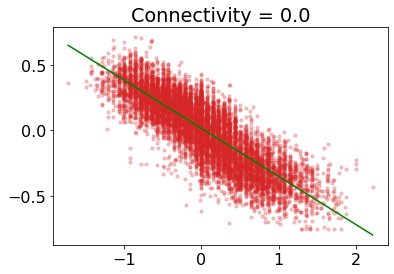}
   \end{subfigure}
    \hfill
    \begin{subfigure}[b]{\fsizeee \textwidth}
        \centering
        \includegraphics[width=\textwidth]{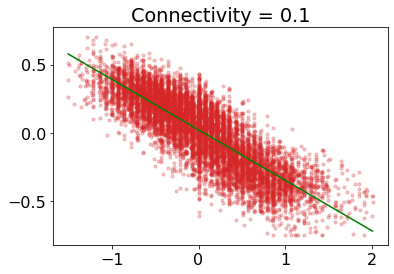}
    \end{subfigure}
    \begin{subfigure}[b]{\fsizeee \textwidth}
        \centering
        \includegraphics[width=\textwidth]{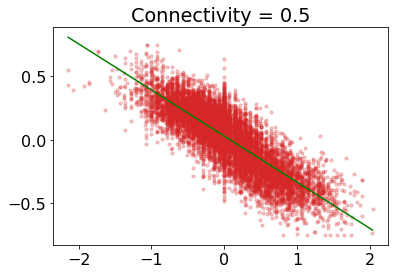}
    \end{subfigure}
    \begin{subfigure}[b]{\fsizeee \textwidth}
        \centering
        \includegraphics[width=\textwidth]{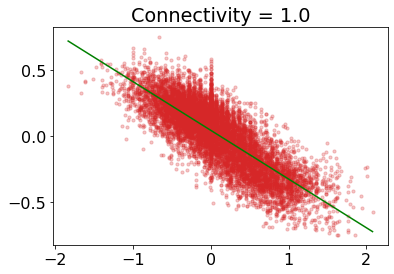}
    \end{subfigure}
    \caption{Skew $\epsilon_{t,\text{skew}}$ (y-axis) vs. inventory (x-axis) for various PnL LT connectivity values.}
    \label{figg12}
\end{figure}

\begin{figure}[ht]
    \centering
    \begin{subfigure}[b]{\fsizeee \textwidth}
        \centering
        \includegraphics[width=\textwidth]{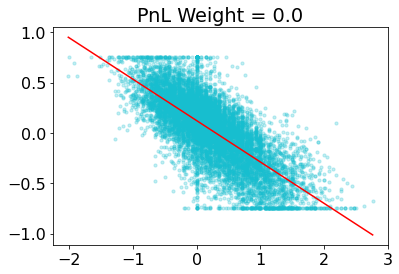}
   \end{subfigure}
    \hfill
    \begin{subfigure}[b]{\fsizeee \textwidth}
        \centering
        \includegraphics[width=\textwidth]{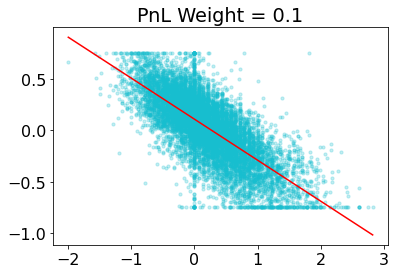}
    \end{subfigure}
    \begin{subfigure}[b]{\fsizeee \textwidth}
        \centering
        \includegraphics[width=\textwidth]{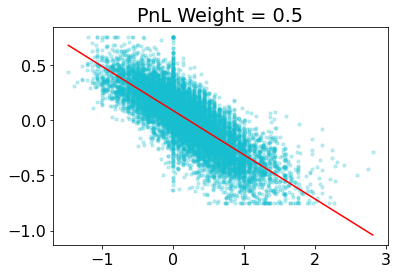}
    \end{subfigure}
    \begin{subfigure}[b]{\fsizeee \textwidth}
        \centering
        \includegraphics[width=\textwidth]{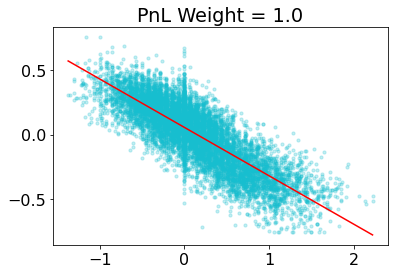}
    \end{subfigure}
    \caption{Skew $\epsilon_{t,\text{skew}}$ (y-axis) vs. inventory (x-axis) for various PnL weight values.}
    \label{figg13}
\end{figure}

\begin{figure}[ht]
    \centering
    \begin{subfigure}[b]{\fsizeee \textwidth}
        \centering
        \includegraphics[width=\textwidth]{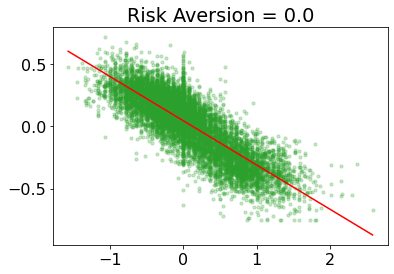}
   \end{subfigure}
    \hfill
    \begin{subfigure}[b]{\fsizeee \textwidth}
        \centering
        \includegraphics[width=\textwidth]{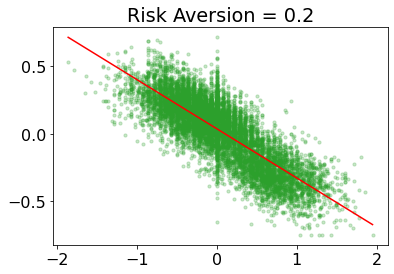}
    \end{subfigure}
    \begin{subfigure}[b]{\fsizeee \textwidth}
        \centering
        \includegraphics[width=\textwidth]{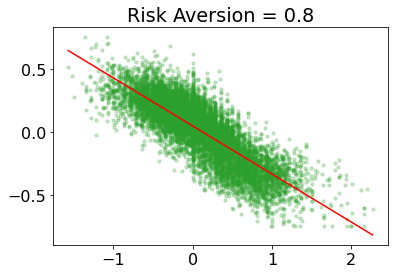}
    \end{subfigure}
    \begin{subfigure}[b]{\fsizeee \textwidth}
        \centering
        \includegraphics[width=\textwidth]{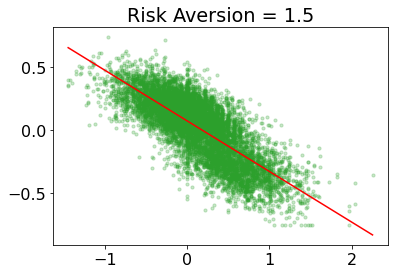}
    \end{subfigure}
    \caption{Skew $\epsilon_{t,\text{skew}}$ (y-axis) vs. inventory (x-axis) for various risk aversion values.}
    \label{figg14}
\end{figure}

\begin{figure}[ht]
    \centering
    \begin{subfigure}[b]{ \fsizee \textwidth}
        \centering
        \includegraphics[width=\textwidth]{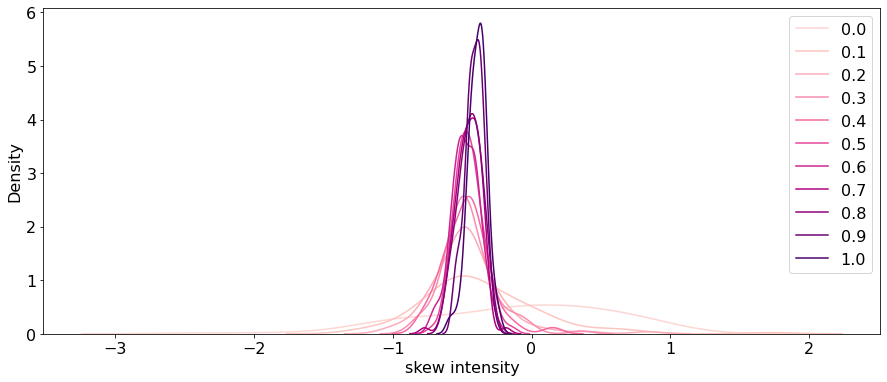}
    \end{subfigure}
    \begin{subfigure}[b]{\fsizee\textwidth}
        \centering
        \includegraphics[width=\textwidth]{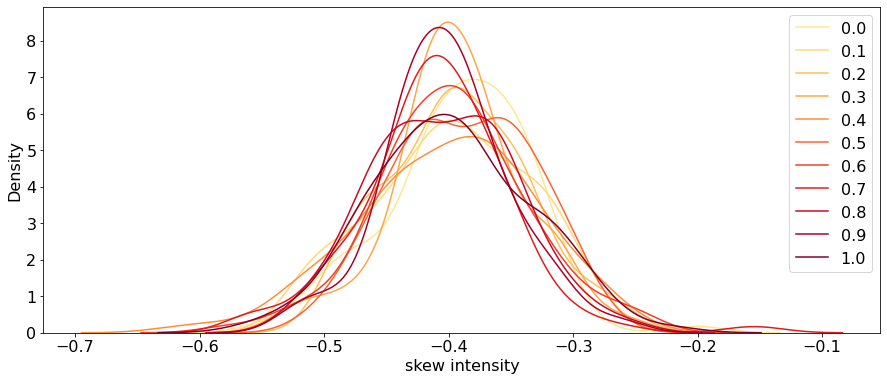}
    \end{subfigure}
    \begin{subfigure}[b]{\fsizee\textwidth}
        \centering
        \includegraphics[width=\textwidth]{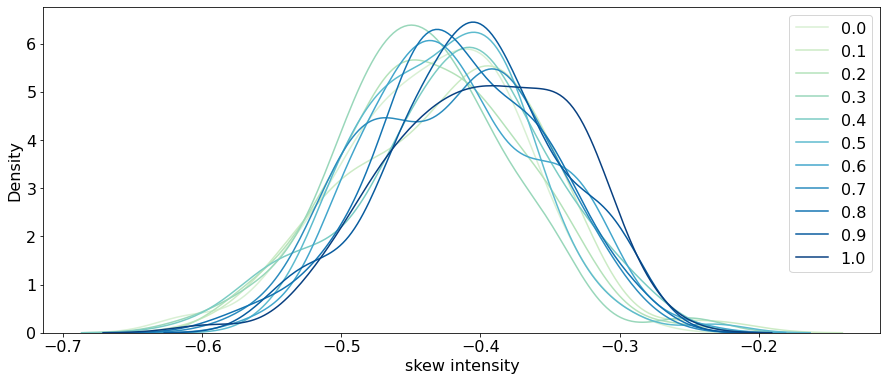}
    \end{subfigure}
    \begin{subfigure}[b]{\fsizee\textwidth}
        \centering
        \includegraphics[width=\textwidth]{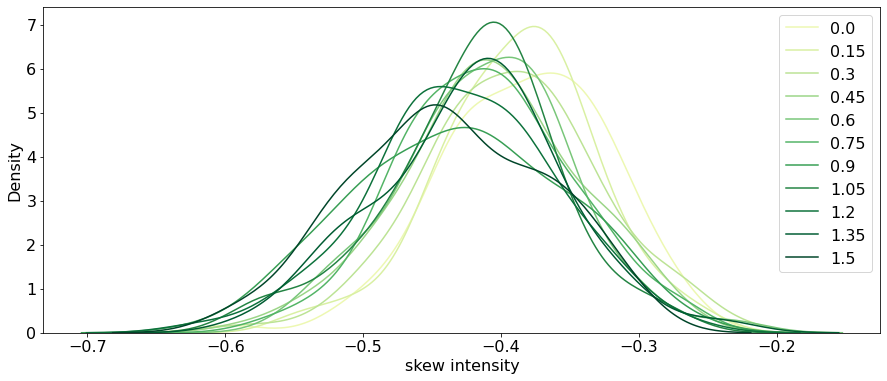}
    \end{subfigure}
    \caption{Skew intensity distribution for different values of flow LT connectivity, PnL LT connectivity, PnL weight, risk aversion (from top to bottom).}
    \label{figg9}
\end{figure}

\begin{figure}[ht]
  \centering
  \centerline{\includegraphics[width=\textwidth]{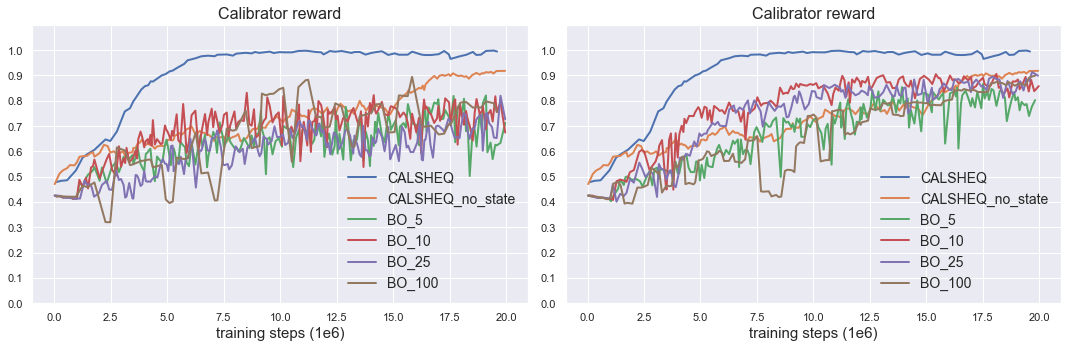}}
  \caption{Calibrator cumulative reward during training for various number of Bayesian Optimization (BO) frequencies $M$ - Experiment 1 - \textbf{(Left)} BO Expected Improvement (EI) - \textbf{(Right)} BO UCB with exploration parameter $\kappa=1.5$.}
  \label{f0}
\end{figure}

\begin{figure}[ht]
  \centering
  \centerline{\includegraphics[width=\columnwidth]{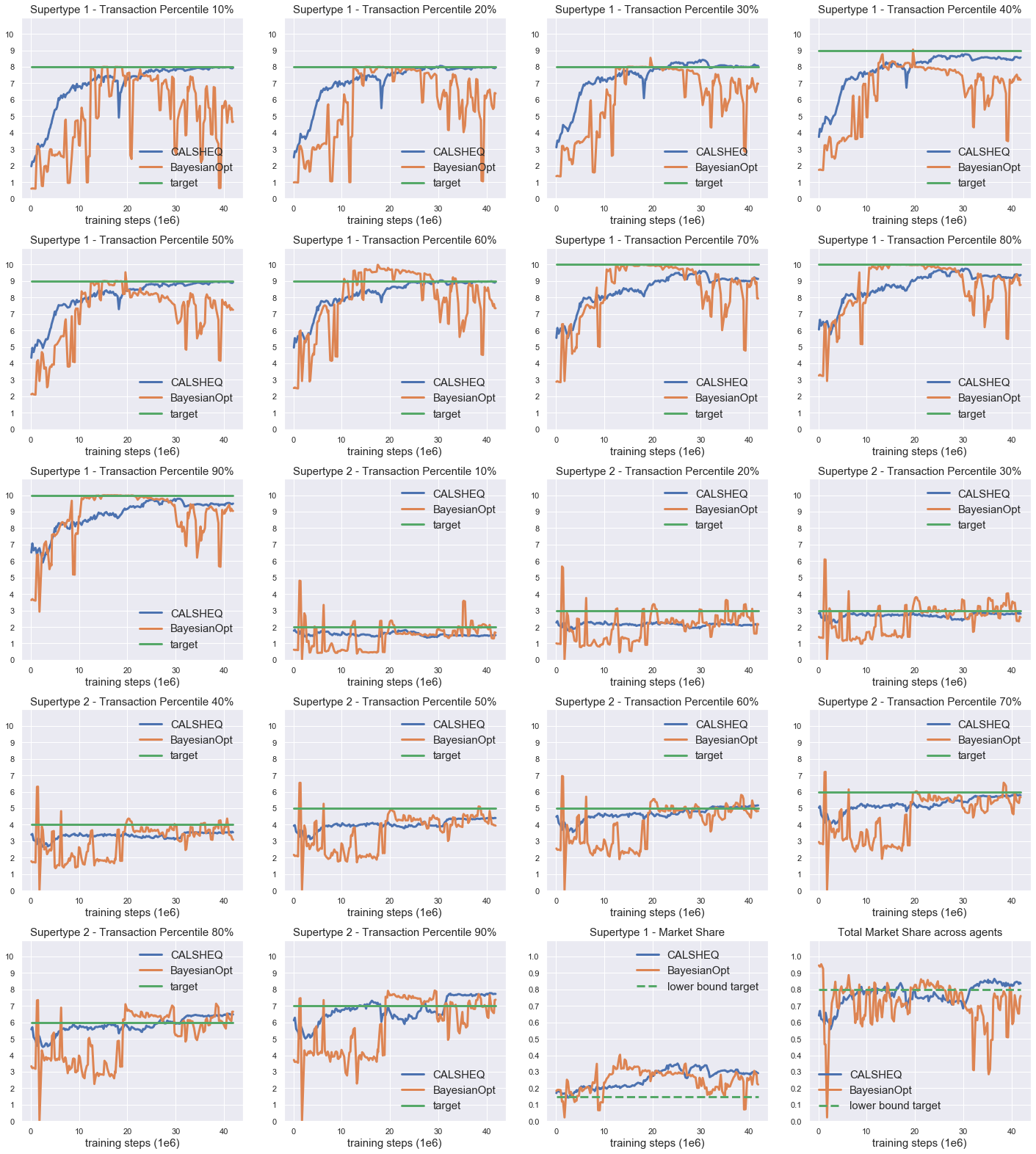}}
  \caption{Calibration experiment 2 - Calibration target fit for trade quantity distribution percentile and market share during training, averaged over episodes $B$. Dashed line target indicates that the constraint was set to be greater than target (not equal to it). \textit{CALSHEQ} (ours) and baseline (Bayesian optimization).}
  \label{f8}
\end{figure}

\begin{figure}[ht]
  \centering
  \centerline{\includegraphics[width=\columnwidth]{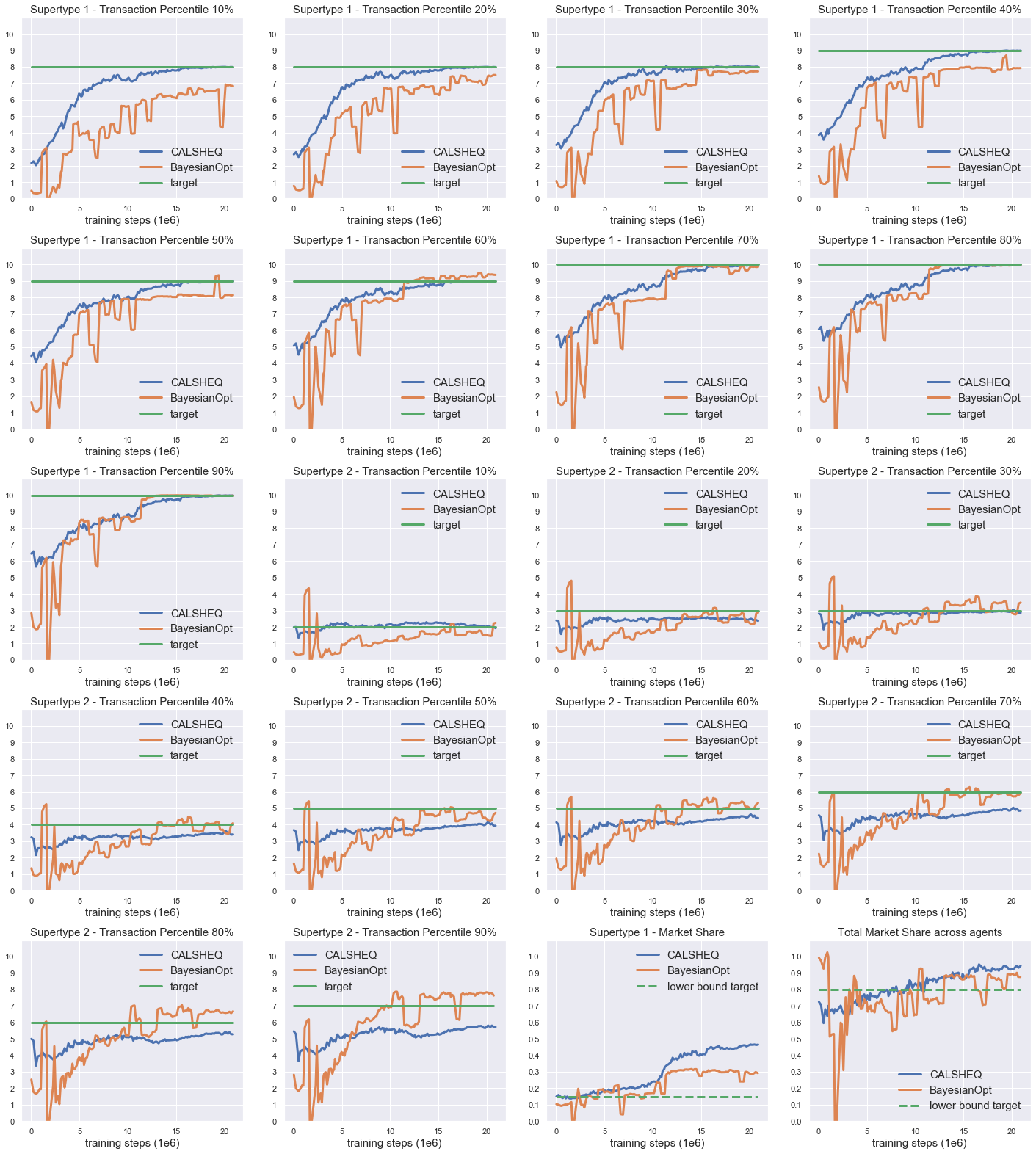}}
  \caption{Calibration experiment 3 - Calibration target fit for trade quantity distribution percentile and market share during training, averaged over episodes $B$. Dashed line target indicates that the constraint was set to be greater than target (not equal to it). \textit{CALSHEQ} (ours) and baseline (Bayesian optimization).}
  \label{f9}
\end{figure}

\begin{figure}[ht]
  \centering
  \centerline{\includegraphics[scale=0.35]{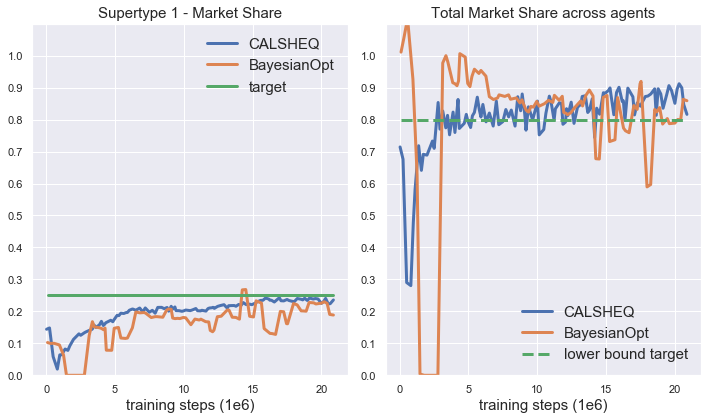}}
  \caption{Calibration experiment 4 - Calibration target fit for market share during training, averaged over episodes $B$. Dashed line target indicates that the constraint was set to be greater than target (not equal to it). \textit{CALSHEQ} (ours) and baseline (Bayesian optimization).}
  \label{f10}
\end{figure}

\begin{figure}[ht]
  \centering
  \centerline{\includegraphics[width=\columnwidth]{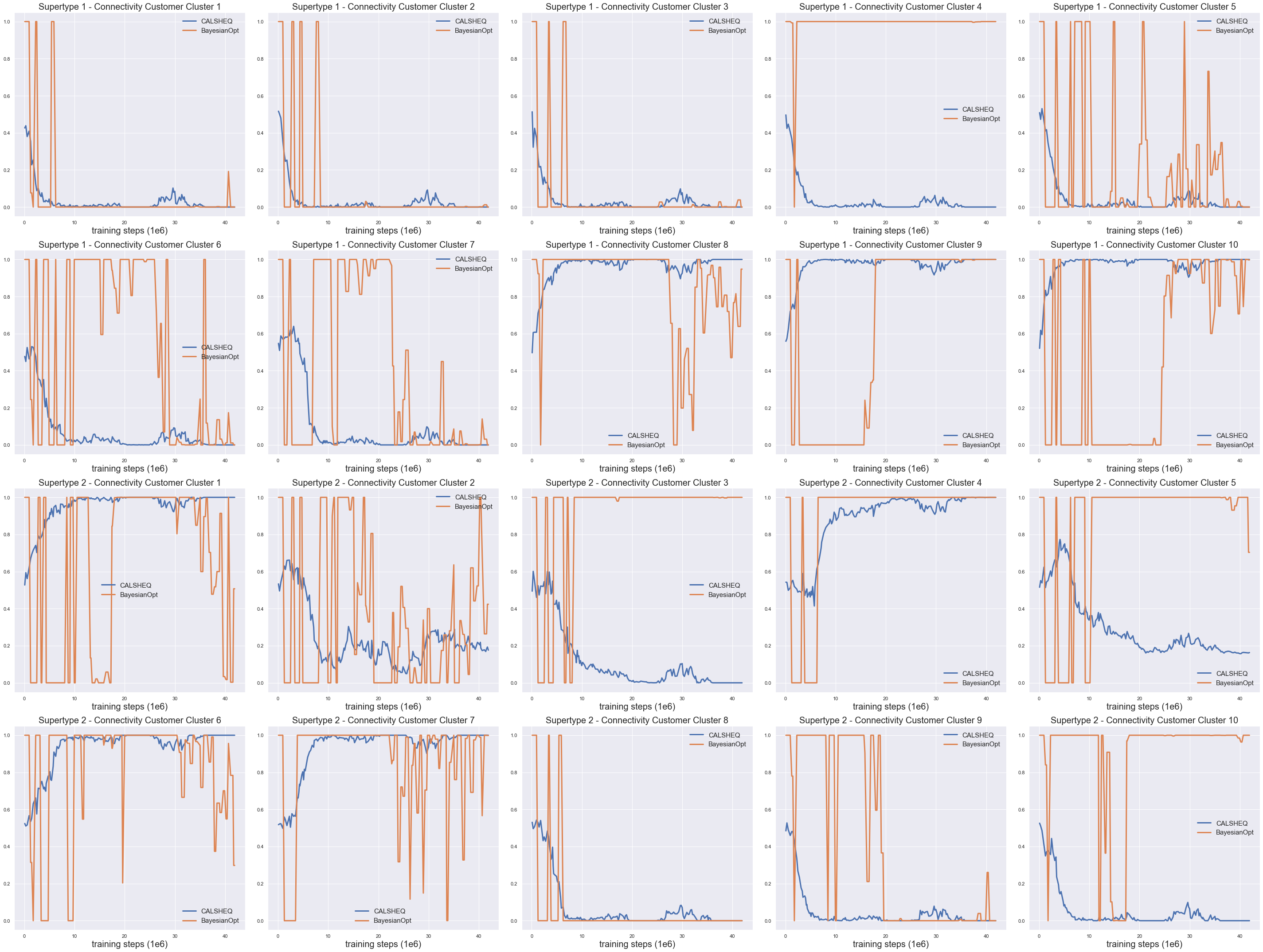}}
  \caption{Calibration experiment 1 - calibrated parameter values as a function of time, averaged over episodes $B$. \textit{CALSHEQ} (ours) and baseline (Bayesian optimization). CALSHEQ varies the parameters smoothly, contrary to Bayesian optimization.}
  \label{f2}
\end{figure}

\begin{figure}[ht]
  \centering
  \centerline{\includegraphics[width=\columnwidth]{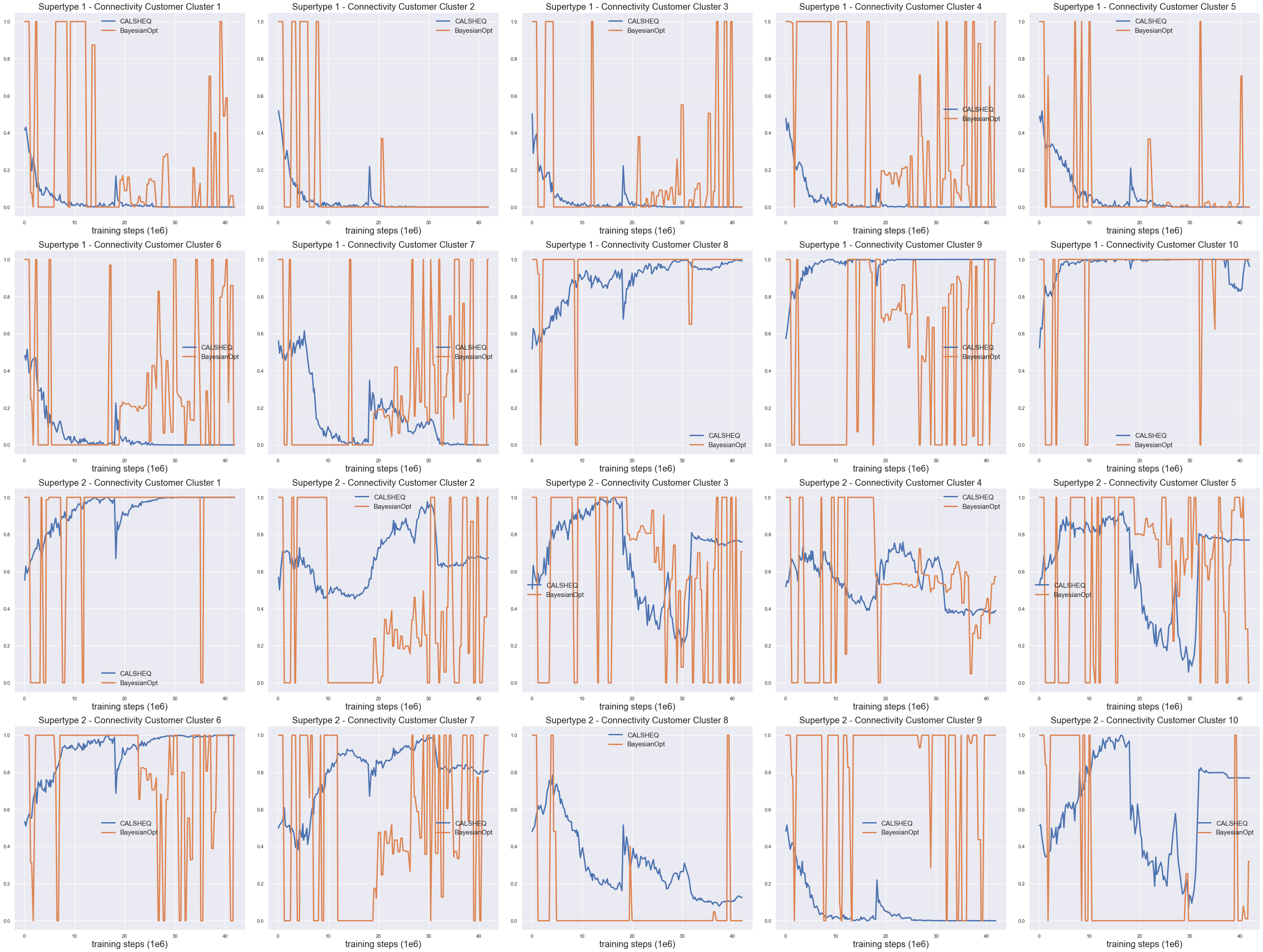}}
  \caption{Calibration experiment 2 - calibrated parameter values as a function of time, averaged over episodes $B$. \textit{CALSHEQ} (ours) and baseline (Bayesian optimization). CALSHEQ varies the parameters smoothly, contrary to Bayesian optimization.}
  \label{f3}
\end{figure}

\begin{figure}[ht]
  \centering
  \centerline{\includegraphics[width=\columnwidth]{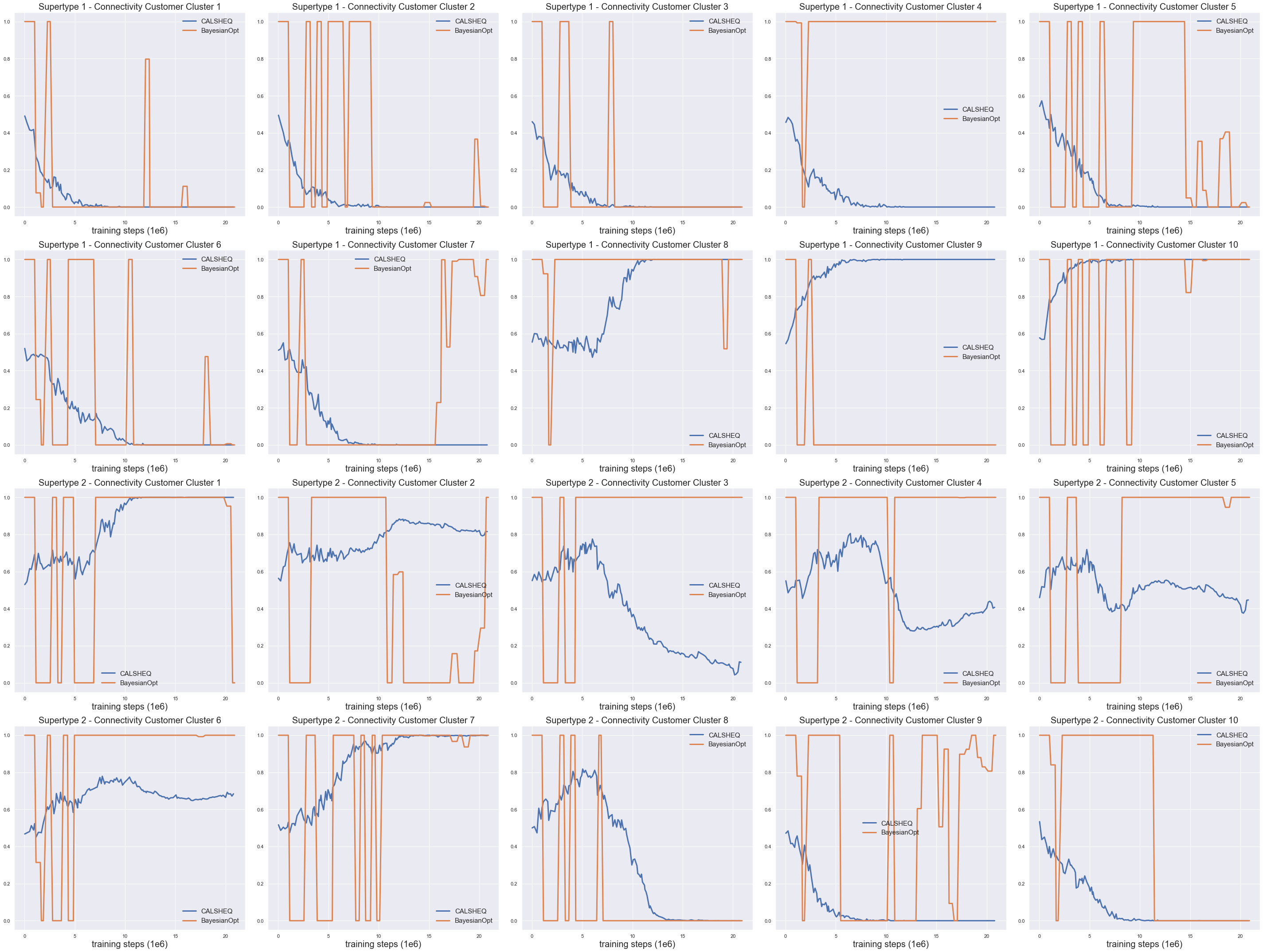}}
  \caption{Calibration experiment 3 - calibrated parameter values as a function of time, averaged over episodes $B$. \textit{CALSHEQ} (ours) and baseline (Bayesian optimization). CALSHEQ varies the parameters smoothly, contrary to Bayesian optimization.}
  \label{f4}
\end{figure}

\begin{figure}[ht]
  \centering
  \centerline{\includegraphics[width=\columnwidth]{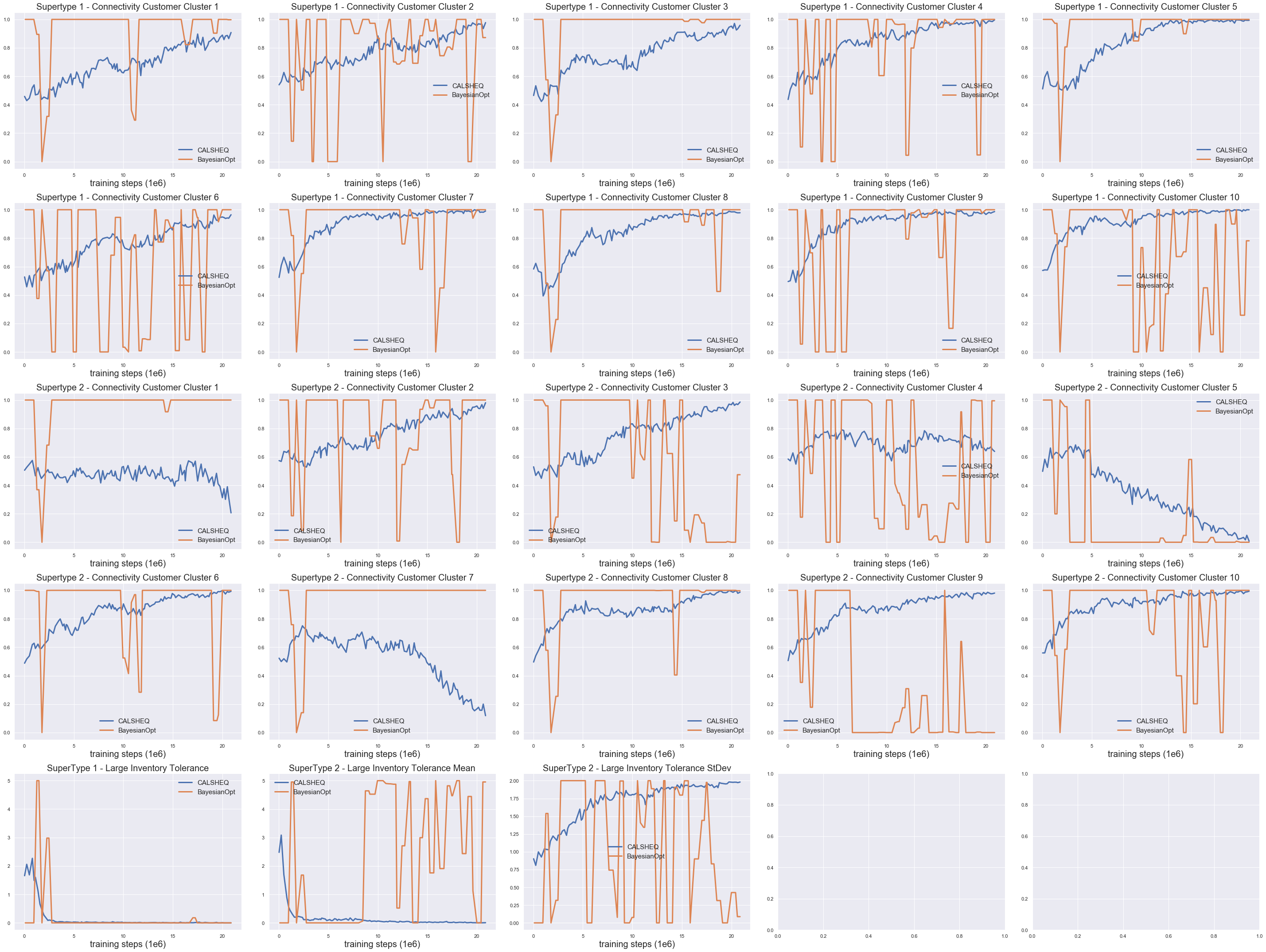}}
  \caption{Calibration experiment 4 - calibrated parameter values as a function of time, averaged over episodes $B$. \textit{CALSHEQ} (ours) and baseline (Bayesian optimization). CALSHEQ varies the parameters smoothly, contrary to Bayesian optimization.}
  \label{f5}
\end{figure}

\end{document}